\documentclass[12pt]{article}
\usepackage[margin=0.8in]{geometry}
\usepackage{amssymb,amsmath,amsthm}
\usepackage{amsfonts,bbm}
\usepackage{blindtext}
\usepackage{booktabs}
\usepackage{bm}
\usepackage{caption}
\usepackage{comment}
\usepackage[shortlabels]{enumitem}
\usepackage{float}
\usepackage{graphicx}
\usepackage{bbm}
\graphicspath{ {../figures/} }
\usepackage{indentfirst}
\usepackage{mathtools}
\usepackage{mathrsfs}
\usepackage{multirow}
\usepackage{upgreek}
\usepackage[utf8]{inputenc}
\usepackage[backend=biber,style=apa,natbib=true,doi=false,url=true,uniquelist=false,uniquename=false,sorting=nyt]{biblatex}

\renewbibmacro{in:}{}
\DeclareNameAlias{default}{last-first/first-last}

\usepackage{array}
\newcolumntype{L}[1]{>{\raggedright\arraybackslash}m{#1}}
\usepackage{rotating}
\usepackage{ragged2e}
\usepackage{setspace}
\usepackage{subcaption}
\usepackage{tikz}
\usepackage{ctable}
\usepackage[colorlinks,citecolor=blue,urlcolor=blue,linkcolor=blue]{hyperref}
\usepackage{cleveref}
\crefname{enumi}{}{} 
\usepackage{dcolumn}
\usepackage{float}
\usepackage{pdflscape}
\usepackage{longtable}
\usepackage{threeparttable}

\newcommand{\tablenote}[1]{%
\begin{justify}
{\footnotesize \textit{Notes:} #1}
\end{justify}
}
\newcommand{\figurenote}[1]{%
\begin{justify}
{\footnotesize \textit{Notes:} #1}
\end{justify}
}
\makeatletter

\newcommand{\Rmnum}[1]{\expandafter\@slowromancap\romannumeral #1@}
\newcommand{\indicator}[1]{\mathbbm{1}\{#1\}}
\newcommand\independent{\protect\mathpalette{\protect\independenT}{\perp}}
\def\independenT#1#2{\mathrel{\rlap{$#1#2$}\mkern2mu{#1#2}}}

\newcommand{\E}{\mathbb{E}}
\newcommand{\V}{\mathbb{V}}
\newcommand{\Prob}{\mathbb{P}}
\newcommand{\G}{\mathcal{G}}
\newcommand{\Perm}[2]{\mathcal{P}(#1,#2)}
\newcommand{\bigO}{\mathcal O}
\newcommand{\bigOp}{\mathcal O_p}

\makeatother
\newtheorem{theorem}{Theorem}

\newtheorem{example}{Example}[section]
\newtheorem{lemma}{Lemma}[section]

\newtheorem{definition}{Definition}[section]
\newtheorem{assumption}{Assumption}

\newtheorem{remark}{Remark}[section]
\newtheorem{proposition}{Proposition}[section]

\newtheorem{inneruassumption}{Assumption}
\newenvironment{namedassumption}[2][]
  {\renewcommand\theinneruassumption{#2}\begin{inneruassumption}[#1]}
  {\end{inneruassumption}}

\newtheorem{innernamedtheorem}{Theorem}
\newenvironment{namedtheorem}[2][]
  {\renewcommand\theinnernamedtheorem{#2}\begin{innernamedtheorem}[#1]}
  {\end{innernamedtheorem}}

\crefname{figure}{Figure}{Figures}
\crefname{theorem}{Theorem}{Theorems}
\crefname{innernamedtheorem}{Theorem}{Theorems}
\crefname{assumption}{Assumption}{Assumptions}
\crefname{inneruassumption}{Assumption}{Assumptions}
\crefname{proposition}{Proposition}{Propositions}
\crefname{condition}{Condition}{Conditions}
\crefname{lemma}{Lemma}{Lemmata}
\crefname{example}{Example}{Examples}
\crefname{remark}{Remark}{Remarks}
\numberwithin{lemma}{section}
\numberwithin{result}{section}
\numberwithin{definition}{section}
\numberwithin{remark}{section}
\numberwithin{table}{section}
\numberwithin{proposition}{section}
\numberwithin{condition}{section}
\numberwithin{equation}{section}

\newfloat{sfigure}{tbp}{sfig}
\DeclarePairedDelimiter{\ceil}{\lceil}{\rceil}
\DeclarePairedDelimiter{\floor}{\lfloor}{\rfloor}
\allowdisplaybreaks[4]

\begin{document}
\hypersetup{pageanchor=false}
\pagenumbering{gobble}
	\title{Limit Theory for U-Statistics under Clustered and Weakly Dependent Data}
	\date{\today}
	\author{Emmanuel Selorm Tsyawo\footnote{email: estsyawo@gmail.com, Department of Economics, Finance and Legal Studies, Culverhouse College of Business, University of Alabama}}
	\maketitle
	\thispagestyle{empty}

\begin{abstract}
\noindent This paper develops asymptotic theory and feasible inference for unbounded-kernel order-$k$ $U$-statistics under clustered sampling and weakly dependent time-series sampling. The analysis first builds the complete order-$2$ pipeline, moving from clustered data to exact $m$-dependence and then to near-epoch dependence. The same logic is subsequently extended to general order $k\geq2$. Under clustered sampling, the theory allows arbitrary within-cluster dependence and growing, unbalanced cluster sizes. Under weak dependence, an i.i.d.-based approximating sequence carries the exact-$m$ theory to near-epoch-dependent processes. The common combinatorial device partitions the sample into columns, separating sampling-generic tuples, where the first-order Hoeffding projection is analysed, from collision terms and higher-order Hoeffding projection remainders, which are controlled explicitly. Cluster-robust and HAC estimators of the covariance of the first-order projection, needed for feasible inference, are shown to be consistent. Empirical applications and data-calibrated simulations for inequality, L-moment, and rank-dependence statistics illustrate the finite-sample performance of the proposed procedures.

\vspace{.5cm}


\vspace{.5cm}

		
\end{abstract}

\begin{refsection}

\newpage
\hypersetup{pageanchor=true}
\pagenumbering{arabic}
\section{Introduction}\label{Sect:Introduction}
$U$-statistics arise naturally in economics and finance whenever the object of interest is built from pairwise or higher-order comparisons rather than from one-observation moments. Direct economic examples include inequality, distribution-shape, and rank-dependence measures such as the Gini mean difference, L-moments, Kendall's tau, and Spearman's rho, e.g., \citet{bhattacharya2007inference,darku2020gini,hosking1990lmoments,dehling2017testing}. $U$-statistics also appear more broadly in score and criterion functions associated with pairwise differencing estimators, e.g., \citet{honore1994pairwise,jochmans2013pairwise}, integrated conditional moment estimators and \(\chi^2\) specification tests, e.g., \citet{dominguez2004consistent,escanciano2018simple,jiang-tsyawo-2026consistent}, and gravity models, e.g. \citet{jochmans-2017-two_way_gravity,yang-zhang-2023-three_way_gravity}. Under dependence, these objects are often easy to motivate substantively but harder to analyse, because their summands overlap and the usual independent-tuple geometry breaks down. The difficulty is therefore not only point estimation, but also feasible covariance estimation for the first-order projection, whether the $U$-statistic is itself the target of inference or enters via an estimation or specification testing framework.

The results developed here treat these objects under two dependence structures usually handled separately: clustered sampling and weakly dependent time series. The unifying combinatorial device partitions the sample into columns. In the clustered design, columns are formed so that no column contains two observations from the same cluster; in the fixed-$m$ time-series design, the analogue is the set of residue-class columns modulo $(m+1)$, whose observations are more than $m$ periods apart. The leading term of an order-$k$ $U$-statistic then comes from sampling-generic tuples: cluster-generic tuples in the clustered design and lag-generic tuples in the fixed-$m$ design, where ordinary Hoeffding projection logic applies. The remaining tuples are collision terms, either same-cluster or short-lag. This makes the argument modular; column partitioning standardises the treatment of collision terms and higher-order Hoeffding projection terms, while the remaining first-order projection is handled by the limit theory appropriate to the sampling scheme.

\citet{hansen2019asymptotic} gives the clustered-sum central limit theorem that underlies the linear projection step, allowing arbitrary within-cluster dependence and heterogeneous, unbounded cluster growth. The clustered results below extend that input to $U$-statistics by isolating same-cluster collision terms, extracting the Hoeffding linear term, and controlling the higher-order Hoeffding projection remainders. This complements bounded-cluster two-sample $U$-statistic results such as \citet{Lee-Dehling-2005-generalized}, where the residual terms are lower order under uniformly bounded cluster sizes. It also sits alongside recent clustered rank-statistic work, such as the two-sample Mood statistic of \citet{suzuki2025mood}, which allows cluster-size heterogeneity in a two-sample setting.

For weakly dependent time series, the main nearby results cover related but distinct settings. \citet{sen1963properties} studies $U$-statistics for stationary $m$-dependent processes by separating non-serial tuples from tuples containing short-lag interactions and reducing the limit theory to the first-order projection. \citet{malevich_abdalimov_1982} and \citet{janson2023asymptotic} also formulate their exact-$m$ $U$-statistic limit theory under stationarity. The exact-$m$ theory in this paper allows heterogeneous, non-stationary sequences by keeping the lag-generic first projection indexed and using a blocking construction together with the heterogeneous-cluster CLT of \citet{hansen2019asymptotic}. \citet{fischer2016multivariate} proves general-order $U$-statistic central limit theorems for strongly mixing sequences with bounded kernels satisfying an extended variation condition, and leaves the near-epoch-dependent multivariate extension as a conjectural direction. \citet{fischer2017robust} develops a general NED theory for multivariate kernels and Generalised Linear (GL) statistics, including an unbounded-kernel CLT and a HAC variance estimator in the bounded-kernel case. \citet{dehling2017testing}, working with P-near-epoch dependence on an absolutely regular process, treats the order-2 Kendall setting and uses empirical Hoeffding projections for HAC inference. The weak-dependence results in this paper work within the present $L_2$-NED framework, with expected-H\"older or data-dependent-modulus perturbation conditions used to transfer finite-window approximation to the $U$-statistic. Relative to NED-on-mixing frameworks of, e.g., \citet{dehling2017testing}, the current NED formulation is more specialised: the NED condition is imposed on an i.i.d. innovation base, so finite-window approximants are exactly $m$-dependent. This restriction is used deliberately, because it lets the exact-$m$ order-$k$ decomposition and remainder bounds carry the NED CLT and the empirical-projection HAC consistency argument.

The paper makes three contributions. First, it identifies a common combinatorial structure behind clustered sampling and finite-range temporal dependence, where, after partitioning the sample into columns, sampling-generic tuples are built from mutually independent units across the relevant vertical partitions, while same-cluster and lag collisions are counted and controlled. Second, it uses that structure to separate the common remainder argument from the sampling-specific linear-projection step, yielding order-$k$ $U$-statistic laws of large numbers and central limit theorems under arbitrary within-cluster dependence with growing, heterogeneous cluster sizes, and under heterogeneous, non-stationary exact $m$-dependence, before carrying the weak-dependence theory to near-epoch dependence through an i.i.d.-based approximating sequence. Third, it pairs that limit theory with feasible inference by providing a cluster-robust plug-in estimator in the clustered case and a HAC estimator based on the empirical first-order projection in the time-series case.

The remainder of the paper is organised as follows. \Cref{Sect:Order2Pipeline} develops the order-$2$ results, while \Cref{Sect:OrderKTheory} extends them to general order $k$. \Cref{Sect:FeasibleInference} develops the corresponding covariance estimators. \Cref{Sect:Empirical} presents the empirical applications to which simulations in \Cref{Sect:Simulations} are calibrated. Finally, \Cref{Sect:Conclusion} concludes. The proofs of all theoretical results are collected in the appendix.

\paragraph{Notation:} For \(p\geq1\), \(\|X\|_p:=\{\E[\|X\|^p]\}^{1/p}\), where the inner norm is Euclidean; \([N]:=\{1,\ldots,N\}\); and \(\V[\cdot]\) denotes variance. The set \(\G_g\) collects units in cluster \(g\in[G]\), \(n_g:=\#\G_g\), and \(g(i)\) denotes the cluster of unit \(i\). All random variables are defined on an underlying probability space \((\Omega,\mathcal A,\Prob)\); the state space is standard Borel, kernels are measurable, and conditional expectations are taken as fixed measurable versions. Constants denoted by \(C\) may change across occurrences.

\section{Order-\texorpdfstring{$2$}{2} Theory}\label{Sect:Order2Pipeline}
This section develops the full argument for order-$2$ $U$-statistics. It begins with clustered sampling, where partitioning the sample into columns makes same-cluster collision terms visible, and then gives the exact-$m$ time-series analogue, where residue-class columns play the same role. The near-epoch-dependent extension is added next, by approximating the observed process with exact finite-range versions.

\subsection{Clustered sampling}\label{Sect:Cluster}

For a sampled clustered array \(\{W_i\}_{i\in\G_g}\), \(g\in[G]\), the following sampling assumption is imposed.
\begin{assumption}[Sampling]\label{ass:sampling}
	The arrays \(\{W_i\}_{i\in\G_g}\), \(g\in[G]\), are independently distributed across clusters.
\end{assumption}
\noindent \Cref{ass:sampling} leaves within-cluster dependence unrestricted and imposes independence only across clusters. Cluster sizes may be equal or unequal. The framework therefore covers short balanced panels, unbalanced panels, clustered data, cross-sectional data, and repeated cross-sections with arbitrary within-cluster dependence.

The order-$2$ clustered case is the simplest place to see the full mechanism. For intuition, sort the clusters in decreasing order of size so that the number of columns equals the largest cluster size. Suppose the ordered cluster sizes are $4,3,2,1$, and write $W_{g,a}$ for the $a$-th observation in cluster $g$. In the following display, box styles identify clusters and braces identify the vertical partitions:
\begingroup
\tikzset{
    clusterbox/.style={draw, rounded corners=2pt, minimum width=1.35cm,
        minimum height=0.48cm, inner sep=1pt, align=center},
    clusterone/.style={clusterbox, fill=gray!8, draw=black, text=black, line width=0.45pt},
    clusterthree/.style={clusterbox, fill=gray!22, draw=black, text=blue!45!black, line width=0.45pt},
    clustertwo/.style={clusterbox, fill=white, draw=black, text=red!45!black, dashed, line width=0.55pt},
    clusterfour/.style={clusterbox, fill=gray!35, draw=black, text=green!35!black, dotted, line width=0.65pt}
}
\newcommand{\clusterentry}[2]{\tikz[baseline=(x.base)] \node[#1] (x) {$#2$};}
\[
\begin{array}{cccc}
\overbrace{
\begin{array}{c}
\clusterentry{clusterone}{W_{1,1}}\\[2pt]
\clusterentry{clustertwo}{W_{2,1}}\\[2pt]
\clusterentry{clusterthree}{W_{3,2}}
\end{array}}^{\mathcal M_n(1)}
&
\overbrace{
\begin{array}{c}
\clusterentry{clusterone}{W_{1,2}}\\[2pt]
\clusterentry{clustertwo}{W_{2,2}}\\[2pt]
\clusterentry{clusterfour}{W_{4,1}}
\end{array}}^{\mathcal M_n(2)}
&
\overbrace{
\begin{array}{c}
\clusterentry{clusterone}{W_{1,3}}\\[2pt]
\clusterentry{clustertwo}{W_{2,3}}\\[2pt]
\varnothing
\end{array}}^{\mathcal M_n(3)}
&
\overbrace{
\begin{array}{c}
\clusterentry{clusterone}{W_{1,4}}\\[2pt]
\clusterentry{clusterthree}{W_{3,1}}\\[2pt]
\varnothing
\end{array}}^{\mathcal M_n(4)}
\end{array}.
\]
\endgroup
The useful property is columnwise: each brace covers at most one observation from any cluster, even with unequal cluster sizes. Thus observations read down a column are cross-cluster and independent under \Cref{ass:sampling}. Same-cluster pairs must occur across columns and are counted as same-cluster collision terms. The appendix formalises this construction in the unified linked-tuple notation used for the proofs.

The next assumption, which is \citet[Assumption 2]{hansen2019asymptotic}, regulates the heterogeneity in cluster sizes.
\begin{assumption}[Cluster size heterogeneity]\label{ass:clus_het} For some \(2\leq r<\infty\),
	\[
	\frac{\Big( \sum_{g=1}^G n_g^r \Big)^{2/r}}{n} \leq C < \infty, \quad \text{ and } \quad \max_{g\in [G]}\frac{n_g^2}{n} \rightarrow 0 \quad \text{as} \quad n\rightarrow \infty.
	\]
\end{assumption}
\noindent Let \(\displaystyle m_n:=\max_{g\in[G]}n_g\) denote the maximum cluster size; then \Cref{ass:clus_het} implies \(m_n^2/n\to0\), so \(m_n/\sqrt n\to0\) and, \emph{a fortiori}, \(m_n/n\to0\).

Denote the U-statistic by
\[
    U_{n,2}:= \frac{1}{n(n-1)}\sum_{i=1}^n \sum_{j\neq i}^n \varphi(W_i,W_j),
\]
where, without loss of generality, the scalar kernel $\varphi$ is symmetric and centred on cluster-generic pairs, i.e., \(\E[\varphi(W_i,W_j)]=0\) whenever \(j\notin\G_{g(i)}\). No centring is imposed on within-cluster pairs. Define the first-order projection \(\varphi_j^{(1)}(W_i):=\E[\varphi(W_i,W_j)\mid W_i]\), its cross-cluster average
\( \displaystyle 
    \bar\varphi_n(W_i):= \frac{1}{n-n_{g(i)}} \sum_{j=1}^{n} \indicator{j\not\in \G_{g(i)} } \varphi_j^{(1)}(W_i),
\)
and the scalar variance term
\[
    \sigma_{n,2}^2:= \frac{4}{n} \sum_{g=1}^G \E\Big[ \Big( \sum_{i\in\G_g} \bar\varphi_n(W_i) \Big)^2 \Big].
\]
\noindent The following result provides a limit theory for \(U_{n,2}\) under clustered sampling.

\begin{theorem}\label{thm:CLT_Ustats_Clus}
    Suppose
    \begin{equation}\label{eq:clus_ui_2}
        \lim_{M\rightarrow \infty} \sup_{i,j \leq n } \E\big[ | \varphi(W_i,W_j) |^r \indicator{ | \varphi(W_i,W_j) | > M } \big] = 0
    \end{equation}
    and
    \Cref{ass:sampling,ass:clus_het} hold. Then
    (a) \( \displaystyle U_{n,2} = \frac{2}{n} \sum_{i=1}^n \bar\varphi_n(W_i) + \bigOp\Big(\frac{m_n}{n}\Big) \);
    (b) \( \displaystyle U_{n,2}=\bigOp\Big(\sqrt{\frac{m_n}{n}}\Big)=o_p(1) \).
    If, in addition, \( \sigma_{n,2} \geq \lambda > 0 \), then
    (c) \( \displaystyle \sigma_{n,2}^{-1}\sqrt{n} U_{n,2} \xrightarrow{d} \mathcal{N}(0,1) \).
\end{theorem}
\noindent The theorem makes the order-$2$ mechanism explicit. Part (a) gives the asymptotically linear, Hájek-type representation, with the non-linear remainder reduced to order $m_n/n$ by counting same-cluster collisions and the degenerate term. Part (b) is then the corresponding weak law, while part (c) applies the clustered central limit theorem to the leading projection under a non-degeneracy condition. When $m_n=1$, i.e., when $n_g=1$ for every cluster, $\G_{g(i)}=\{i\}$ for every $i$, $\bar\varphi_n(W_i)$ reduces to the usual Hájek projection $\frac{1}{n-1}\sum_{j\neq i}\varphi_j^{(1)}(W_i)$ of an iid $U$-statistic, and all three parts recover the classical iid $U$-statistic asymptotics as a special case, e.g., \citet[Theorem 1, Sect. 3.2.1]{lee1990ustatistics}.

\subsection[Exact m-dependence]{Exact $m$-dependence}\label{Sect:FixedM}
The clustered order-$2$ argument has a direct temporal analogue. Same-cluster collision terms are replaced by lag-collision terms, and the clustered columns are replaced by residue-class columns modulo $m+1$. Exact finite-range dependence is a clean setting for seeing this analogy because those columns comprise independent samples. This lag-generic/lag-collision distinction is close in spirit to the distinction made by \citet{sen1963properties} between non-serial tuples and tuples containing serial interactions.

The next definition formalises $m$-dependence, following \citet[Definition 4.1]{henze2024asymptotic}.

\begin{definition}[$m$-Dependence]\label{def:m_dep}
    Let $m\in\mathbb{N}_0:=\{0,1,2,\ldots\}$. A sequence $\{W_i\}_{i\geq1}$ is called \emph{$m$-dependent} if the $\sigma$-fields $\sigma(W_1,\ldots,W_s)$ and $\sigma(W_{s+m+j}:j\geq1)$ are independent for every $s\geq1$.
\end{definition}
\noindent Equivalently, any two blocks of the sequence separated by more than $m$ time periods are independent, while dependence within an $m$-neighbourhood is unrestricted. The case $m=0$ recovers an independent sequence $\{W_i\}$. Index sets that are pairwise separated by more than $m$ are mutually independent, e.g. \citet[Sect.~2.2]{janson2023asymptotic}.

Let $W_i$, $i\in[n]$, generically denote data satisfying \Cref{def:m_dep}. Denote the U-statistic by
\[
    U_{n,2}:= \frac{1}{n(n-1)}\sum_{i=1}^n \sum_{j\neq i} \varphi(W_i,W_j),
\]
where, without loss of generality, the scalar kernel $\varphi$ is symmetric and centred on lag-generic pairs, i.e., $\E[\varphi(W_i,W_j)]=0$ whenever $|i-j|>m$. No centring is imposed on lag-collision pairs. Write \(\varphi_j^{(1)}(W_i):=\E[\varphi(W_i,W_j)\mid W_i]\), let \(Q_i:=\#\{j\in[n]\setminus\{i\}:|i-j|>m\}\), and define the lag-generic average \( \displaystyle \bar\varphi_{n,m,2,i}^{(1)}(W_i):=Q_i^{-1}\sum_{j\neq i,\,|i-j|>m}\varphi_j^{(1)}(W_i)\).

Partition $[n]$ into \(\mathcal M_{n,m}(\ell):=\{i\in[n]:i\bmod(m+1)=\ell\}\), \(\ell=0,\ldots,m\). Consecutive elements of each partition are spaced $m+1$ apart and are therefore mutually independent. Call a pair $(i,j)$ a \emph{lag-collision} if \(|i-j|\leq m\) and \emph{lag-generic} otherwise; no lag-collision pair lies within a single partition. Let \(\sigma_{n,m,2}^2:=4n^{-1}\V[\sum_{i=1}^n\bar\varphi_{n,m,2,i}^{(1)}(W_i)]\). The following extends \Cref{thm:CLT_Ustats_Clus} to the \(m\)-dependence setting.

\begin{theorem}[Fixed-\texorpdfstring{$m$}{m} CLT]\label{thm:CLT_Ustats_mdep}
    Suppose
    \begin{equation}\label{eq:kernel_ui_2}
        \lim_{M\rightarrow\infty}\sup_{i,j\leq n}\E\big[|\varphi(W_i,W_j)|^r\indicator{|\varphi(W_i,W_j)|>M}\big]=0
    \end{equation}
    for some $r\geq2$, and $\{W_i\}_{i\geq1}$ satisfies \Cref{def:m_dep} for some fixed $m\in\mathbb N_0$. Then
    (a) \( \displaystyle U_{n,2}=\frac{2}{n}\sum_{i=1}^n \bar\varphi_{n,m,2,i}^{(1)}(W_i)+\bigOp\Big(\frac{m+1}{n}\Big) \);
    (b) \( \displaystyle U_{n,2}=\bigOp\Big(\sqrt{\frac{m+1}{n}}\Big)=o_p(1) \).
    If, in addition, \eqref{eq:kernel_ui_2} holds for some \(r>2\) and \(\sigma_{n,m,2}^2\geq\lambda>0\), 
    then (c) \( \displaystyle \sigma_{n,m,2}^{-1}\sqrt{n}\,U_{n,2}\xrightarrow{d}\mathcal N(0,1) \).
\end{theorem}
\subsection{Near-epoch dependence}\label{Sect:NED}
The exact-$m$ result completes the order-$2$ argument when dependence has a fixed range. Near-epoch dependence relaxes that restriction by replacing the observed process with finite-range approximations. In the clustered design, the column partition is data-driven through the largest cluster size \(m_n\); under exact \(m\)-dependence, the corresponding partition width is structural and equals \(m+1\); under NED, the sequence \(m_n\) is instead an \emph{approximating-window device} used only to transfer the exact-\(m\) argument. Under exact \(m\)-dependence, the span \(m\) marks a region of unrestricted local dependence and exact independence beyond it. Under NED, the same span appears only in the approximating proxy \(W_i^{(m)}\), which depends on a block of \(m+1\) nearby i.i.d. innovations. Thus the local dependence inside the proxy is structured rather than arbitrary, and the original process is recovered by letting the approximation window widen. For $W_i$ generated by an i.i.d.-driven process, write
\[
    \mathcal F_a^b:=\sigma\big(\mathcal E_t:a\leq t\leq b \big),\qquad a\leq b\in\mathbb Z\cup\{\pm\infty\},
\]
where $\{\mathcal E_t\}_{t\in\mathbb Z}$ is an i.i.d.\ base sequence. The construction below allows measurable maps \(\Psi_i\) such that \(W_i=\Psi_i((\mathcal E_{i+t})_{t\in\mathbb Z})\). For $m\in\mathbb N_0$, define
\(
    W_i^{(m)}:=\E\big[W_i\mid\mathcal F_{i-\floor{m/2}}^{i+\ceil{m/2}}\big],
\)
where $\mathcal F_{i-\floor{m/2}}^{i+\ceil{m/2}}$ is the sigma-field generated by a retained innovation block of width $m+1$ around date $i$. This convention indexes the approximation by its dependence span, thereby rendering the approximating sequence $\{W_i^{(m)}\}$ exactly $m$-dependent.

The exact-\(m\) theorem can then be used for the approximating statistic:
\[
    U_{n,2}^{[m]}:= \frac{1}{n(n-1)}\sum_{i=1}^n \sum_{j\neq i} \varphi(W_i^{(m)},W_j^{(m)}).
\]
The main remaining step is to make the difference between the original and approximating $U$-statistics negligible. Process approximation alone is not enough for this: one also needs a kernel-regularity bridge that transfers the input approximation rate \(\nu_m=\|W_i-W_i^{(m)}\|_2\) into a kernel approximation rate.

This finite-window approximation motivates the following $L_2$-NED condition, adapted to the sequence formulation in \citet[Def.~17.1]{davidson1994stochastic}.
\begin{assumption}[Near-epoch dependence of $\{W_i\}$]\label{ass:ned}
    There is an i.i.d.\ base sequence $\{\mathcal E_t\}_{t\in\mathbb Z}$ and measurable maps \(\Psi_i\) such that \(W_i=\Psi_i((\mathcal E_{i+t})_{t\in\mathbb Z})\), where \(\mathcal F_a^b:=\sigma(\mathcal E_t:a\leq t\leq b)\). Moreover, $\{W_i\}$ is uniformly $L_2$-near-epoch dependent on this base sequence, i.e., there is a sequence of constants $\nu_m\to0$ such that
    \[
        \|W_i-W_i^{(m)}\|_2=\big\|W_i-\E[W_i\mid\mathcal F_{i-\floor{m/2}}^{i+\ceil{m/2}}]\big\|_2\leq\nu_m\quad\text{for every }i\in\mathbb N\text{ and }m\in\mathbb N_0.
    \]
\end{assumption}
\noindent Davidson's scaling constants $d_i$ are absorbed into the uniform sequence \(\nu_m\). This is only a normalisation of the NED bound; it does not impose stationarity on \(\{W_i\}\).

Following \citet[p.~262]{davidson1994stochastic}, $\{W_i\}$ is $L_2$-NED \emph{of size $-\zeta_0$} if $\nu_m=\bigO(m^{-\zeta})$ for some $\zeta>\zeta_0$. \Cref{ass:ned} controls the input approximation error \(W_i-W_i^{(m)}\), but the $U$-statistic also requires a kernel-level bridge from input approximation to kernel approximation. Standard kernels exhibit two useful perturbation behaviours. The Gini mean difference kernel \(\varphi(x,y)=|x-y|\) is globally Lipschitz, while rank kernels such as Kendall's tau and Spearman-type sign-product kernels are not pointwise Lipschitz but can be controlled after taking expectations. These kernels are naturally handled by an expected-H\"older condition. By contrast, the squared-difference kernel satisfies \( |(w-y)^2-(\widetilde w-y)^2|=|w-\widetilde w|\,|w+\widetilde w-2y| \), so its local slope is available but random and unbounded; this is better captured by a data-dependent modulus with tail and moment control. Because both behaviours arise for common $U$-statistic kernels, the next assumption supplies either of two non-nested one-coordinate perturbation bridges. In both branches, the held-fixed argument is understood to range over the relevant original, truncated, mixed, and generic independent-copy choices.

\begin{assumption}[Kernel perturbation for the order-$2$ kernel]\label{ass:kernel_perturb_2}
    One of the following two branches holds.
	    \begin{enumerate}[label=(\alph*)]
	        \item \emph{Expected-H\"older branch.} There exist $C<\infty$ and \(\alpha\in(0,1]\) such that, uniformly over \(i\neq j\), \(m\in\mathbb N_0\), and every relevant held-fixed argument \(Z\),
	        \[
	            \|\varphi(W_i,Z)-\varphi(W_i^{(m)},Z)\|_2
	            \leq
	            C\|W_i-W_i^{(m)}\|_2^\alpha
		        \]

        \item \emph{Data-dependent modulus branch.} There is a measurable function \(L(w,\widetilde w,z)\), taking values in \([0,\infty)\), such that, uniformly over \(i\neq j\), \(m\in\mathbb N_0\), and every relevant held-fixed argument \(Z\),
        \[
            |\varphi(W_i,Z)-\varphi(W_i^{(m)},Z)|
            \leq
            L(W_i,W_i^{(m)},Z)\|W_i-W_i^{(m)}\|
        \]
        There exists a \(\delta>0\) such that, as \(M\to\infty\),
        \[
            \E[L(W_i,W_i^{(m)},Z)\indicator{L(W_i,W_i^{(m)},Z)>M}]
            =\bigO(M^{-\delta}),
        \]
        uniformly over the same indices and held-fixed argument choices. For some \(r>2\),
        \[
        \sup_{i\neq j,\,m\in\mathbb N_0}
        \big\|\max\{|\varphi(W_i,W_j)|,|\varphi(W_i^{(m)},W_j)|,
        |\varphi(W_i^{(m)},W_j^{(m)})|\}\big\|_r<\infty.
        \]
        The same tail and moment bounds are understood for the corresponding generic independent copies used in projection arguments.
    \end{enumerate}
\end{assumption}

In either branch, \(\rho_m\) is the effective kernel-approximation rate induced by the NED approximation. It is read branchwise: in the expected-H\"older branch, \(\rho_m:=\nu_m^\alpha\); in the data-dependent modulus branch, with \(\delta\) and \(r\) as in the assumption, let \(\tau:=(1+\delta)(r-2)/(2r)\), \(\beta:=\tau/(1+\tau)\), and \(\rho_m:=\nu_m^\beta\). Both branches keep the same one-coordinate perturbation format: only one argument is replaced, while the co-index is read with the relevant independent or mixed original/truncated law. The case \(\alpha=1\) in the expected-H\"older branch is the expected Lipschitz condition along the NED couplings, and a deterministic Lipschitz bound verifies it immediately. Some kernels with unbounded local moduli may still fall into this branch when the relevant NED couplings have enough integrability to absorb the local slope into the finite constant \(C\).\footnote{For example, under common-mean, common-scale Gaussian NED couplings, the non-globally-Lipschitz kernel \(\varphi(x,y)=(x-y)^2\) can satisfy the expected branch along the relevant couplings: with \(D:=W-\widetilde W\), \(A:=W+\widetilde W-2W^\star\), and \(W^\star\independent(W,\widetilde W)\), joint normality gives \(D\independent A\), so \(\|\varphi(W,W^\star)-\varphi(\widetilde W,W^\star)\|_2=\|D\|_2\|A\|_2\), while \(\|A\|_2\leq2\sqrt2\,\sigma\). Thus \(C=2\sqrt2\,\sigma\) is admissible for those couplings.}

The two branches are complementary. The Gini mean difference sits in the overlap: it satisfies the expected-H\"older branch with \(\alpha=1\), and it also satisfies the data-dependent modulus branch with a constant modulus. For globally Lipschitz kernels, the expected-H\"older route is simpler and sharper. For kernels with random unbounded slopes, including sample variances, squared deviations, and quadratic loss functions, the data-dependent modulus branch can be the more natural route. For sign kernels, the expected branch may be verified by a conditional crossing-probability argument rather than by pointwise continuity; \Cref{lem:sign_crossing_bound} gives a compact sufficient condition yielding \(\alpha=1/2\).

The following provides a polynomial rate illustration.
\begin{example}[Kernel--dependence rate balance]\label{ex:rate_balance}
    If \(\rho_m=\bigO(m^{-\eta})\) and \(m_n\sim n^\gamma\), then \(m_n=o(\sqrt n)\) requires \(\gamma<1/2\), while \(\sqrt n\,\rho_{m_n}=o(1)\) requires \(\gamma\eta>1/2\). Hence the admissible region is \(\gamma\in(1/(2\eta),1/2)\), which is non-empty exactly when \(\eta>1\). In the expected-H\"older branch, \(\eta=\alpha\zeta\) when \(\nu_m=\bigO(m^{-\zeta})\); in the data-dependent modulus branch, \(\eta=\beta\zeta\). Thus feasibility is governed jointly by kernel regularity and the finite-window approximation rate of the data. Smoother kernels permit slower NED approximation, whereas rougher kernels require faster NED approximation to compensate. For instance, in the expected-H\"older branch with \(\alpha=1/2\) and \(\zeta=3\), any \(\gamma\in(1/3,1/2)\) is admissible.
\end{example}

All ingredients are now in place for the order-$2$ limiting result under NED data.
In the NED subsections, centring is understood with respect to the date-specific independent-copy projection used in the displayed linear term; cf. \citet{fischer2016multivariate}. 

\begin{theorem}[Order-\texorpdfstring{$2$}{2} NED CLT]\label{thm:CLT_NED_2}
    Let $\varphi(w_1,w_2)$ be symmetric and centred in the independent-copy sense, let \(\varphi_i^{(1)}(W_i)\) denote the corresponding centred first-order projection, and set \(\sigma_n^2:=\V[n^{-1/2}\sum_{i=1}^n\varphi_i^{(1)}(W_i)]\). Suppose \Cref{ass:ned,ass:kernel_perturb_2} hold, the uniform integrability condition \eqref{eq:kernel_ui_2} holds for $\{\varphi(W_i,W_j)\}$ and uniformly for $\{\varphi(W_i^{(m)},W_j^{(m)})\}$, \(\rho_m=\bigO(m^{-\eta})\) for some \(\eta>1\).
    If $m_n\to\infty$, $m_n=o(\sqrt n)$, and $\sqrt n\,\rho_{m_n}=o(1)$, then
    (a) $\displaystyle U_{n,2}=\frac2n\sum_{i=1}^n\varphi_i^{(1)}(W_i)+o_p(n^{-1/2})$;
    (b) $\displaystyle U_{n,2}=\bigOp(n^{-1/2})=o_p(1)$.
    If, in addition, the uniform integrability condition \eqref{eq:kernel_ui_2} holds for some \(r>2\) and \(\sigma_n\geq\lambda>0\), then
    (c) $\displaystyle (2\sigma_n)^{-1}\sqrt n\,U_{n,2}\xrightarrow{d}\mathcal N(0,1)$.
\end{theorem}
\noindent The same projection logic is still at work. The statistic is reduced to the linear term $2n^{-1}\sum_i\varphi_i^{(1)}(W_i)$, but now near-epoch dependence supplies the approximation step that replaces the fixed-$m$ argument. The CLT is normalised by the finite-sample variance of the linear term.

\begin{remark}
	The i.i.d.\ base formulation is used to make the truncation step exact. Once $W_i$ is replaced by $W_i^{(m)}=\E[W_i\mid\mathcal F_{i-\floor{m/2}}^{i+\ceil{m/2}}]$, retained innovation blocks attached to observations more than $m$ dates apart are non-overlapping and therefore independent, so the approximating sequence is exactly $m$-dependent; the formal verification is given in \Cref{app:ned}. The observed process $\{W_i\}$ may still exhibit persistent and non-linear dependence through its underlying i.i.d.-innovation representation.\footnote{For treatments of shift transformations and their role in representing stochastic sequences, see, for example, \citet[Sects.~13.1--13.2]{davidson2021stochastic}.} This formulation covers many standard econometric data-generating processes, including linear processes such as MA($\infty$) and stationary ARMA models, volatility recursions such as ARCH/GARCH-type specifications, and smooth non-linear autoregressions, whenever the usual stability, contraction, and moment conditions deliver finite-window approximation errors, e.g., \citet[Examples 17.3--17.4]{davidson2021stochastic}. A broader NED-on-mixing formulation would require an additional coupling layer, which is not pursued here.
\end{remark}

\section{Order-\texorpdfstring{$k$}{k} Theory}\label{Sect:OrderKTheory}
The order-$2$ pipeline in the preceding section isolates the main ideas: arrange the data so that independent pieces are visible, control the relevant same-cluster or lag-collision terms, reduce the statistic to its first-order projection, and then apply the relevant central limit theorem to that projection. The order-$k$ theory keeps the same structure. Throughout the order-$k$ results, \(k\geq2\) is fixed as \(n\to\infty\); constants may depend on \(k\), but this dependence is not optimised.

\subsection[Clustered order-k U-statistic]{Clustered order-$k$ U-statistic}\label{Sect:Orderk}

Let $\varphi(w_1,\ldots,w_k)$, $k\geq2$, be symmetric and centred on cluster-generic $k$-tuples, so
\[
    \E[\varphi(W_{i_1},\ldots,W_{i_k})]=0
\]
whenever the indices are distinct and belong to distinct clusters. The order-$k$ $U$-statistic is
\[
    U_{n,k}:= \frac{1}{\Perm{n}{k}}\sum_{(i_1,\ldots,i_k)_{\neq}}\varphi(W_{i_1},\ldots,W_{i_k}),
\]
where $\Perm{n}{k}:=n(n-1)\cdots(n-k+1)$ denotes the number of permutations of $k$ distinct indices from $[n]$. Define the first-order projection $\varphi^{(1)}_{(i_2,\ldots,i_k)}(W_i):=\E[\varphi(W_i,W_{i_2},\ldots,W_{i_k})\mid W_i]$. Let
\(
    Q_{n,i}:=\#\{(i_2,\ldots,i_k)_{\neq}: g(i,i_2,\ldots,i_k)_{\neq}\},
\)
and define the cluster-generic co-index average
\[
    \bar\varphi_{n,k}^{(1)}(W_i):=\frac{1}{Q_{n,i}}\sum_{\substack{(i_2,\ldots,i_k)_{\neq}\\ g(i,i_2,\ldots,i_k)_{\neq}}}\varphi^{(1)}_{(i_2,\ldots,i_k)}(W_i),
\]
where the sum ranges over ordered co-index tuples for which \((i,i_2,\ldots,i_k)\) is cluster-generic. Hence \(\E[\bar\varphi_{n,k}^{(1)}(W_i)]=0\). \(Q_{n,i}=n-n_{g(i)}\) and \(\bar\varphi_{n,2}^{(1)}(W_i)=\bar\varphi_n(W_i)\) are recovered at $k=2$. Let $\displaystyle \sigma_{n,k}^2:= \frac{k^2}{n} \sum_{g=1}^G \E\Big[ \Big( \sum_{i\in\G_g} \bar\varphi_{n,k}^{(1)}(W_i) \Big)^2 \Big]$, which reduces to $\sigma_{n,2}^2$ at $k=2$. 

\begin{namedtheorem}[Order-\texorpdfstring{$k$}{k} CLT]{$1'$}\label{thm:CLT_Ustats_Clus_k}
    Suppose
    \begin{equation}\label{eq:clus_ui_k}
        \lim_{M\rightarrow \infty} \sup_{\mathbf i}
        \E\big[ | \varphi(W_{i_1},\ldots,W_{i_k}) |^r
        \indicator{ | \varphi(W_{i_1},\ldots,W_{i_k}) | > M } \big] = 0
    \end{equation}
    and \Cref{ass:sampling,ass:clus_het} hold. Then (a)
    \( \displaystyle
        U_{n,k} = \frac{k}{n} \sum_{i=1}^n \bar\varphi_{n,k}^{(1)}(W_i)
        + \bigOp\Big(\frac{m_n}{n}\Big);
    \)
    (b) \( \displaystyle U_{n,k}=\bigOp\Big(\sqrt{\frac{m_n}{n}}\Big)=o_p(1) \).
    If, in addition, \( \sigma_{n,k} \geq \lambda > 0 \), then
    (c) \( \displaystyle \sigma_{n,k}^{-1}\sqrt{n} \, U_{n,k} \xrightarrow{d} \mathcal{N}(0,1) \).
\end{namedtheorem}
\noindent The structure is unchanged from the order-$2$ case: the statistic is linearised by its first-order projection, while the column partition keeps higher-order collision and degenerate terms asymptotically smaller.

\subsection[Order-k fixed-m CLT]{Order-$k$ fixed-$m$ CLT}\label{Sect:FixedMk}

The fixed-$m$ order-$k$ theorem is the temporal counterpart of \Cref{thm:CLT_Ustats_Clus_k}. Let \(\varphi(w_1,\ldots,w_k)\), \(k\geq2\), be symmetric and centred on lag-generic \(k\)-tuples. Thus \(\E[\varphi(W_{i_1},\ldots,W_{i_k})]=0\) whenever every pair of coordinates is more than \(m\) apart. Denote the statistic by
\[
    U_{n,k}:=\frac1{\Perm nk}\sum_{(i_1,\ldots,i_k)_{\neq}}\varphi(W_{i_1},\ldots,W_{i_k}).
\]
Call a $k$-tuple of distinct indices \emph{lag-generic} if every pair of its coordinates is more than $m$ apart, matching the lag-generic/lag-collision distinction in \Cref{Sect:FixedM}. For such tuples, $W_{i_1},\ldots,W_{i_k}$ are mutually independent by \Cref{def:m_dep} (sorting the coordinates, every consecutive gap exceeds $m$). Define \(\varphi_{\mathbf j}^{(1)}(W_i):=\E[\varphi(W_i,W_{j_2},\ldots,W_{j_k})\mid W_i]\). Let \(Q_i\) be the number of ordered co-index tuples \(\mathbf j=(j_2,\ldots,j_k)\) for which \((i,\mathbf j)\) is lag-generic, and define
\[
    \bar\varphi_{n,m,k,i}^{(1)}(W_i)
    :=
    \frac1{Q_i}\sum_{\substack{(j_2,\ldots,j_k)_{\neq}\\ (i,j_2,\ldots,j_k)\ \mathrm{lag\mbox{-}generic}}}
    \varphi_{\mathbf j}^{(1)}(W_i).
\]

Let \(\displaystyle\sigma_{n,m,k}^2:=\frac{k^2}{n}\V\big[\sum_{i=1}^n\bar\varphi_{n,m,k,i}^{(1)}(W_i)\big]\), which reduces to \(\sigma_{n,m,2}^2\) at \(k=2\).

\begin{namedtheorem}[Order-\texorpdfstring{$k$}{k} Fixed-\texorpdfstring{$m$}{m} CLT]{$2'$}\label{thm:CLT_Ustats_mdep_k}
    Suppose
    \begin{equation}\label{eq:kernel_ui_k}
        \lim_{M\rightarrow\infty}\sup_{\mathbf i}
        \E\big[|\varphi(W_{i_1},\ldots,W_{i_k})|^r
        \indicator{|\varphi(W_{i_1},\ldots,W_{i_k})|>M}\big]=0
    \end{equation}
    for some $r\geq2$, and $\{W_i\}_{i\geq1}$ satisfies \Cref{def:m_dep} for some fixed $m\in\mathbb N_0$. Then (a)
    \( \displaystyle
        U_{n,k}=\frac kn\sum_{i=1}^n\bar\varphi_{n,m,k,i}^{(1)}(W_i)+\bigOp\Big(\frac{m+1}{n}\Big);
    \)
    (b) \( \displaystyle U_{n,k}=\bigOp\Big(\sqrt{\frac{m+1}{n}}\Big)=o_p(1) \).
    If, in addition, \eqref{eq:kernel_ui_k} holds for some \(r>2\) and $\sigma_{n,m,k}\geq\lambda>0$, then
    (c) \( \displaystyle
        \sigma_{n,m,k}^{-1}\sqrt n\,U_{n,k}
        \xrightarrow{d}\mathcal N(0,1)
    \).
\end{namedtheorem}
\noindent At $k=2$, \Cref{thm:CLT_Ustats_mdep_k} recovers \Cref{thm:CLT_Ustats_mdep}. The lag-generic/lag-collision split has a classical antecedent in \citet{sen1963properties}, while \citet{malevich_abdalimov_1982} studies stationary $U$-statistics under $m$-dependence in both the fixed-$m$ case and regimes with $m=m(n)\to\infty$. For symmetric $\varphi$ and under stationarity, \Cref{thm:CLT_Ustats_mdep_k} is also consistent with the general fixed-$m$ $U$-statistic CLT of \citet[Thm.~3.3]{janson2023asymptotic}.

\subsection[Order-k NED CLT]{Order-$k$ NED CLT}\label{Sect:DoubleLimit}
The NED extension uses the same approximating sequence and rate balance introduced in \Cref{Sect:NED}. The order-$k$ additions are the order-$k$ version of the same two-branch kernel perturbation condition. In both branches, the held-fixed co-ordinates are understood to range over the relevant original, truncated, mixed, and generic independent-copy choices. For \(a=0,\ldots,k\), write \(Z_{\mathbf i}^{(m,a)}\) for \((W_{i_1},\ldots,W_{i_k})\) with the first \(a\) coordinates replaced by their \(m\)-approximants. For a one-coordinate replacement in coordinate \(a\), write \(Z_{-a}:=(Z_1,\ldots,Z_{a-1},Z_{a+1},\ldots,Z_k)\) for the held-fixed collection.

\begin{namedassumption}[Kernel perturbation for the order-$k$ kernel]{$5'$}\label{ass:kernel_perturb_k}
    One of the following two branches holds.
	    \begin{enumerate}[label=(\alph*),itemsep=0pt,topsep=0pt]
		        \item \emph{Expected-H\"older branch.} There exist $C<\infty$ and \(\alpha\in(0,1]\) such that, uniformly over \(\mathbf i_{\neq}\), \(m\in\mathbb N_0\), \(1\leq a\leq k\), and every relevant held-fixed collection \(Z_{-a}\),
	        \[
	        \begin{aligned}
	            &\|\varphi(Z_1,\ldots,Z_{a-1},W_{i_a},Z_{a+1},\ldots,Z_k)
	            -\varphi(Z_1,\ldots,Z_{a-1},W_{i_a}^{(m)},Z_{a+1},\ldots,Z_k)\|_2\\
	            &\hspace{1.5cm}\leq
	            C\|W_{i_a}-W_{i_a}^{(m)}\|_2^\alpha,
	        \end{aligned}
		        \]

        \item \emph{Data-dependent modulus branch.} There is a measurable function
        \(L(w,\widetilde w,z_{-a})\), taking values in \([0,\infty)\), such that, uniformly over \(\mathbf i_{\neq}\), \(m\in\mathbb N_0\), \(1\leq a\leq k\), and every relevant held-fixed collection \(Z_{-a}\),
        \[
        \begin{aligned}
            &\big|
            \varphi(Z_1,\ldots,Z_{a-1},W_{i_a},Z_{a+1},\ldots,Z_k)\\
            &\hspace{0.5cm}
            -\varphi(Z_1,\ldots,Z_{a-1},W_{i_a}^{(m)},Z_{a+1},\ldots,Z_k)
            \big| \\
            &\hspace{1.5cm}\leq
            L(W_{i_a},W_{i_a}^{(m)},Z_{-a})\|W_{i_a}-W_{i_a}^{(m)}\|.
        \end{aligned}
        \]
        There is $\delta>0$ such that, as \(M\to\infty\),
        \[
            \E[L(W_{i_a},W_{i_a}^{(m)},Z_{-a})
            \indicator{L(W_{i_a},W_{i_a}^{(m)},Z_{-a})>M}]
            =
            \bigO(M^{-\delta}),
        \]
        uniformly over the same indices and held-fixed collections as in the preceding display.
        For some \(r>2\),
        \begin{equation}\label{eq:random_modulus_kernel_moment_k}
            \sup_{\mathbf i_{\neq},\,m\in\mathbb N_0,\,0\le a\le k}
            \|\varphi(Z_{\mathbf i}^{(m,a)})\|_r
            <\infty
        \end{equation}
        The same tail and moment bounds are understood for the corresponding generic independent copies used in projection arguments.
    \end{enumerate}
\end{namedassumption}

\noindent The effective rate \(\rho_m\) is read branchwise exactly as in the order-$2$ case: in the expected-H\"older branch, \(\rho_m:=\nu_m^\alpha\); in the data-dependent modulus branch, with \(\delta\) and \(r\) as in the assumption, let \(\tau:=(1+\delta)(r-2)/(2r)\), \(\beta:=\tau/(1+\tau)\), and \(\rho_m:=\nu_m^\beta\). At \(k=2\), this is the order-$k$ analogue of \Cref{ass:kernel_perturb_2}. It retains the one-coordinate perturbation format in both branches.

The theorem below provides the order-$k$ extension of the NED limit theory.

\begin{namedtheorem}[Order-\texorpdfstring{$k$}{k} NED CLT]{$3'$}\label{thm:CLT_double_limit}
    Let $\varphi(w_1,\ldots,w_k)$, $k\geq2$, be symmetric and centred in the independent-copy sense, let \(\varphi_i^{(1)}(W_i)\) denote the corresponding centred first-order projection, and set \(\sigma_n^2:=\V[n^{-1/2}\sum_{i=1}^n\varphi_i^{(1)}(W_i)]\). Suppose that \Cref{ass:ned,ass:kernel_perturb_k} hold, the uniform integrability condition \eqref{eq:kernel_ui_k} holds for $\{\varphi(W_{i_1},\ldots,W_{i_k})\}$ and uniformly in $m$ for $\{\varphi(W_{i_1}^{(m)},\ldots,W_{i_k}^{(m)})\}$, \(\rho_m=\bigO(m^{-\eta})\) for some \(\eta>1\).
    If $m_n\to\infty$, $m_n=o(\sqrt n)$, and $\sqrt n\,\rho_{m_n}=o(1)$, then
    (a) $\displaystyle U_{n,k}=\frac kn\sum_{i=1}^n\varphi_i^{(1)}(W_i)+o_p(n^{-1/2})$;
    (b) $\displaystyle U_{n,k}=\bigOp(n^{-1/2})=o_p(1)$.
    If, in addition, the uniform integrability condition \eqref{eq:kernel_ui_k} holds for some \(r>2\) and \(\sigma_n\geq\lambda>0\), then
    (c) $\displaystyle (k\sigma_n)^{-1}\sqrt n\,U_{n,k}\xrightarrow{d}\mathcal N(0,1)$.
\end{namedtheorem}

\noindent \Cref{thm:CLT_double_limit} closes the order-$k$ NED extension: it delivers a CLT for $U_{n,k}$ under near-epoch dependence, with no exact-$m$-dependence requirement on $\{W_i\}$ itself.

\section{Feasible Inference}\label{Sect:FeasibleInference}
The limit theorems above reduce both sampling designs to the covariance of the first-order projection. In the clustered case that covariance is estimated by cluster aggregation, while under near-epoch dependence it is estimated by a long-run covariance procedure.

\subsection{Cluster-robust projection covariance}\label{Sect:ClusterEstimation}
Inference via \Cref{thm:CLT_Ustats_Clus_k} requires an estimator of the covariance of the order-$k$ first-order projection. The scalar target is $\sigma_{n,k}^2$ from \Cref{Sect:Orderk}. For each \(i\in[n]\), define the all-coindex leave-one-out average
\[
    \tilde\varphi_{n,k,i}
    :=
    \frac1{\Perm{n-1}{k-1}}
    \sum_{\substack{(i_2,\ldots,i_k)_{\neq}\\ i_a\in[n]\setminus\{i\},\ a=2,\ldots,k}}
    \varphi(W_i,W_{i_2},\ldots,W_{i_k}).
\]
Then \(U_{n,k}=n^{-1}\sum_{i=1}^n\tilde\varphi_{n,k,i}\). Centre this observation-level vector by setting
\(
    \hat\varphi^{(1)}_{n,k,i}:=\tilde\varphi_{n,k,i}-U_{n,k}.
\)
The natural cluster-robust plug-in estimator is
\[
    \hat\sigma_{n,k}^2:=\frac{k^2}{n}\sum_{g=1}^G\Big(\sum_{i\in\G_g}\hat\varphi^{(1)}_{n,k,i}\Big)^2.
\]

\noindent Thus the statistic and covariance estimator are built from the same \(n\)-vector, i.e., the average of \(\tilde\varphi_{n,k,i}\) gives the $U$-statistic, while its centred version gives the cluster summands for feasible inference. The consistency argument compares this all-coindex empirical projection with the infeasible cluster-generic projection after cluster aggregation.

\begin{theorem}[Consistent estimation of $\sigma_{n,k}^2$]\label{thm:cluster_var_consistency}
    Under the assumptions of \Cref{thm:CLT_Ustats_Clus_k}, including $\sigma_{n,k}\geq\lambda>0$, $\hat\sigma_{n,k}^2/\sigma_{n,k}^2\xrightarrow{p}1$.
\end{theorem}
\subsection{HAC projection covariance}\label{Sect:Estimation}

\citet{dehling2017testing} considers this problem for Kendall's tau, an order-$2$ $U$-statistic, under near-epoch dependence. The construction combines a plug-in estimator of the first-order projection with a HAC (heteroscedasticity- and autocorrelation-consistent) kernel estimator of the long-run variance, building on \citet{dejong2000consistency}. \citet{fischer2017robust} develops the same broad oracle-plus-plug-in HAC strategy for multivariate kernels in a bounded-kernel NED setting.

This section extends that plug-in-plus-HAC architecture to the order-$k$ projection in \Cref{thm:CLT_double_limit}, under the present $L_2$-NED framework and either of the complementary kernel-regularity branches. The oracle HAC step uses the heterogeneous-array consistency result of \citet{dejong2000consistency}, so stationarity is not imposed.

For $i=1,\ldots,n$, use the same leave-one-out empirical projection,
\[
    \hat\varphi^{(1)}_{n,k,i}:=\tilde\varphi_{n,k,i}-U_{n,k},
\]
as defined in \Cref{Sect:ClusterEstimation}.
Let
\(
    \hat\rho(\ell) := \frac1n\sum_{i=1}^{n-\ell}\hat\varphi^{(1)}_{n,k,i}\,\hat\varphi^{(1)}_{n,k,i+\ell},\quad \ell=0,1,\ldots,n-1
\) be the empirical covariance for lag $\ell$, and define the HAC-type long-run variance estimator
\[
    \hat\sigma_n^2 := \hat\rho(0)+2\sum_{\ell=1}^{n-1}\kappa\Big(\frac \ell{h_n}\Big)\hat\rho(\ell)
\]
where $\kappa(\cdot)$ is a HAC weight function and $h_n$ is a bandwidth sequence.

The following regularity condition follows \citet[Assumption 2.6]{dehling2017testing}.
\begin{assumption}[HAC weights and bandwidth]\label{ass:hac}
    $\kappa:\mathbb R\to[-1,1]$ satisfies $\kappa(0)=1$, $\kappa(x)=\kappa(-x)$ for all $x$, $\int|\kappa(x)|\,dx<\infty$, its Fourier transform $f(\xi):=(2\pi)^{-1}\int\kappa(x)e^{i\xi x}\,dx$ satisfies $\int|f(\xi)|\,d\xi<\infty$, and $\kappa$ is continuous at $0$ and at all but finitely many points. The bandwidth satisfies $h_n\to\infty$ and $h_n=o(\sqrt n)$.
\end{assumption}
\noindent \Cref{ass:hac} is satisfied by the Bartlett, Parzen, quadratic-spectral, and Tukey--Hanning weight functions; see \citet{dehling2017testing} and \citet{dejong2000consistency}.

\begin{theorem}[Consistent estimation of $\sigma_n^2$]\label{thm:hac_consistency}
    Under the assumptions of \Cref{thm:CLT_double_limit}, with the uniform integrability condition \eqref{eq:kernel_ui_k} imposed for some \(r>2\), and if \Cref{ass:hac} also holds, $\hat\sigma_n^2-\sigma_n^2\xrightarrow{p}0$.
\end{theorem}
\section{Empirical Applications}\label{Sect:Empirical}
The empirical applications mirror the paper's theoretical progression. Each application pairs a familiar order-2 statistic with a genuine order-3 statistic under a common sampling scheme, so the empirical section moves from order-2 intuition to order-$k$ feasible inference in the same spirit as the theory. The county-income application uses the Gini coefficient and L-skewness under state-level clustering, while the financial-return application uses Kendall's tau and Spearman's rho under weak time-series dependence. In both cases, the common issue is feasible inference rather than point estimation: the existing $U$-statistic literature provides only part of what is needed for these applications, and the procedures developed here supply the missing cluster-robust and weak-dependence inference steps used below.

\subsection{County-income inequality and asymmetry}\label{Sect:Empirical_Gini}

\citet{park_shin_2023_county_incomes} studies the distribution dynamics of US county per-capita incomes from 1970 to 2017, including the role of government transfer programmes in compressing cross-county income differences. The application follows that question through two index-based summaries built from the same county panel: one for overall inequality and one for distributional asymmetry. The Gini coefficient is the mean absolute pairwise income gap divided by twice the mean income, so its numerator is a natural order-2 $U$-statistic. Counties are observed repeatedly over time and are naturally grouped within states, making state-clustered inference for annual cross-county inequality a direct use case for \Cref{Sect:ClusterEstimation}. Related work by \citet{bhattacharya2007inference} and \citet{darku2020gini} treats Gini inference under survey designs; the application instead uses administrative county panels and compares pre-transfer and post-transfer county-income measures through both the overall spread and the shape of the cross-county distribution.

Let $Y_{it}$ denote per-capita income in county $i$ in year $t$. For a fixed year $t$, let \(\lambda_{\iota t}\), \(\iota\in[3]\), denote the first three L-moments of the cross-county income distribution, written as
\[
    \lambda_{1t}:=\E[\varphi_1(Y_{1t})],\quad
    \lambda_{2t}:=\E[\varphi_2(Y_{1t},Y_{2t})],\quad \text{and} \quad
    \lambda_{3t}:=\E[\varphi_3(Y_{1t},Y_{2t},Y_{3t})],
\]
with kernels
\(
    \varphi_1(y_1):=y_1,\qquad
    \varphi_2(y_1,y_2):=\frac12|y_1-y_2|,
\)
and
\[
    \varphi_3(y_1,y_2,y_3):=\frac13\{\max(y_1,y_2,y_3)-2\operatorname{med}(y_1,y_2,y_3)+\min(y_1,y_2,y_3)\};
\]
see \citet[eqn. 2.4]{hosking1990lmoments}. The application reports two ratio transforms of adjacent L-moments, namely the Gini coefficient \(\displaystyle G_t:=\frac{\lambda_{2t}}{\lambda_{1t}}\) and the L-skewness \(\displaystyle \tau_{3t}:=\frac{\lambda_{3t}}{\lambda_{2t}}\). The Gini coefficient is the ratio of the second L-moment to the mean, while L-skewness is the ratio of the third to the second L-moment. This gives the inequality and asymmetry summaries a common order-statistic backbone while keeping the second object dimension free. See also \citet{sreelakshmi_asha_nair_2015_lmoments} for related discussion of L-moments and inequality.

The sample analogues are built from an order-1 mean, an order-2 $U$-statistic, and an order-3 $U$-statistic. In particular,
\[
    \widehat\lambda_{1t}:=\frac1{n_t}\sum_{i=1}^{n_t}Y_{it},\qquad
    \widehat\lambda_{2t}:=\frac1{2n_t(n_t-1)}\sum_{i=1}^{n_t}\sum_{j\neq i}|Y_{it}-Y_{jt}|,
\]
so that \(\widehat G_t=\widehat\lambda_{2t}/\widehat\lambda_{1t}\). Here \(i\in[n_t]\) indexes counties observed in year \(t\), and \(n_t\) is the number of such counties. Likewise, \(\widehat\tau_{3t}=\widehat\lambda_{3t}/\widehat\lambda_{2t}\), where \(\widehat\lambda_{3t}\) is the order-3 $U$-statistic sample analogue of \(\lambda_{3t}\). The clustered first-order projection for \(\widehat G_t\) therefore combines the empirical projections for \((\widehat\lambda_{1t},\widehat\lambda_{2t})\) through the delta method, while that for \(\widehat\tau_{3t}\) combines those for \((\widehat\lambda_{2t},\widehat\lambda_{3t})\). In both cases, the cluster-robust variance sums the resulting first-order projections within states, exactly as in \Cref{Sect:ClusterEstimation}. The reported standard errors therefore treat within-state county outcomes as potentially dependent while using the cross-state variation for inference. This order-3 addition asks a distinct empirical question: not only whether transfers are associated with a compressed overall inequality, but also whether they are associated with altered asymmetry of the county-income distribution.

The data are the replication files from \citet{park_shin_2023_county_incomes}. Following that construction, the pre-transfer series is pre-tax county per-capita income: government transfer income and military income are excluded from Bureau of Economic Analysis (BEA) personal income before dividing by county population. Annual Gini and L-skewness measures are computed from this pre-transfer series and from the post-transfer series that adds all government transfer receipts available in the replication package. The aim is to isolate an index-based implication of that transfer exercise and equip it with inference designed for clustered $U$-statistics. For the Gini, the estimand of interest is
\[
    \Delta_t:=G_t^{\mathrm{post\ transfer}}-G_t^{\mathrm{pre\ transfer}},
\]
so negative values indicate that the post-transfer income measure is less unequal across counties than the pre-transfer measure. The corresponding L-skewness gap is defined analogously. \Cref{fig:park_shin_gaps} plots both annual paths. Each panel reports pointwise intervals and simultaneous sup-$t$ bands over calendar years, constructed following the sup-$t$ principle of \citet{li_liao_2020_uniform_nonparametric}.\footnote{The sup-$t$ step is applied to the year-indexed path of annual $U$-statistics: point estimates and state-clustered standard errors are computed year by year, and the simultaneous band calibrates the supremum of the resulting studentised path, without requiring the annual estimates to be independent across years.}

\begin{figure}[h!]
\centering
\caption{County-income order-2 and order-3 transfer effects, 1970--2017.}
\begin{subfigure}[t]{0.49\textwidth}
\centering
\includegraphics[width=\textwidth]{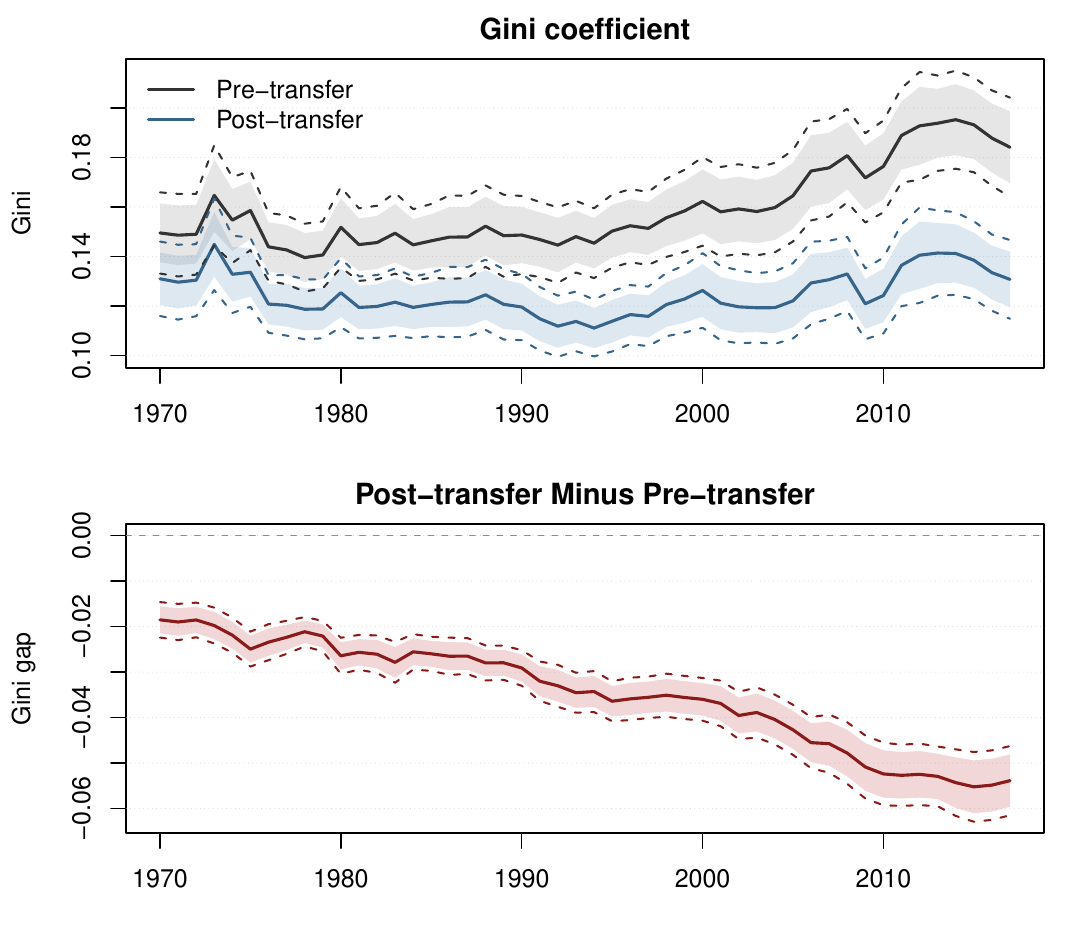}
\caption{Gini coefficients and transfer-induced Gini gaps.}
\label{fig:park_shin_gini_gap}
\end{subfigure}\hfill
\begin{subfigure}[t]{0.49\textwidth}
\centering
\includegraphics[width=\textwidth]{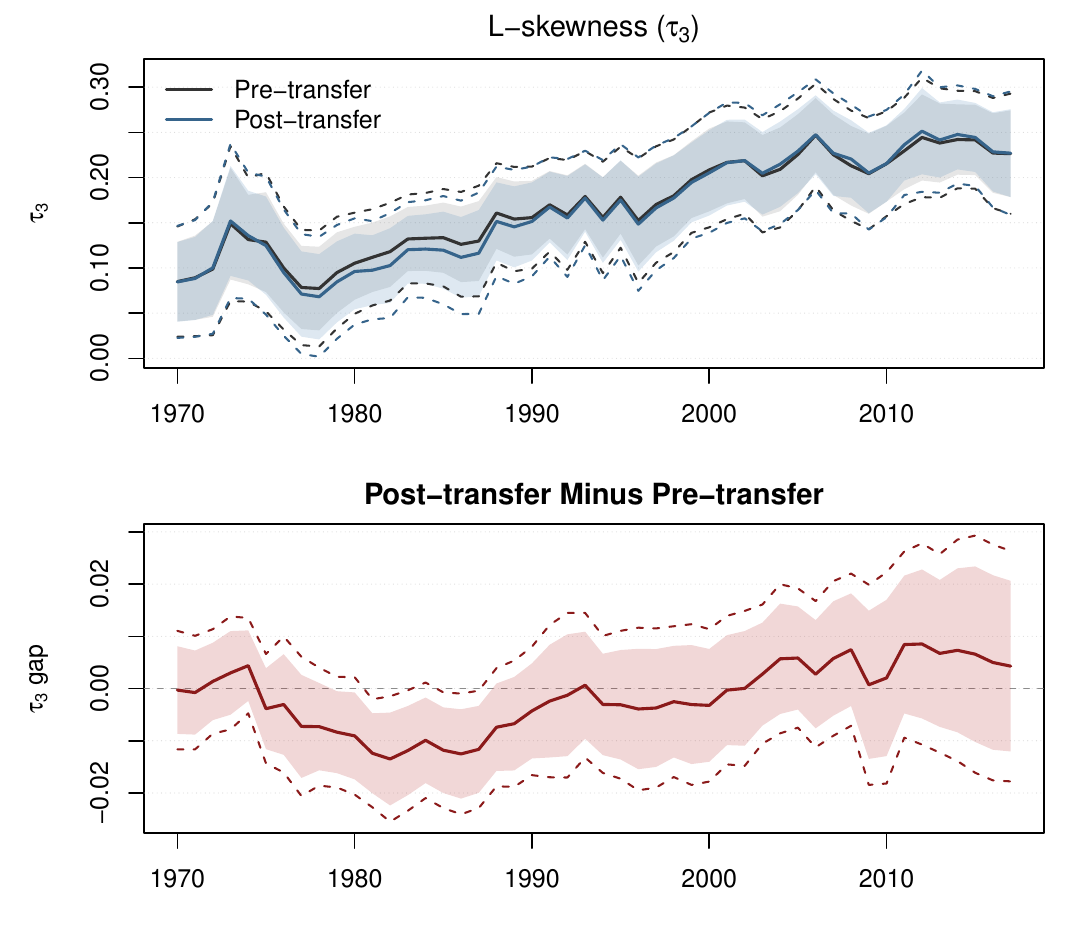}
\caption{L-skewness and transfer-induced gaps.}
\label{fig:park_shin_l3_gap}
\end{subfigure}
\label{fig:park_shin_gaps}
\figurenote{Panel (a) plots the pre-transfer and post-transfer cross-county Gini series in the upper sub-panel and post-transfer minus pre-transfer Gini in the lower sub-panel. Panel (b) plots the analogous L-skewness series and gap. Shaded bands are pointwise 95\% intervals using state-clustered empirical first-order projections; dashed bands are simultaneous 95\% sup-$t$ bands over calendar years, following the same sup-$t$ principle as in \citet{li_liao_2020_uniform_nonparametric}.}
\end{figure}

The pre-transfer cross-county Gini rises from $0.1495$ in 1970 to $0.1841$ in 2017. By contrast, the post-transfer Gini remains nearly flat, moving from $0.1310$ to $0.1308$ over the same period. The post-transfer gap is already negative in 1970, $-0.0185$, and widens to $-0.0539$ by 2017, with a state-clustered standard error of $0.0029$. L-skewness gives the order-3 counterpart on a scale-free metric, but the signal is much weaker: the post-transfer gap is essentially zero in 1970, at about $-0.0003$, equals about $0.0043$ in 2017, and has a 2017 state-clustered standard error of about $0.0083$. Transfers therefore clearly compress cross-county inequality, whereas the evidence for a systematic effect on right-tail asymmetry is weak once the third L-moment is standardised. This is consistent with the substantive conclusion of \citet{park_shin_2023_county_incomes} that transfers mitigate rising county-income inequality, while adding state-clustered uncertainty for annual $U$-statistic summaries. The simultaneous bands sharpen two features of the Gini figure. First, the zero-effect path is rejected uniformly: the upper dashed band for $\Delta_t$ is below zero in every year, with its maximum equal to $-0.0146$. Second, the transfer effect strengthens over time rather than remaining merely negative throughout the sample.

The L-skewness figure is more cautious at the order-3 level: the pre-transfer and post-transfer paths remain close, the gap changes sign over the sample, and the simultaneous bands do not support a uniform asymmetry effect analogous to the Gini result. Taken together, the two figures show that transfers compress the cross-county distribution, while the evidence for a systematic effect on distributional asymmetry is weaker once uncertainty is evaluated under arbitrary within-state dependence.

\subsection{Rank correlations in standardised financial returns}\label{Sect:Empirical_Tong_Hansen}

\citet{tong_hansen_2026_dynamic_factor_correlations} studies dynamic factor correlations in daily standardised returns for US equities and factor portfolios. That empirical design is well suited to the weak-dependence side of the paper. The observations form a long daily time series, while the substantive object is dependence among factors, sectors, and stock returns. The application turns from linear dependence to rank dependence and again pairs an order-2 statistic with an order-3 statistic under the same sampling scheme. Kendall's tau gives the order-2 benchmark related to \citet{dehling2017testing}, while Spearman's rho provides the higher-order counterpart through its order-3 $U$-statistic representation.

The data in \citet{tong_hansen_2026_dynamic_factor_correlations} contain $4,278$ daily observations from January 3, 2007 to December 29, 2023. The variables used here are the AR(1)-EGARCH standardised returns from those replication files for the Fama--French five factors, the momentum factor, nine sector SPDR ETF factors, and $323$ individual stocks. The cited paper models dynamic linear correlation and sparse idiosyncratic correlation matrices. The exercise instead tracks rank dependence and equips it with HAC inference under weak serial dependence.

For the rank-correlation extension, let $Z_t=(X_t,Y_t)$ denote a bivariate standardised-return series, where $X_t$ is the market factor and $Y_t$ is a sector factor. Write
\[
    \tau_K:=\E[\psi_2(Z_1,Z_2)]
    \quad\text{with}\quad
    \psi_2(z_1,z_2):=\mathrm{sign}(x_1-x_2)\mathrm{sign}(y_1-y_2),
\]
for Kendall's tau, and
\[
    \rho_S:=\E[\psi_3(Z_1,Z_2,Z_3)]
    \quad\text{with}\quad
    \psi_3(z_1,z_2,z_3):=\frac12\sum_{\pi\in S_3}
    \mathrm{sign}(x_{\pi(1)}-x_{\pi(2)})
    \mathrm{sign}(y_{\pi(1)}-y_{\pi(3)}),
\]
for Spearman's rho, where \(S_3\) denotes the six permutations of \(\{1,2,3\}\). This is the symmetrised order-3 kernel associated with the basic non-symmetric term \(\mathrm{sign}(x_1-x_2)\mathrm{sign}(y_1-y_3)\). Under bounded conditional crossing densities for the relevant pairwise differences, \Cref{lem:sign_crossing_bound} verifies the expected-H\"older branch with \(\alpha=1/2\) for these rank kernels. For both statistics, standard errors come from the empirical first-order projection and a Bartlett HAC long-run variance estimator. For an order-$k$ statistic, \(\widehat{\mathrm{Var}}(U_{n,k})=k^2\hat\sigma_n^2/n\), so \(\mathrm{SE}(U_{n,k})=k\hat\sigma_n/\sqrt n\).\footnote{Throughout the paper, HAC implementations use the Bartlett weight function \(\kappa(x)=(1-|x|)\indicator{|x|\leq1}\) with bandwidth \(h_n=\lfloor n^{1/4}\rfloor+1\).}

\begin{table}[h!]
\centering
\caption{Market-sector rank dependence in Tong--Hansen standardised returns, 2007--2023.}
\label{tab:tong_hansen_rank_hac}
\small
\begin{tabular}{lrrrrr}
\toprule
Sector & Pearson & Kendall $\hat\tau$ & SE & Spearman $\hat\rho_3$ & SE \\
\midrule
XLY & 0.8896 & 0.7032 & 0.0064 & 0.8781 & 0.0053 \\
XLK & 0.8906 & 0.7034 & 0.0067 & 0.8754 & 0.0057 \\
XLI & 0.8779 & 0.6831 & 0.0078 & 0.8589 & 0.0068 \\
XLF & 0.8383 & 0.6421 & 0.0086 & 0.8225 & 0.0083 \\
XLB & 0.8242 & 0.6206 & 0.0078 & 0.8069 & 0.0076 \\
XLV & 0.7661 & 0.5548 & 0.0097 & 0.7379 & 0.0106 \\
XLE & 0.6664 & 0.4787 & 0.0116 & 0.6560 & 0.0138 \\
XLP & 0.6769 & 0.4705 & 0.0108 & 0.6457 & 0.0129 \\
XLU & 0.4861 & 0.3179 & 0.0128 & 0.4528 & 0.0172 \\
\bottomrule
\end{tabular}

\tablenote{Each sector series is paired with the market factor \( \mathrm{MKT}\text{-}\mathrm{RF}_t \). Standard errors use empirical first-order projections and Bartlett HAC long-run variances with bandwidth \(h_n=\lfloor n^{1/4}\rfloor+1\).}
\end{table}

\Cref{tab:tong_hansen_rank_hac} shows that rank dependence is strongest between the market factor and Consumer Discretionary (XLY), Information Technology (XLK), Industrials (XLI), Financials (XLF), and Materials (XLB). The weakest market-sector link is Utilities (XLU), but it remains strongly positive. The Spearman order-3 estimates are numerically almost identical to conventional rank Spearman correlations, with a maximum full-sample gap of about $0.00013$ across the market-factor pairs, confirming that the higher-order representation recovers the familiar rank-correlation object. The empirical payoff is feasible HAC inference for rank-based $U$-statistic dependence measures in a weakly dependent daily time series, now for an order-3 statistic alongside the order-2 Kendall benchmark treated by \citet{dehling2017testing}.

\section{Monte Carlo Simulations}\label{Sect:Simulations}
This section reports two Monte Carlo studies calibrated to the empirical applications in \Cref{Sect:Empirical_Gini,Sect:Empirical_Tong_Hansen}. The clustered arm uses the order-2 Gini mean difference and the order-3 third L-moment, while the weak-dependence arm uses the order-2 Kendall statistic and the order-3 Spearman statistic. In each case, the goal is the same: to assess the corresponding feasible variance estimator through finite-sample dispersion, standard-error calibration, and test size. Both simulation arms use \(2000\) replications. For each design, the reported diagnostics are the normalised median absolute deviation (MAD, or MAD/IQR in the clustered arm), the ratio of the feasible standard error to the Monte Carlo standard deviation (SE/SD), and the empirical \(5\%\) rejection frequency.

\subsection{Clustered sampling: Gini mean difference and the third L-moment}\label{Sect:Sim_Cluster}
The clustered simulation uses the same order-2/order-3 kernel pair as the county-income application. The order-2 Gini mean-difference kernel is centred as
\[
    \varphi_2(w_1,w_2):=|w_1-w_2|-\theta_{\mathrm{GMD}},\qquad \theta_{\mathrm{GMD}}:=\E|W_1-W_2|,
\]
and the order-3 third L-moment kernel is centred as
\[
    \varphi_3(w_1,w_2,w_3):=\frac13\{\max(w_1,w_2,w_3)-2\operatorname{med}(w_1,w_2,w_3)+\min(w_1,w_2,w_3)\}-\theta_{L3},
\]
where \(\theta_{L3}:=\E[\{\max(W_1,W_2,W_3)-2\operatorname{med}(W_1,W_2,W_3)+\min(W_1,W_2,W_3)\}/3]\), with the \(W_a\)'s independent draws from the calibrated marginal law. The DGP is calibrated to the county-income application. Cluster sizes are chosen deterministically from the empirical cluster-size distribution subject to the growth restriction on \(m_n\) in \Cref{ass:clus_het}, and the outcome is generated by applying the empirical quantile function to
\(
    \Phi\{\sqrt\rho\,Z_{g(i)}+\sqrt{1-\rho}\,\epsilon_i\},\quad Z_g,\epsilon_i\overset{\mathrm{iid}}\sim\mathcal N(0,1),
\)
where \(\Phi\) is the standard-normal distribution function. Thus \(\rho\) tunes within-cluster dependence while preserving the calibrated marginal distribution. The resulting clusters are heterogeneous and unbalanced, chosen to respect the growth pattern in \Cref{ass:clus_het}.

The design parameters are \(G_n\in\{200,400,600\}\) and \(\rho\in\{0,0.25,0.5\}\), with the implied sample sizes chosen by the calibrated cluster-size rule. The population targets are \(\theta_{\mathrm{GMD}}\) and \(\theta_{L3}\), computed exactly under the calibrated marginal distribution. Because each kernel is centred at its target, the object under study is the centred \(U\)-statistic, its cluster-robust standard error, and the associated cluster-robust \(t\)-statistic for testing the known null value \(0\). To show what is lost by ignoring the sampling scheme, \Cref{tab:sim_cluster_coverage} also reports the corresponding iid standard-error ratio and iid rejection rate alongside the cluster-robust diagnostics.

\begin{table}[h!]
\centering
\caption{Clustered arm: finite-sample behaviour of the cluster-robust \(t\)-statistic.}
\label{tab:sim_cluster_coverage}
\scriptsize
\resizebox{\textwidth}{!}{%
\begin{tabular}{lrrrrrrrrrrr}
\toprule
\multicolumn{2}{c}{Design} & \multicolumn{5}{c}{GMD} & \multicolumn{5}{c}{Third L-moment} \\
\cmidrule(lr){1-2}\cmidrule(lr){3-7}\cmidrule(lr){8-12}
$(G_n,m_n^2/n)$ & $\rho$ & $\frac{\mathrm{MAD}}{\mathrm{IQR}}$ & $\frac{\mathrm{SE}}{\mathrm{SD}}_{\mathrm{rob}}$ & 5\%$_{\mathrm{rob}}$ & $\frac{\mathrm{SE}}{\mathrm{SD}}_{\mathrm{iid}}$ & 5\%$_{\mathrm{iid}}$ & $\frac{\mathrm{MAD}}{\mathrm{IQR}}$ & $\frac{\mathrm{SE}}{\mathrm{SD}}_{\mathrm{rob}}$ & 5\%$_{\mathrm{rob}}$ & $\frac{\mathrm{SE}}{\mathrm{SD}}_{\mathrm{iid}}$ & 5\%$_{\mathrm{iid}}$ \\
\midrule
$(200,0.097)$ & 0.00 & 0.036 & 0.973 & 0.076 & 0.975 & 0.075 & 0.014 & 0.952 & 0.098 & 0.954 & 0.097 \\
 & 0.25 & 0.042 & 0.929 & 0.091 & 0.867 & 0.104 & 0.016 & 0.900 & 0.118 & 0.875 & 0.117 \\
 & 0.50 & 0.046 & 0.965 & 0.081 & 0.792 & 0.133 & 0.016 & 0.912 & 0.108 & 0.808 & 0.135 \\
\midrule
$(400,0.070)$ & 0.00 & 0.026 & 0.992 & 0.051 & 0.995 & 0.051 & 0.010 & 0.983 & 0.066 & 0.986 & 0.065 \\
 & 0.25 & 0.027 & 0.950 & 0.074 & 0.869 & 0.100 & 0.010 & 0.943 & 0.086 & 0.907 & 0.091 \\
 & 0.50 & 0.032 & 0.965 & 0.079 & 0.754 & 0.145 & 0.012 & 0.958 & 0.089 & 0.816 & 0.129 \\
\midrule
$(600,0.054)$ & 0.00 & 0.019 & 1.019 & 0.054 & 1.020 & 0.054 & 0.007 & 1.016 & 0.063 & 1.018 & 0.064 \\
 & 0.25 & 0.022 & 0.981 & 0.058 & 0.888 & 0.083 & 0.008 & 0.969 & 0.066 & 0.923 & 0.075 \\
 & 0.50 & 0.026 & 0.965 & 0.064 & 0.738 & 0.155 & 0.010 & 0.945 & 0.081 & 0.790 & 0.132 \\
\bottomrule
\end{tabular}
}
\tablenote{The design tuple is \((G_n,m_n^2/n)\), where \(G_n\) is the number of clusters, \(n\) is the implied sample size, and \(m_n=\max_{g\leq G_n}n_g\) is the maximum cluster size. The three designs have implied \((n,m_n)=(507,7),(1149,9),(1852,10)\). GMD is the centred Gini mean difference (order 2) and the third L-moment is the centred order-3 kernel described in \Cref{Sect:Empirical_Gini}. \(\rho\) is the one-factor within-cluster dependence parameter. ``MAD/IQR'' is the Monte Carlo median absolute deviation of the centred statistic across replications, divided by the population interquartile range of the empirical county-income calibration, \(\mathrm{IQR}(W)=11{,}587.154\); original-unit MAD values are recovered by multiplying by this IQR. ``SE/SD'' is the ratio of the mean feasible standard error to the Monte Carlo standard deviation, with subscripts ``rob'' and ``iid'' denoting the cluster-robust and iid formulas. Each row uses \(2000\) replications; ``5\%'' is the empirical rejection rate of the corresponding two-sided \(5\%\) \(t\)-test.}
\end{table}

\Cref{tab:sim_cluster_coverage} shows a clear separation between the robust and iid procedures once within-cluster dependence is present. Under the cluster-robust formula, the SE/SD ratios remain close to one, ranging from \(0.929\) to \(1.019\) for the Gini mean difference and from \(0.900\) to \(1.016\) for the third L-moment, while the corresponding \(5\%\) rejection rates improve as \(G_n\) increases and \(m_n^2/n\) declines. Even under \(\rho=0\), where the robust and iid procedures target the same population variance, finite-sample differences remain because the robust estimator is cluster-aggregated. The distortion is most visible for the order-3 L-moment at \(G_n=200\), where the cluster-robust rejection rate is \(9.8\%\), then falls to \(6.6\%\) and \(6.3\%\) as \(G_n\) increases. The iid comparator looks similar only when \(\rho=0\). Once \(\rho>0\), it understates dispersion, with SE/SD ratios falling as low as \(0.738\) for the Gini mean difference and \(0.790\) for the third L-moment, and its rejection rates rise to about \(16\%\). As expected, the order-3 statistic is the more demanding object, but the cluster-robust procedure remains well behaved across the calibrated designs.

\subsection{Near-epoch dependence: Kendall and Spearman rank correlations}\label{Sect:Sim_NED}
The weak-dependence simulation uses the same order-2/order-3 rank-correlation pair as the Tong--Hansen application. Kendall's tau is the order-2 statistic, while Spearman's rho is computed through its order-3 \(U\)-statistic representation.

For each replication, a bivariate stationary Gaussian AR(1) process \(W_t:=(X_t,Y_t)\) is generated through
\(
    X_t=\rho X_{t-1}+\sqrt{1-\rho^2}\,u_t,\quad
    Y_t=\rho Y_{t-1}+\sqrt{1-\rho^2}\,v_t,
\)
where \((u_t,v_t)\) are i.i.d.\ bivariate normal innovations with correlation \(\alpha\). Thus \(\rho\) tunes serial dependence here, in direct parallel with the clustered design. The dependence levels are calibrated to the full-sample Tong--Hansen market-sector estimates from \Cref{Sect:Empirical_Tong_Hansen}: high dependence uses MKT--XLK, middle dependence uses MKT--XLE, and low dependence uses MKT--XLU. The design parameters are \(n\in\{400,800,1200\}\) and \(\rho\in\{0,0.25,0.50\}\). The null value being tested is the known Gaussian-copula population rank correlation,
\[
    \tau(\alpha)=\frac{2}{\pi}\arcsin(\alpha),\qquad
    \rho_S(\alpha)=\frac{6}{\pi}\arcsin(\alpha/2),
\]
so the reported rejection rates are size checks for the HAC-based Wald statistic rather than power calculations against independence. As in the clustered arm, \Cref{tab:sim_ned_rank} also includes an iid comparator to show the effect of ignoring serial dependence.

\begin{table}[H]
\centering
\caption{NED rank-correlation arm: finite-sample behaviour of the HAC Wald statistic.}
\label{tab:sim_ned_rank}
\scriptsize
\begin{tabular}{lrrrrrrrrrrrr}
\toprule
 &  &  &  & \multicolumn{4}{c}{Kendall} & \multicolumn{4}{c}{Spearman} \\
\cmidrule(lr){5-8}\cmidrule(lr){9-12}
$n$ & Scenario & $\rho$ & $\alpha$ & $\mathrm{MAD}$ & $\frac{\mathrm{SE}}{\mathrm{SD}}_{\mathrm{rob}}$ & 5\%$_{\mathrm{rob}}$ & 5\%$_{\mathrm{iid}}$ & $\mathrm{MAD}$ & $\frac{\mathrm{SE}}{\mathrm{SD}}_{\mathrm{rob}}$ & 5\%$_{\mathrm{rob}}$ & 5\%$_{\mathrm{iid}}$ \\
\midrule
400 & High & 0.00 & 0.891 & 0.0101 & 1.0148 & 0.060 & 0.053 & 0.0083 & 1.0084 & 0.065 & 0.059 \\
 & High & 0.25 & 0.891 & 0.0113 & 1.0094 & 0.049 & 0.059 & 0.0089 & 0.9956 & 0.062 & 0.068 \\
 & High & 0.50 & 0.891 & 0.0134 & 0.9527 & 0.069 & 0.111 & 0.0107 & 0.9407 & 0.073 & 0.114 \\
 & Middle & 0.00 & 0.666 & 0.0173 & 0.9925 & 0.051 & 0.047 & 0.0209 & 0.9897 & 0.053 & 0.050 \\
 & Middle & 0.25 & 0.666 & 0.0185 & 0.9631 & 0.067 & 0.070 & 0.0225 & 0.9576 & 0.070 & 0.070 \\
 & Middle & 0.50 & 0.666 & 0.0213 & 0.9435 & 0.071 & 0.130 & 0.0259 & 0.9451 & 0.072 & 0.128 \\
 & Low & 0.00 & 0.486 & 0.0196 & 0.9767 & 0.062 & 0.059 & 0.0270 & 0.9715 & 0.068 & 0.065 \\
 & Low & 0.25 & 0.486 & 0.0214 & 0.9720 & 0.059 & 0.065 & 0.0290 & 0.9672 & 0.062 & 0.073 \\
 & Low & 0.50 & 0.486 & 0.0253 & 0.9366 & 0.075 & 0.132 & 0.0349 & 0.9309 & 0.078 & 0.127 \\
\midrule
800 & High & 0.00 & 0.891 & 0.0080 & 0.9782 & 0.058 & 0.055 & 0.0064 & 0.9737 & 0.062 & 0.059 \\
 & High & 0.25 & 0.891 & 0.0077 & 1.0106 & 0.050 & 0.055 & 0.0063 & 0.9992 & 0.048 & 0.051 \\
 & High & 0.50 & 0.891 & 0.0093 & 0.9545 & 0.069 & 0.114 & 0.0073 & 0.9471 & 0.070 & 0.116 \\
 & Middle & 0.00 & 0.666 & 0.0124 & 1.0039 & 0.045 & 0.047 & 0.0152 & 0.9965 & 0.051 & 0.049 \\
 & Middle & 0.25 & 0.666 & 0.0128 & 1.0045 & 0.048 & 0.056 & 0.0156 & 1.0035 & 0.048 & 0.057 \\
 & Middle & 0.50 & 0.666 & 0.0151 & 0.9586 & 0.066 & 0.121 & 0.0183 & 0.9588 & 0.071 & 0.116 \\
 & Low & 0.00 & 0.486 & 0.0142 & 0.9974 & 0.054 & 0.055 & 0.0193 & 0.9960 & 0.058 & 0.058 \\
 & Low & 0.25 & 0.486 & 0.0152 & 0.9814 & 0.060 & 0.067 & 0.0206 & 0.9835 & 0.059 & 0.068 \\
 & Low & 0.50 & 0.486 & 0.0184 & 0.9451 & 0.061 & 0.119 & 0.0253 & 0.9423 & 0.060 & 0.123 \\
\midrule
1200 & High & 0.00 & 0.891 & 0.0062 & 0.9812 & 0.057 & 0.052 & 0.0050 & 0.9797 & 0.058 & 0.057 \\
 & High & 0.25 & 0.891 & 0.0063 & 0.9955 & 0.056 & 0.066 & 0.0050 & 0.9916 & 0.061 & 0.072 \\
 & High & 0.50 & 0.891 & 0.0076 & 0.9370 & 0.060 & 0.130 & 0.0060 & 0.9383 & 0.061 & 0.120 \\
 & Middle & 0.00 & 0.666 & 0.0103 & 0.9901 & 0.049 & 0.047 & 0.0125 & 0.9914 & 0.049 & 0.049 \\
 & Middle & 0.25 & 0.666 & 0.0102 & 0.9984 & 0.054 & 0.061 & 0.0122 & 0.9933 & 0.054 & 0.064 \\
 & Middle & 0.50 & 0.666 & 0.0129 & 0.9320 & 0.067 & 0.124 & 0.0155 & 0.9309 & 0.070 & 0.126 \\
 & Low & 0.00 & 0.486 & 0.0118 & 0.9839 & 0.057 & 0.056 & 0.0161 & 0.9822 & 0.056 & 0.053 \\
 & Low & 0.25 & 0.486 & 0.0120 & 0.9986 & 0.049 & 0.059 & 0.0164 & 0.9984 & 0.050 & 0.060 \\
 & Low & 0.50 & 0.486 & 0.0147 & 0.9470 & 0.065 & 0.121 & 0.0201 & 0.9476 & 0.065 & 0.117 \\
\bottomrule
\end{tabular}

\tablenote{Kendall is Kendall's tau (order 2) and Spearman is Spearman's rho (order 3). High, middle, and low dependence use MKT--XLK, MKT--XLE, and MKT--XLU, respectively; $\rho$ is the AR(1) persistence parameter and $\alpha$ is the calibrated contemporaneous correlation. Each row uses $2000$ replications; ``MAD'' is the Monte Carlo median absolute deviation of the centred statistic across replications, ``SE/SD'' is the ratio of the mean feasible standard error to the Monte Carlo standard deviation, and ``5\%'' is the empirical rejection rate of the corresponding two-sided $5\%$ Wald test. Subscripts ``rob'' and ``iid'' denote the HAC and iid formulas.}
\end{table}

\Cref{tab:sim_ned_rank} tells the same story for weak dependence. The HAC estimator remains well calibrated across dependence levels, with SE/SD$_{\mathrm{rob}}$ ratios between \(0.932\) and \(1.015\) for Kendall and between \(0.931\) and \(1.008\) for Spearman. The robust \(5\%\) rejection rates mostly stay in the \(0.045\) to \(0.078\) range, whereas the iid comparator again breaks down in the persistent designs: once \(\rho=0.5\), iid rejection rates move into the \(0.111\) to \(0.132\) range at \(n=400\) and remain clearly oversized even at \(n=1200\). Dispersion rises, as expected, in the smaller and more persistent designs, with the order-3 Spearman statistic again the more demanding object, but the HAC procedure remains stable across the calibrated designs.

\section{Conclusion}\label{Sect:Conclusion}

This paper develops a unified asymptotic theory and feasible inference framework for order-$k$ $U$-statistics under clustered and weakly dependent sampling. The limit theory delivers asymptotically linear representations, weak laws of large numbers, and central limit theorems for clustered, exact $m$-dependent, and near-epoch-dependent designs. Feasible inference is then developed through cluster-robust and HAC covariance estimators. This makes inference feasible for parameters of economic interest often cast as U-statistics, such as Gini coefficients and L-skewness parameters, under clustered and weakly dependent sampling schemes. The central organising idea is simple: after partitioning the sample into columns, the difficult non-linear remainder can be analysed through sampling-generic tuples of mutually independent units across the relevant vertical partitions, while same-cluster and lag collisions are counted and controlled explicitly. That step is common across sampling schemes; what remains sampling scheme-specific is the limit theory for the first-order projection and the way its covariance is estimated. The empirical and simulation results illustrate that this separation leads to usable cluster-robust and HAC procedures for both order-2 and higher-order $U$-statistics.

\vspace{0.35cm}

\noindent \texttt{Replication files:}
The replication package, including code, data, and output files, is available on the author's website.

\noindent \texttt{Declaration of AI use:}
During the preparation of this manuscript, the author used OpenAI's ChatGPT and Codex for research assistance. The author reviewed and edited all AI-assisted output and takes full responsibility for the content of the manuscript.

\printbibliography
\end{refsection}

\newpage
\setcounter{page}{1}
\begin{refsection}
	\appendix
	\renewcommand{\thetable}{S.\arabic{table}}
	\renewcommand{\thefigure}{S.\arabic{figure}}
	\renewcommand{\thesection}{S.\arabic{section}}
	\setcounter{equation}{0}
	\renewcommand{\theequation}{S.\arabic{equation}}
	
	\begin{center}
		\Large{\bf Supplementary Material for \\
        Limit Theory for U-Statistics under Clustered and Weakly Dependent Data}
	\end{center}
	\begin{center}
		Emmanuel Selorm Tsyawo
	\end{center}

The proofs follow their dependency order rather than the order of presentation in the main text. The appendix first presents the unified linked-tuple machinery in \Cref{app:linked_tuple_machinery}: \Cref{app:vertical_arrangement} gives the vertical partitioning device, and \Cref{app:abstract_linked_tuples} gives the linked-tuple formulation used by the clustered and exact $m$-dependent proofs. The clustered results are collected in \Cref{app:clustered_limit_proofs}, including the projection-transfer lemma first needed there. The exact $m$-dependent and NED results are collected in \Cref{app:weak_dependence_limit_proofs}, with the NED approximation bounds stated in \Cref{app:ned}. Feasible-inference proofs then follow in \Cref{app:inference_proofs}.

\section{Unified Linked-Tuple Machinery}\label{app:linked_tuple_machinery}

\subsection{Vertical partitioning}\label{app:vertical_arrangement}
A theoretically useful device for the proof is a partition of the index set into vertical columns. Let \(1\leq M_n\leq n\), set \(s_n:=\ceil{n/M_n}\), and define \(\kappa_n:=n-M_n(s_n-1)\). Arrange the ordered observations in \(s_n\) rows and \(M_n\) columns as follows:

\begin{align}\label{eqn:dat_arrange}
    \begin{bmatrix}
W_1 & \ldots & W_{\kappa_n} & W_{{\kappa_n}+1} & \ldots & W_{M_n}\\
\vdots & \ddots & \vdots & \vdots & \ddots & \vdots \\
W_{{M_n}(s_n-2)+1} & \ldots & W_{{M_n}(s_n-2)+{\kappa_n}} & W_{{M_n}(s_n-2)+{\kappa_n}+1} & \ldots &W_{{M_n}(s_n-1)} \\
W_{{M_n}(s_n-1)+1} & \ldots & W_n & \varnothing &\ldots& \varnothing
\end{bmatrix}
\end{align}
The symbols \( \varnothing \) are visual placeholders for empty slots and are not observations. The display partitions the observed indices into vertical sets \(\mathcal M_n(1),\ldots,\mathcal M_n(M_n)\), namely
\begin{align*}
\overbrace{
\begin{bmatrix}
        W_1 \\ W_{M_n+1} \\ \vdots \\ W_{M_n(s_n-2)+1} \\ W_{M_n(s_n-1)+1}
    \end{bmatrix}
    }^{\mathcal{M}_n(1)};\ \ldots
\overbrace{
    \begin{bmatrix}
        W_{\kappa_n} \\ W_{M_n+{\kappa_n}} \\ \vdots \\ W_{M_n(s_n-2)+{\kappa_n}} \\ W_n
    \end{bmatrix}}^{\mathcal{M}_n(\kappa_n)}; \
\overbrace{
    \begin{bmatrix}
        W_{{\kappa_n}+1}\\ W_{M_n+{\kappa_n}+1} \\ \vdots \\ W_{M_n(s_n-2)+{\kappa_n}+1}\\ \varnothing
    \end{bmatrix}}^{\mathcal{M}_n(\kappa_n+1)}; \ \ldots;
\overbrace{
    \begin{bmatrix}
        W_{M_n} \\ W_{2M_n} \\ \vdots \\ W_{M_n(s_n-1)} \\ \varnothing
    \end{bmatrix}}^{\mathcal{M}_n(M_n)}
\end{align*}
\noindent after ignoring any \( \varnothing \) placeholder. Denote the corresponding numbers of observed units by \(N_\ell:=\#\mathcal M_n(\ell)\), \(\ell\in[M_n]\). Then \(N_\ell=s_n\) if \(\ell\leq\kappa_n\) and \(N_\ell=s_n-1\) otherwise. By construction, \(n=\sum_{\ell=1}^{M_n}N_\ell\), \(M_n/n\leq1\), and \(s_n\leq2n/M_n\).

For clustered sampling, order clusters in decreasing order of size and take \(M_n=m_n:=\max_{g\in[G]}n_g\). This choice ensures that each vertical partition contains at most one observation from any cluster, so observations within a vertical partition are mutually independent under \Cref{ass:sampling}. Moreover, each observation has at most one same-cluster counterpart in any other vertical partition. For exact \(m\)-dependence, take \(M_n=m+1\) and use the residue classes modulo \(m+1\) as the vertical partitions. Observations within a vertical partition are then mutually independent, and an \(m\)-lag neighbourhood can meet any other vertical partition at most twice.

The above vertical partitioning is a proof device only; it does not alter the \(U\)-statistic, but exposes vertical partitions of mutually independent observations. This is useful in rendering the projection and CLT arguments below tractable.

\subsection{Unified linked-tuple formulation}\label{app:abstract_linked_tuples}
This subsection records the common order-$k$ decomposition used by the clustered and exact $m$-dependent sampling schemes below. The formulation is deliberately scheme-neutral: observations are arranged into vertical partitions, and a link relation records which pairs may be dependent. In clustered sampling, links mean shared cluster membership; under exact $m$-dependence, links mean being within $m$ lags. The unified formulation keeps the two features needed for the decomposition: observations in a vertical partition are mutually independent, while linked observations may be arbitrarily dependent. Scheme-specific proofs only need to specify the vertical partitions and the link relation, substitute the resulting counts, and then handle the first-order projection limit.

\paragraph{Tuple notation.}
For an ordered tuple \(\mathbf i=(i_1,\ldots,i_k)\in[N]^k\), \(\mathbf i_{\neq}\) means \(i_a\neq i_b\) for all \(a\neq b\), and such tuples are index-generic. Scheme-specific genericity is imposed through the link relation below: under clustered sampling, linked pairs share a cluster; under exact \(m\)-dependence, linked pairs lie within \(m\) lags.

\paragraph{Unified vertical-link setup.}
For each \(n\), let \(\{\mathcal M_n(\ell):1\leq\ell\leq M_n\}\) be a vertical partition of \([n]\), and write \(N_\ell:=\#\mathcal M_n(\ell)\leq s_n:=\ceil{n/M_n}\), with \(M_n/n\leq1\) as in the construction of \Cref{app:vertical_arrangement}. The observations in each \(\mathcal M_n(\ell)\) are mutually independent. Let \(\sim_n\) be a symmetric relation on \([n]\). For a finite set \(A\subset[n]\), write \(\mathcal F_A:=\sigma(W_i:i\in A)\). The relation is assumed to encode possible dependence in the following sense: if \(A_1,\ldots,A_q\) are disjoint finite subsets of \([n]\), and no element of \(A_p\) is linked to any element of \(A_{p'}\) whenever \(p\neq p'\), then \(\mathcal F_{A_1},\ldots,\mathcal F_{A_q}\) are mutually independent. Linked observations may otherwise be arbitrarily dependent. The two constructions used below satisfy the column-link count
\begin{equation}\label{eq:abstract_column_link_count}
    \#\{j\in\mathcal M_n(\ell):j\sim_n i\}\leq 2
    \qquad\text{for all }i\in[n]\text{ and }1\leq\ell\leq M_n.
\end{equation}
For clustered data, the vertical partitioning with \(M_n=m_n\) gives \(N_\ell\leq s_n\) and the sharper bound one in \eqref{eq:abstract_column_link_count}. For exact \(m\)-dependence, the residue classes modulo \(m+1\) give \(M_n=m+1\), and an \(m\)-lag window can meet any residue column at most once on each side.

Throughout this subsection write \(\mathbf i=(i_1,\ldots,i_k)\). An ordered distinct-index tuple is called \emph{linked-generic} if \(i_a\not\sim_n i_b\) for all \(a\neq b\). Otherwise it is \emph{non-generic}. For a linked-generic tuple, \(W_{i_1},\ldots,W_{i_k}\) are mutually independent by the preceding independence condition.

\paragraph{Projection notation.}
For a linked-generic tuple \(\mathbf i\) and \(S\subseteq[k]\), write \(c=|S|\), \(\mathbf i_S=(i_a:a\in S)\), and \(W_{\mathbf i_S}=(W_{i_a}:a\in S)\). Define
\[
    \pi_S^{\mathbf i}(w_S)
    :=
    \E[\varphi(W_{i_1},\ldots,W_{i_k})\mid W_{i_a}=w_a,\ a\in S],
    \qquad
    \pi_\emptyset^{\mathbf i}=0,
\]
where the kernel is centred. The tuple-specific Hoeffding projection is
\[
    \varphi_{\mathbf i,S}^{(c)}(w_S)
    :=
    \sum_{S'\subseteq S}(-1)^{|S\setminus S'|}\pi_{S'}^{\mathbf i}(w_{S'}).
\]
Hence, for every linked-generic tuple,
\begin{equation}\label{eq:abstract_hoeffding_reconstruction}
    \varphi(W_{i_1},\ldots,W_{i_k})
    =
    \sum_{S\subseteq[k]}\varphi_{\mathbf i,S}^{(|S|)}(W_{\mathbf i_S})
    \quad\text{a.s.},
\end{equation}
and, for every non-empty \(S\subseteq[k]\) and \(a\in S\),
\begin{equation}\label{eq:abstract_hoeffding_degeneracy}
    \E\!\left[
    \varphi_{\mathbf i,S}^{(|S|)}(W_{\mathbf i_S})
    \mid W_{\mathbf i_{S\setminus\{a\}}}
    \right]=0.
\end{equation}

For \(i\in[n]\), let
\[
    \mathcal A_{n,k}^{\sim}(i)
    :=
    \{(i_2,\ldots,i_k)_{\neq}:(i,i_2,\ldots,i_k)\text{ is linked-generic}\},
    \qquad
    Q_{n,i}^{\sim}:=\#\mathcal A_{n,k}^{\sim}(i).
\]
In the two dictionaries used below, \(Q_{n,i}^{\sim}>0\) for all sufficiently large \(n\). Define the linked-generic first-order average, for such \(n\), by
\[
    \bar\varphi_{n,k,i}^{\sim}(W_i)
    :=
    \frac{1}{Q_{n,i}^{\sim}}
    \sum_{\mathbf j\in\mathcal A_{n,k}^{\sim}(i)}
    \varphi_{\mathbf j}^{(1)}(W_i),
\]
where \(\varphi_{\mathbf j}^{(1)}(W_i)=\E[\varphi(W_i,W_{j_2},\ldots,W_{j_k})\mid W_i]\). Let \(R_{n,i}^{\sim}:=\Perm{n-1}{k-1}-Q_{n,i}^{\sim}\).

\begin{lemma}[Unified order-$k$ linked-tuple decomposition]\label{lem:abstract_linked_orderk_decomposition}
Let \( \displaystyle S_{n,k}:=\sum_{\mathbf i_{\neq}}\varphi(W_{i_1},\ldots,W_{i_k})\) under the unified vertical-link setup of the preceding paragraph.
\begin{align}\label{eq:abstract_linked_orderk_decomp}
    S_{n,k}
    &=
    \underbrace{k\Perm{n-1}{k-1}\sum_{i=1}^n\bar\varphi_{n,k,i}^{\sim}(W_i)}_{L_{n,k}^{\sim}:\ \text{linear projection}}
    \;-\;
    \underbrace{k\sum_{i=1}^nR_{n,i}^{\sim}\bar\varphi_{n,k,i}^{\sim}(W_i)}_{\mathcal C_{n,k}^{\sim}:\ \text{projection-count correction}}
    \notag\\
    &\quad
    +\underbrace{
    \sum_{\mathbf i_{\neq}\ \mathrm{linked\mbox{-}generic}}
    \sum_{\substack{S\subseteq[k]\\ |S|\ge2}}
    \varphi_{\mathbf i,S}^{(|S|)}(W_{\mathbf i_S})}_{\mathcal R_{n,k}^{\sim}:\ \text{higher-order Hoeffding remainder}}
    +
    \underbrace{
    \sum_{\mathbf i_{\neq}\ \mathrm{non\mbox{-}generic}}
    \varphi(W_{i_1},\ldots,W_{i_k})}_{\mathcal H_{n,k}^{\sim}:\ \text{linked-collision remainder}} .
\end{align}
\end{lemma}

\begin{proof}
Using the unified vertical partition,
\begin{align*}
    S_{n,k}
    &=
    \sum_{\ell_1=1}^{M_n}\cdots\sum_{\ell_k=1}^{M_n}
    \sum_{\substack{i_a\in\mathcal M_n(\ell_a),\,a\in[k]\\ \mathbf i_{\neq}}}
    \varphi(W_{i_1},\ldots,W_{i_k})\\
    &=
    \underbrace{
    \sum_{\ell_1=1}^{M_n}\cdots\sum_{\ell_k=1}^{M_n}
    \sum_{\substack{i_a\in\mathcal M_n(\ell_a),\,a\in[k]\\ \mathbf i_{\neq}\\ \mathrm{linked\mbox{-}generic}}}
    \varphi(W_{i_1},\ldots,W_{i_k})}_{G_{n,k}^{\sim}}
    +
    \underbrace{
    \sum_{\ell_1=1}^{M_n}\cdots\sum_{\ell_k=1}^{M_n}
    \sum_{\substack{i_a\in\mathcal M_n(\ell_a),\,a\in[k]\\ \mathbf i_{\neq}\\ \mathrm{non\mbox{-}generic}}}
    \varphi(W_{i_1},\ldots,W_{i_k})}_{\mathcal H_{n,k}^{\sim}} .
\end{align*}
For each fixed linked-generic tuple, the coordinates are mutually independent, so the tuple-specific Hoeffding decomposition in \eqref{eq:abstract_hoeffding_reconstruction} applies tuple by tuple. Hence the linked-generic part equals
\begin{align*}
    G_{n,k}^{\sim}
    &=
    \underbrace{
    \sum_{\ell_1=1}^{M_n}\cdots\sum_{\ell_k=1}^{M_n}
    \sum_{\substack{i_a\in\mathcal M_n(\ell_a),\,a\in[k]\\ \mathbf i_{\neq}\\ \mathrm{linked\mbox{-}generic}}}
    \sum_{a=1}^k\varphi_{\mathbf i_{-a}}^{(1)}(W_{i_a})}_{L_{n,\mathrm{gen},k}^{\sim}}
    +
    \underbrace{
    \sum_{\ell_1=1}^{M_n}\cdots\sum_{\ell_k=1}^{M_n}
    \sum_{\substack{i_a\in\mathcal M_n(\ell_a),\,a\in[k]\\ \mathbf i_{\neq}\\ \mathrm{linked\mbox{-}generic}}}
    \sum_{\substack{S\subseteq[k]\\ |S|\ge2}}
    \varphi_{\mathbf i,S}^{(|S|)}(W_{\mathbf i_S})}_{\mathcal R_{n,k}^{\sim}} .
\end{align*}
By symmetry of \(\varphi\), the linked-generic tuple domain is invariant under permutations of the coordinates, so every first-order position contributes the same value. Therefore
\[
    L_{n,\mathrm{gen},k}^{\sim}
    =
    k\underbrace{\sum_{i=1}^n
    \sum_{\mathbf j\in\mathcal A_{n,k}^{\sim}(i)}
    \varphi_{\mathbf j}^{(1)}(W_i)}_{T_{n,k}^{\sim}} .
\]
The linked-generic co-index count gives
\[
    T_{n,k}^{\sim}
    =
    \underbrace{\sum_{i=1}^nQ_{n,i}^{\sim}\bar\varphi_{n,k,i}^{\sim}(W_i)}_{B_{n,k}^{\sim}} .
\]
Finally, since \(R_{n,i}^{\sim}=\Perm{n-1}{k-1}-Q_{n,i}^{\sim}\),
\[
    B_{n,k}^{\sim}
    =
    \underbrace{\Perm{n-1}{k-1}\sum_{i=1}^n\bar\varphi_{n,k,i}^{\sim}(W_i)}_{L_{n,k}^{\sim}/k}
    -
    \underbrace{\sum_{i=1}^nR_{n,i}^{\sim}\bar\varphi_{n,k,i}^{\sim}(W_i)}_{\mathcal C_{n,k}^{\sim}/k}.
\]
Combining these displays gives \eqref{eq:abstract_linked_orderk_decomp}.
\end{proof}

\paragraph{Unified projection-count correction.}
The projection-count correction is governed by co-index tuples that fail linked-genericity after an anchor \(i\) is fixed.

\begin{lemma}[Unified projection-count correction]\label{lem:abstract_projection_count}
Under the unified vertical-link setup,
\begin{equation}\label{eq:abstract_Ri_count}
    0\leq R_{n,i}^{\sim}
    \leq
    2\binom{k}{2}M_n n^{k-2}
    \qquad\text{uniformly in }i.
\end{equation}
If, in addition, \(\sup_i\|\bar\varphi_{n,k,i}^{\sim}(W_i)\|_2\leq C_1\), then
\[
    \|\mathcal C_{n,k}^{\sim}\|_2
    \leq
    2k\binom{k}{2}C_1M_nn^{k-1}.
\]
\end{lemma}

\begin{proof}
The count \(R_{n,i}^{\sim}=\Perm{n-1}{k-1}-Q_{n,i}^{\sim}\) is the number of ordered co-index tuples \((i_2,\ldots,i_k)_{\neq}\) for which \((i,i_2,\ldots,i_k)\) is not linked-generic. Split this count as \(R_{n,i}^{\sim}\leq R_{n,i}^{(1)}+R_{n,i}^{(2)}\), where \(R_{n,i}^{(1)}\) counts co-index tuples with at least one co-index linked to the anchor \(i\), and \(R_{n,i}^{(2)}\) counts co-index tuples with no co-index linked to \(i\), but with an internal link among the co-indices.

For \(R_{n,i}^{(1)}\), fix a position \(b\in\{2,\ldots,k\}\) at which the co-index is linked to \(i\). By \eqref{eq:abstract_column_link_count}, the number of possible values across all vertical columns is at most \(2M_n\). The remaining \(k-2\) co-index positions can be filled in at most \(n^{k-2}\) ways. A union bound over the \(k-1\) possible positions gives
\[
    R_{n,i}^{(1)}
    \leq
    2(k-1)M_nn^{k-2}.
\]

For \(R_{n,i}^{(2)}\), fix a pair of co-index positions \(2\leq b<c\leq k\) at which an internal link occurs. There are at most \(n\) choices for the first coordinate in the pair and at most \(2M_n\) choices for the linked second coordinate. The remaining \(k-3\) co-index positions can be filled in at most \(n^{k-3}\) ways. A union bound over the \(\binom{k-1}{2}\) possible pairs gives
\[
    R_{n,i}^{(2)}
    \leq
    2\binom{k-1}{2}M_nn^{k-2}.
\]
Combining the two bounds gives \eqref{eq:abstract_Ri_count}. The displayed \(L_2\) bound follows from
\[
    \|\mathcal C_{n,k}^{\sim}\|_2
    \leq
    k\sum_{i=1}^nR_{n,i}^{\sim}\|\bar\varphi_{n,k,i}^{\sim}(W_i)\|_2
\]
and \eqref{eq:abstract_Ri_count}.
\end{proof}

\paragraph{Unified degenerate bound.}
The next lemma isolates the higher-order Hoeffding projection terms, i.e., the components with \(|S|\ge2\). It is the linked analogue of the column-tuple argument used by the scheme-specific order-$k$ proofs below.

\begin{lemma}[Unified linked degenerate variance]\label{lem:abstract_linked_degenerate_variance}
Fix \(2\leq c\leq k\) and \(S\subseteq[k]\) with \(|S|=c\). If the unified vertical-link setup holds and \(\sup_{\mathbf i_{\neq}}\E|\varphi(W_{i_1},\ldots,W_{i_k})|^r\leq C\) for some \(r\ge2\), then, for every fixed column tuple \((\ell_1,\ldots,\ell_k)\in[M_n]^k\),
\[
    \E\left[
    \left\{
    \sum_{\substack{i_a\in\mathcal M_n(\ell_a),\,a\in[k]\\ \mathbf i_{\neq}\ \mathrm{linked\mbox{-}generic}}}
    \varphi_{\mathbf i,S}^{(c)}(W_{\mathbf i_S})
    \right\}^2
    \right]
    \leq
    (2c)^c4^kC^{2/r}s_n^{2k-c}.
\]
\end{lemma}

\begin{proof}
Write \(S=\{1,\ldots,c\}\) without loss of generality, by relabelling positions.

\emph{Step 1 (uniform second moment).} By the standard projection formula underlying \eqref{eq:abstract_hoeffding_reconstruction}, \(\varphi_{\mathbf i,S}^{(c)}\) is a sum of \(2^c\leq2^k\) conditional expectations of \(\varphi\). Conditional Jensen's inequality and Lyapunov's inequality give
\[
    \sup_{\mathbf i_{\neq}\ \mathrm{linked\mbox{-}generic}}
    \E\big[(\varphi_{\mathbf i,S}^{(c)}(W_{\mathbf i_S}))^2\big]
    \leq
    4^kC^{2/r}
    =:C_k'.
\]

\emph{Step 2 (covariance vanishes below full active linkage).} Let \(\mathbf i,\mathbf i'\) be linked-generic tuples from the fixed column tuple, \(i_a,i_a'\in\mathcal M_n(\ell_a)\) for each \(a\in[k]\). Suppose some active position \(a_0\in S\) of \(\mathbf i\) is unlinked from every active position of \(\mathbf i'\). Since \(\mathbf i\) is linked-generic, \(i_{a_0}\) is also unlinked from every other active coordinate \(i_a\), \(a\in S\setminus\{a_0\}\). Thus the singleton sigma-field \(\mathcal F_{\{i_{a_0}\}}\) is independent of the sigma-field generated by \(\{W_{i_a}:a\in S\setminus\{a_0\}\}\) and \(\{W_{i_b'}:b\in S\}\). Combining this independence with \eqref{eq:abstract_hoeffding_degeneracy} gives
\[
    \E\!\left[
        \varphi_{\mathbf i,S}^{(c)}(W_{\mathbf i_S})
        \mid (W_{i_a})_{a\in S\setminus\{a_0\}},\,(W_{i_b'})_{b\in S}
    \right]
    =
    0.
\]
Consequently,
\[
\begin{aligned}
    \E\!\left[
        \varphi_{\mathbf i,S}^{(c)}(W_{\mathbf i_S})
        \varphi_{\mathbf i',S}^{(c)}(W_{\mathbf i'_S})
    \right]
    &=
    \E\!\left[
        \varphi_{\mathbf i',S}^{(c)}(W_{\mathbf i'_S})
        \E\!\left[
            \varphi_{\mathbf i,S}^{(c)}(W_{\mathbf i_S})
            \mid (W_{i_a})_{a\in S\setminus\{a_0\}},\,(W_{i_b'})_{b\in S}
        \right]
    \right]
    =0.
\end{aligned}
\]
The same argument applies symmetrically if some active position of \(\mathbf i'\) is unlinked from every active position of \(\mathbf i\).

\emph{Step 3 (counting surviving linked configurations).} The only covariance terms that may survive are therefore those for which every active position of \(\mathbf i'\) is linked to at least one active position of \(\mathbf i\), and symmetrically every active position of \(\mathbf i\) is linked to at least one active position of \(\mathbf i'\). It is enough to count the first requirement. The active coordinates of \(\mathbf i\) have at most \(\prod_{a\in S}N_{\ell_a}\leq s_n^c\) choices. Once these are fixed, each active coordinate \(i_b'\), \(b\in S\), must be linked to one of the \(c\) fixed active coordinates of \(\mathbf i\). By \eqref{eq:abstract_column_link_count}, each fixed active coordinate has at most two linked observations in the prescribed column \(\mathcal M_n(\ell_b)\). Hence \(i_b'\) has at most \(2c\) choices, and the active vector \((i_b')_{b\in S}\) has at most \((2c)^c\) choices. The inactive coordinates contribute at most \(\prod_{a\notin S}N_{\ell_a}\leq s_n^{k-c}\) choices for \(\mathbf i\) and at most \(s_n^{k-c}\) choices for \(\mathbf i'\). Thus the product-over-positions count for potentially nonzero covariance terms is bounded by
\[
    \prod_{a\in S}N_{\ell_a}
    \times (2c)^c
    \times \prod_{a\notin S}N_{\ell_a}
    \times \prod_{a\notin S}N_{\ell_a}
    \leq
    (2c)^cs_n^{2k-c}.
\]

\emph{Step 4 (assembly).} By Steps 1 and 3 and Cauchy--Schwarz, every surviving covariance is bounded in absolute value by \(C_k'=4^kC^{2/r}\), while Step 2 zeroes out every other pair. Therefore the second moment is at most \((2c)^cC_k's_n^{2k-c}\), as claimed.
\end{proof}

\begin{lemma}[Unified linked remainder controls]\label{lem:abstract_linked_remainder_controls}
Under the assumptions of \Cref{lem:abstract_linked_degenerate_variance}, write \(\mathcal R_{n,k}^{\sim}=\sum_{c=2}^k\mathcal R_{n,k,c}^{\sim}\), where
\[
    \mathcal R_{n,k,c}^{\sim}
    :=
    \sum_{\mathbf i_{\neq}\ \mathrm{linked\mbox{-}generic}}
    \sum_{\substack{S\subseteq[k]\\ |S|=c}}
    \varphi_{\mathbf i,S}^{(c)}(W_{\mathbf i_S}).
\]
Then
\[
    \|\mathcal R_{n,k,c}^{\sim}\|_2
    \leq
    \binom{k}{c}(2c)^{c/2}2^kC^{1/r}M_n^ks_n^{k-c/2},
    \qquad
    2\leq c\leq k,
\]
and therefore \(\|\mathcal R_{n,k}^{\sim}\|_2\leq\sum_{c=2}^k\binom{k}{c}(2c)^{c/2}2^kC^{1/r}M_n^ks_n^{k-c/2}\). The linked-collision remainder satisfies
\[
    \|\mathcal H_{n,k}^{\sim}\|_2
    \leq
    2\binom{k}{2}C^{1/r}M_nn^{k-1}.
\]
\end{lemma}

\begin{proof}
Fix \(2\leq c\leq k\). For each active set \(S\) with \(|S|=c\), decompose the corresponding sum into the \(M_n^k\) column tuples. By Minkowski's inequality and \Cref{lem:abstract_linked_degenerate_variance},
\[
    \left\|
    \sum_{\mathbf i_{\neq}\ \mathrm{linked\mbox{-}generic}}
    \varphi_{\mathbf i,S}^{(c)}(W_{\mathbf i_S})
    \right\|_2
    \leq
    \sum_{\ell_1=1}^{M_n}\cdots\sum_{\ell_k=1}^{M_n}
    (2c)^{c/2}2^kC^{1/r}s_n^{k-c/2}.
\]
Summing over the \(\binom{k}{c}\) choices of \(S\) gives the displayed bound for \(\mathcal R_{n,k,c}^{\sim}\), and summing over \(c=2,\ldots,k\) gives the bound for \(\mathcal R_{n,k}^{\sim}\).

For the collision term, every non-generic ordered \(k\)-tuple has at least one linked pair. Fix a pair of positions \(1\leq a<b\leq k\). There are at most \(n\) choices for \(i_a\). Once \(i_a\) is fixed, \eqref{eq:abstract_column_link_count} gives at most two linked choices for \(i_b\) in each vertical column, hence at most \(2M_n\) choices across all columns. The remaining \(k-2\) positions have at most \(n^{k-2}\) choices. A union bound over the \(\binom{k}{2}\) possible linked pairs gives at most \(2\binom{k}{2}M_nn^{k-1}\) non-generic ordered tuples. Therefore Minkowski's inequality and Lyapunov's inequality give
\[
    \|\mathcal H_{n,k}^{\sim}\|_2
    \leq
    \sum_{\mathbf i_{\neq}\ \mathrm{non\mbox{-}generic}}
    \|\varphi(W_{i_1},\ldots,W_{i_k})\|_2
    \leq
    2\binom{k}{2}C^{1/r}M_nn^{k-1}.
\]
This proves the lemma.
\end{proof}

\begin{proposition}[Unified asymptotically linear representation]\label{prop:abstract_linked_linear_representation}
Fix \(k\geq2\). Under the assumptions of \Cref{lem:abstract_linked_remainder_controls}, suppose also that the squared kernels are uniformly integrable over all distinct ordered \(k\)-tuples:
\[
    \lim_{M\to\infty}\sup_{\mathbf i_{\neq}}
    \E\!\left[
    |\varphi(W_{i_1},\ldots,W_{i_k})|^2
    \indicator{|\varphi(W_{i_1},\ldots,W_{i_k})|>M}
    \right]=0.
\]
Then
\begin{align}\label{eq:abstract_linked_linear_representation}
    U_{n,k}
    &=
    \frac{k}{n}\sum_{i=1}^n\bar\varphi_{n,k,i}^{\sim}(W_i)
    +
    \bigOp\!\left(\frac{M_n}{n}\right).
\end{align}
\end{proposition}

\begin{proof}
Divide the decomposition in \eqref{eq:abstract_linked_orderk_decomp} by \(\Perm nk\). The linear term becomes
\[
    \frac{L_{n,k}^{\sim}}{\Perm nk}
    =
    \frac{k\Perm{n-1}{k-1}}{\Perm nk}\sum_{i=1}^n\bar\varphi_{n,k,i}^{\sim}(W_i)
    =
    \frac{k}{n}\sum_{i=1}^n\bar\varphi_{n,k,i}^{\sim}(W_i).
\]
The uniform-integrability condition implies \(C_\varphi:=\sup_{\mathbf i_{\neq}}\|\varphi(W_{i_1},\ldots,W_{i_k})\|_2<\infty\). Conditional Jensen's inequality and Minkowski's inequality give
\[
    \sup_i\|\bar\varphi_{n,k,i}^{\sim}(W_i)\|_2
    \leq
    C_\varphi.
\]
By \Cref{lem:abstract_projection_count},
\[
    \frac{\mathcal C_{n,k}^{\sim}}{\Perm nk}
    =
    \bigOp\!\left(\frac{M_nn^{k-1}}{\Perm nk}\right)
    =
    \bigOp\!\left(\frac{M_n}{n}\right),
\]
because \(k\) is fixed and \(\Perm nk\asymp n^k\). The same argument and \Cref{lem:abstract_linked_remainder_controls} give
\[
    \frac{\mathcal H_{n,k}^{\sim}}{\Perm nk}
    =
    \bigOp\!\left(\frac{M_nn^{k-1}}{\Perm nk}\right)
    =
    \bigOp\!\left(\frac{M_n}{n}\right).
\]
Finally, for each \(2\leq c\leq k\), \Cref{lem:abstract_linked_remainder_controls} gives
\[
    \frac{\mathcal R_{n,k,c}^{\sim}}{\Perm nk}
    =
    \bigOp\!\left(\frac{M_n^ks_n^{k-c/2}}{\Perm nk}\right)
    =
    \bigOp\!\left(\left(\frac{M_n}{n}\right)^{c/2}\right),
\]
using \(s_n\leq2n/M_n\) from the vertical partition, \(k\) fixed, and \(\Perm nk\asymp n^k\).
Since \(k\) is fixed and \(M_n/n\leq1\), the sum over \(c=2,\ldots,k\) is dominated by its \(c=2\) term and is therefore \(\bigOp(M_n/n)\).
Combining these bounds gives \eqref{eq:abstract_linked_linear_representation}.
\end{proof}

\section{Proofs for Clustered Sampling}\label{app:clustered_limit_proofs}

\paragraph{Cluster dictionary.}

Throughout this section write \(\mathbf i=(i_1,\ldots,i_k)\). Specialise the unified link relation of \Cref{app:abstract_linked_tuples} to \(i\sim_n j\) iff \(g(i)=g(j)\). Then linked-generic tuples are precisely cluster-generic tuples, \( \displaystyle M_n=m_n:=\max_{g\in[G]}n_g\), and the column-link count in \eqref{eq:abstract_column_link_count} holds with the sharper constant one because each vertical column contains at most one observation from any cluster. For a generic tuple and \(S\subseteq[k]\), use the unified projections \(\varphi_{\mathbf i,S}^{(|S|)}\). At \(S=\{a\}\), abbreviate \(\varphi_{\mathbf i,\{a\}}^{(1)}(w_a)\) by \(\varphi^{(1)}_{\mathbf i_{-a}}(w_a)\), where \(\mathbf i_{-a}:=(i_b)_{b\in[k]\setminus\{a\}}\).

The unified H-decomposition in \eqref{eq:abstract_hoeffding_reconstruction}--\eqref{eq:abstract_hoeffding_degeneracy} specialises to cluster-generic tuples because their coordinates are mutually independent by \Cref{ass:sampling}; see also \citet[Sect.~1.6, Thm.~1]{lee1990ustatistics}. Thus
\begin{equation}\label{eq:hoeffding_reconstruction}
    \varphi(W_{i_1},\ldots,W_{i_k})
    =
    \sum_{S\subseteq[k]}\varphi_{\mathbf i,S}^{(|S|)}(W_{\mathbf i_S})
    \quad\text{a.s.},
\end{equation}
The alternating-sum definition of the Hoeffding projections also gives the usual degeneracy property:
\begin{equation}\label{eq:hoeffding_degeneracy}
    \E\!\left[
    \varphi_{\mathbf i,S}^{(|S|)}(W_{\mathbf i_S})
    \mid W_{\mathbf i_{S\setminus\{a\}}}
    \right]=0,
    \qquad \emptyset\neq S\subseteq[k],\ a\in S.
\end{equation}
Write \(R_i:=\Perm{n-1}{k-1}-Q_{n,i}\).
By \Cref{lem:abstract_projection_count}, with \(M_n=m_n\),
\begin{equation}\label{eq:orderk_Ri_count}
    R_i
    \leq
    2\binom{k}{2}m_nn^{k-2}.
\end{equation}

\subsection[Clustered order-k U-statistic]{Proof of \cref{thm:CLT_Ustats_Clus_k}}\label{app:orderk}

\paragraph{Part (a)} The clustered dictionary above identifies \(\bar\varphi_{n,k,i}^{\sim}(W_i)\) with \(\bar\varphi_{n,k}^{(1)}(W_i)\) and \(M_n\) with \(m_n\). The uniform integrability (UI) condition \eqref{eq:clus_ui_k} implies \(\sup_{\mathbf i_{\neq}}\|\varphi(W_{i_1},\ldots,W_{i_k})\|_2<\infty\); hence conditional Jensen's inequality and Minkowski's inequality give \(\sup_i\|\bar\varphi_{n,k}^{(1)}(W_i)\|_2<\infty\). Therefore \Cref{prop:abstract_linked_linear_representation} gives
\[
    U_{n,k}=\frac kn\sum_{i=1}^n\bar\varphi_{n,k}^{(1)}(W_i)+\bigOp\Big(\frac{m_n}{n}\Big),
\]
which is part (a) of \Cref{thm:CLT_Ustats_Clus_k}.

\paragraph{Part (b)}
The projection array \(\{\bar\varphi_{n,k}^{(1)}(W_i)\}_{i=1}^n\) is mean zero by generic centring and has uniformly bounded second moments by conditional Jensen's inequality. Cross-cluster independence gives
\[
    \E\Big[\Big(\frac{1}{n}\sum_{i=1}^n\bar\varphi_{n,k}^{(1)}(W_i)\Big)^2\Big]
    =
    \frac{1}{n^2}\sum_{g=1}^{G}\E\Big[\Big(\sum_{i\in\G_g}\bar\varphi_{n,k}^{(1)}(W_i)\Big)^2\Big]
    \leq
    C\frac{m_n}{n}.
\]
Hence \(n^{-1}\sum_i\bar\varphi_{n,k}^{(1)}(W_i)=\bigOp(\sqrt{m_n/n})\). Combining this with part (a), and using \(m_n/n\leq\sqrt{m_n/n}\), gives
\[
    U_{n,k}
    =
    \bigOp\Big(\sqrt{\frac{m_n}{n}}\Big).
\]
Since \(m_n/n\to0\) under \Cref{ass:clus_het}, the displayed rate is \(o_p(1)\).

\paragraph{Part (c)}
From part (a),
\[
    \sqrt n\,U_{n,k}
    =
    \frac{k}{\sqrt n}\sum_{i=1}^n\bar\varphi_{n,k}^{(1)}(W_i)
    +
    \bigOp\Big(\frac{m_n}{\sqrt n}\Big).
\]
Apply \citet[Theorem 2]{hansen2019asymptotic} to the triangular array \(\{\bar\varphi_{n,k}^{(1)}(W_i)\}_{i=1}^n\) with the same cluster assignment. Mean zero follows from generic centring. \Cref{lem:UI_transfer_projection} transfers uniform integrability from the kernel to the first projection family, and averaging over generic co-index tuples then transfers it to \(\bar\varphi_{n,k}^{(1)}\). The cluster variance appearing in \citet[Theorem 2]{hansen2019asymptotic} is \(\sigma_{n,k}^2/k^2\), and \(\sigma_{n,k}\geq\lambda>0\) gives non-degeneracy. Therefore
\[
    \frac{k}{\sigma_{n,k}}\cdot\frac1{\sqrt n}\sum_{i=1}^n\bar\varphi_{n,k}^{(1)}(W_i)\xrightarrow{d}\mathcal N(0,1).
\]
Combining this with part (a)'s remainder, which is $o_p(n^{-1/2})$ under \Cref{ass:clus_het} since $m_n/\sqrt n\to0$, via Slutsky's theorem gives $\sigma_{n,k}^{-1}\sqrt n\,U_{n,k}\xrightarrow{d}\mathcal N(0,1)$, which is part (c). \qed

\subsection[Clustered order-2 U-statistic]{Proof of \cref{thm:CLT_Ustats_Clus}}\label{app:order2}
\Cref{thm:CLT_Ustats_Clus} is the $k=2$ case of \Cref{thm:CLT_Ustats_Clus_k}. Indeed, \(Q_{n,i}=n-n_{g(i)}\), \(\bar\varphi_{n,2}^{(1)}(W_i)=\bar\varphi_n(W_i)\), and \(\sigma_{n,k}^2=\sigma_{n,2}^2\) when \(k=2\). The uniform integrability condition \eqref{eq:clus_ui_k} reduces to \eqref{eq:clus_ui_2}. Hence parts (a)--(c) of \Cref{thm:CLT_Ustats_Clus} follow directly from \Cref{thm:CLT_Ustats_Clus_k}. \qed

\subsection{Auxiliary lemmata for clustered sampling}

\subsubsection{Projection transfer}

\noindent The next lemma is used in the clustered-sampling CLT verification to transfer uniform integrability from the kernel to its first projection before averaging over cluster-generic co-indices. Its proof mirrors the proof of Lemma 1 in \citet{hansen2019asymptotic}. For each distinct \(k\)-tuple \((i,\mathbf j)\), where \(\mathbf j=(j_2,\ldots,j_k)\), let
\(
\varphi_{\mathbf j}^{(1)}(W_i)
:=
\E\!\left[\varphi(W_i,W_{j_2},\ldots,W_{j_k})\mid W_i\right].
\)
\begin{lemma}[Uniform-integrability transfer to the first projection]\label{lem:UI_transfer_projection}
Fix \(k\geq2\). Suppose that, for some \(p\geq1\), the kernel family has uniformly integrable \(p\)-th powers:
\begin{equation}\label{eqn:ass_unif_int}
	\lim_{M\to\infty}\sup_{(i,\mathbf j)_{\neq}}
	\E\!\left[
	|\varphi(W_i,W_{j_2},\ldots,W_{j_k})|^p
	\indicator{|\varphi(W_i,W_{j_2},\ldots,W_{j_k})|>M}
	\right]
	=0.
\end{equation}
Then the first-projection family has uniformly integrable \(p\)-th powers:
\[
    \lim_{M\to\infty}\sup_{(i,\mathbf j)_{\neq}}
    \E\!\left[
    |\varphi_{\mathbf j}^{(1)}(W_i)|^p
    \indicator{|\varphi_{\mathbf j}^{(1)}(W_i)|>M}
    \right]
    =0.
\]
\end{lemma}
\begin{proof}
By conditional Jensen,
\begin{equation}\label{eqn:cond_jen}
	|\varphi_{\mathbf j}^{(1)}(W_i)|^p
	=
	\Big|\E\!\left[\varphi(W_i,W_{j_2},\ldots,W_{j_k})\mid W_i\right]\Big|^p
	\leq
	\E\!\left[|\varphi(W_i,W_{j_2},\ldots,W_{j_k})|^p\mid W_i\right],
\end{equation}
thus by the LIE for a fixed \(M>0\),
\begin{align*}
    &\E\!\left[
    |\varphi_{\mathbf j}^{(1)}(W_i)|^p
    \indicator{|\varphi_{\mathbf j}^{(1)}(W_i)|>M}
    \right] \leq
    \E\!\left[
    |\varphi(W_i,W_{j_2},\ldots,W_{j_k})|^p
    \indicator{|\varphi_{\mathbf j}^{(1)}(W_i)|>M}
    \right].
\end{align*}
Split the right-hand side according to whether \(|\varphi(W_i,W_{j_2},\ldots,W_{j_k})|>\sqrt M\) or not:
\begin{align*}
    &\E\!\left[
    |\varphi(W_i,W_{j_2},\ldots,W_{j_k})|^p
    \indicator{|\varphi_{\mathbf j}^{(1)}(W_i)|>M}
    \right] \\
     &\leq
    \E\!\left[
    |\varphi(W_i,W_{j_2},\ldots,W_{j_k})|^p
    \indicator{|\varphi(W_i,W_{j_2},\ldots,W_{j_k})|>\sqrt M}
    \right] \\
    &\qquad\quad+
    \E\!\left[
    |\varphi(W_i,W_{j_2},\ldots,W_{j_k})|^p
    \indicator{|\varphi_{\mathbf j}^{(1)}(W_i)|>M}
    \indicator{|\varphi(W_i,W_{j_2},\ldots,W_{j_k})|\leq\sqrt M}
    \right].
\end{align*}
The second term is bounded by
\(
    M^{p/2}\Prob\!\left(|\varphi_{\mathbf j}^{(1)}(W_i)|>M\right).
\)
By Markov's inequality and conditional Jensen in \eqref{eqn:cond_jen} again,
\[
    \Prob\!\left(|\varphi_{\mathbf j}^{(1)}(W_i)|>M\right)
    \leq
    M^{-p}\E\!\left[|\varphi_{\mathbf j}^{(1)}(W_i)|^p\right]
    \leq
    M^{-p}\E\!\left[|\varphi(W_i,W_{j_2},\ldots,W_{j_k})|^p\right].
\]
Hence
\begin{align*}
    &\E\!\left[
    |\varphi_{\mathbf j}^{(1)}(W_i)|^p
    \indicator{|\varphi_{\mathbf j}^{(1)}(W_i)|>M}
    \right] \\
    &\leq
    \E\!\left[
    |\varphi(W_i,W_{j_2},\ldots,W_{j_k})|^p
    \indicator{|\varphi(W_i,W_{j_2},\ldots,W_{j_k})|>\sqrt M}
    \right]
    +
    M^{-p/2}
    \E\!\left[
    |\varphi(W_i,W_{j_2},\ldots,W_{j_k})|^p
    \right].
\end{align*}
	Taking the supremum over \((i,\mathbf j)_{\neq}\) in the preceding display, fix \(\varepsilon>0\). The assumption in \eqref{eqn:ass_unif_int} implies
	\[
	    \sup_{(i,\mathbf j)_{\neq}}
	    \E\!\left[
	    |\varphi(W_i,W_{j_2},\ldots,W_{j_k})|^p
	    \right]<\infty.
	\]
	Choose \(M\) sufficiently large that
	\[
	    M
	    \geq
	    \left(
	    \frac{
	    \sup_{(i,\mathbf j)_{\neq}}
	    \E\!\left[
	    |\varphi(W_i,W_{j_2},\ldots,W_{j_k})|^p
	    \right]}
	    {\varepsilon}
	    \right)^{2/p}
	\]
	and
	\[
	    \sup_{(i,\mathbf j)_{\neq}}
	    \E\!\left[
	    |\varphi(W_i,W_{j_2},\ldots,W_{j_k})|^p
	    \indicator{|\varphi(W_i,W_{j_2},\ldots,W_{j_k})|>\sqrt M}
	    \right]
	    \leq \varepsilon,
	\]
	where the second inequality is feasible by \eqref{eqn:ass_unif_int}. The first inequality gives
	\[
	    M^{-p/2}
	    \sup_{(i,\mathbf j)_{\neq}}
	    \E\!\left[
	    |\varphi(W_i,W_{j_2},\ldots,W_{j_k})|^p
	    \right]
	    \leq \varepsilon.
	\]
	Therefore, for this \(M\), and hence for all sufficiently large \(M\),
	\[
	    \sup_{(i,\mathbf j)_{\neq}}
	    \E\!\left[
	    |\varphi_{\mathbf j}^{(1)}(W_i)|^p
	    \indicator{|\varphi_{\mathbf j}^{(1)}(W_i)|>M}
	    \right]
	    \leq 2\varepsilon.
	\]
	Since \(\varepsilon>0\) is arbitrary,
	\[
	    \lim_{M\to\infty}\sup_{(i,\mathbf j)_{\neq}}
	    \E\!\left[
    |\varphi_{\mathbf j}^{(1)}(W_i)|^p
    \indicator{|\varphi_{\mathbf j}^{(1)}(W_i)|>M}
    \right]
    =0.
\]
\end{proof}

\section{Proofs for Weak Dependence}\label{app:weak_dependence_limit_proofs}

\paragraph{Temporal dictionary.}

This section maps exact \(m\)-dependence into the unified linked-tuple machinery. Relabel the residue classes \(\mathcal M_{n,m}(0),\ldots,\mathcal M_{n,m}(m)\) as \(M_n=m+1\) vertical partitions, with \(s=\ceil{n/(m+1)}\). Put \(i\sim_{n,m}j\) iff \(|i-j|\leq m\). Then linked-generic tuples are precisely lag-generic tuples, and the column-link count in \eqref{eq:abstract_column_link_count} holds because an \(m\)-lag window can meet any residue column at most twice. The exact-\(m\) decomposition and remainder rates below are therefore translations of \Cref{app:abstract_linked_tuples}. The temporal CLT is scheme-specific: it applies the blocking idea behind \citet[Theorem~4.6]{henze2024asymptotic} to the lag-generic first-order projection and then uses \citet[Theorem~2]{hansen2019asymptotic} on the retained block sums.

\subsection[Exact m-dependent order-k U-statistic]{Proof of \cref{thm:CLT_Ustats_mdep_k}}\label{app:mdepk}

\subsubsection{Projection notation and decomposition}
Write \( \displaystyle S_{n,m,k}:=\sum_{(i_1,\ldots,i_k)_{\neq}}\varphi(W_{i_1},\ldots,W_{i_k})\), so \(U_{n,k}=S_{n,m,k}/\Perm nk\). For each \(i\), let \(Q_i\) be the number of ordered co-index tuples \((i_2,\ldots,i_k)_{\neq}\) for which \((i,i_2,\ldots,i_k)\) is lag-generic, and set \(R_i:=\Perm{n-1}{k-1}-Q_i\), the number of excluded co-index tuples. Define
\[
    \bar\varphi_{n,m,k,i}^{(1)}(W_i)
    :=
    \frac1{Q_i}\sum_{\substack{(i_2,\ldots,i_k)_{\neq}\\ (i,i_2,\ldots,i_k)\ \mathrm{lag\mbox{-}generic}}}
    \varphi_{(i_2,\ldots,i_k)}^{(1)}(W_i),
\]
where \(\varphi_{(i_2,\ldots,i_k)}^{(1)}(W_i):=\E[\varphi(W_i,W_{i_2},\ldots,W_{i_k})\mid W_i]\).

Under the temporal dictionary \(i\sim_{n,m}j\) iff \(|i-j|\leq m\), lag-generic tuples are linked-generic tuples. Hence the unified Hoeffding reconstruction and the alternating-sum degeneracy relations in \eqref{eq:abstract_hoeffding_reconstruction}--\eqref{eq:abstract_hoeffding_degeneracy} apply directly.

For later bounds, write
\[
\mathcal R_{n,m,k,c}:=
\sum_{\boldsymbol\ell\in\{0,\ldots,m\}^k}
\sum_{\substack{i_a\in\mathcal M_{n,m}(\ell_a),\,a\in[k]\\ \mathbf i_{\neq}\\ \text{lag-gen}}}
\sum_{\substack{S\subseteq[k]\\ |S|=c}}
\varphi_{\mathbf i,S}^{(c)}(W_{\mathbf i_S})
\]
is the aggregate Hoeffding contribution of order \(c\) on the lag-generic part. Applying \Cref{lem:abstract_linked_orderk_decomposition} under the temporal dictionary gives the exact decomposition
\begin{align}\label{eq:mdep_orderk_master_decomp}
    S_{n,m,k}
    &=
    \underbrace{k\Perm{n-1}{k-1}\sum_{i=1}^n\bar\varphi_{n,m,k,i}^{(1)}(W_i)}_{L_{n,m,k}:\ \text{linear projection}}
    \;-\;\underbrace{k\sum_{i=1}^nR_i\bar\varphi_{n,m,k,i}^{(1)}(W_i)}_{\mathcal C_{n,m,k}:\ \text{projection-count correction}} \notag\\
    &\quad
    \;+\;\underbrace{
    \sum_{\boldsymbol\ell\in\{0,\ldots,m\}^k}
    \sum_{\substack{i_a\in\mathcal M_{n,m}(\ell_a),\,a\in[k]\\ \mathbf i_{\neq}\\ \text{lag-gen}}}
    \sum_{\substack{S\subseteq[k]\\ |S|\ge2}}
    \varphi_{\mathbf i,S}^{(|S|)}(W_{\mathbf i_S})
    }_{\mathcal R_{n,m,k}:\ \text{higher-order Hoeffding remainder}}
    \;+\;\underbrace{
    \sum_{\substack{(i_1,\ldots,i_k)_{\neq}\\ \text{non-gen}}}
    \varphi(W_{i_1},\ldots,W_{i_k})
    }_{\mathcal H_{n,m,k}:\ \text{lag-collision remainder}} .
\end{align}
For the bounds below, write \(\mathcal R_{n,m,k}=\sum_{c=2}^k\mathcal R_{n,m,k,c}\), where \(\mathcal R_{n,m,k,c}\) collects the order-\(c\) Hoeffding terms on the lag-generic part.
The remainder rates used below are obtained directly from \Cref{lem:abstract_projection_count,lem:abstract_linked_degenerate_variance,lem:abstract_linked_remainder_controls} under this temporal dictionary, with \(M_n=m+1\) and \(s_n=s=\ceil{n/(m+1)}\). This is the same degeneracy-and-linking mechanism as in \citet[pp.~858--859, Eqs~(6.5)--(6.10)]{janson2023asymptotic}: covariances without enough linked positions vanish, leaving only the linked configurations to count.

\paragraph{Part (a)} Under the temporal dictionary, \(\bar\varphi_{n,k,i}^{\sim}(W_i)=\bar\varphi_{n,m,k,i}^{(1)}(W_i)\) and \(M_n=m+1\). The UI condition \eqref{eq:kernel_ui_k} verifies the kernel condition in \Cref{prop:abstract_linked_linear_representation}; conditional Jensen's inequality and Minkowski's inequality give \(\sup_i\|\bar\varphi_{n,m,k,i}^{(1)}(W_i)\|_2<\infty\). Hence \Cref{prop:abstract_linked_linear_representation} gives
\[
    U_{n,k}
    =
    \frac kn\sum_{i=1}^n\bar\varphi_{n,m,k,i}^{(1)}(W_i)
    +
    \bigOp\Big(\frac{m+1}{n}\Big),
\]
which is part (a).

\paragraph{Part (b)} Since \(\E[\bar\varphi_{n,m,k,i}^{(1)}(W_i)]=0\), \(m\)-dependence gives
\begin{align}\label{eqn:cov_count}
    \E\Big[\Big(\frac1n\sum_{i=1}^n\bar\varphi_{n,m,k,i}^{(1)}(W_i)\Big)^2\Big]
    &=
    \frac1{n^2}\sum_{i=1}^n\sum_{i'=1}^n
    \mathrm{Cov}\big(\bar\varphi_{n,m,k,i}^{(1)}(W_i),\bar\varphi_{n,m,k,i'}^{(1)}(W_{i'})\big) \nonumber \\
    &\leq C\frac{m+1}{n},
\end{align}
because covariances vanish when \(|i-i'|>m\), leaving at most \(\bigO(n(m+1))\) nonzero covariance terms, and \Cref{lem:UI_transfer_projection} plus Jensen's inequality gives a uniform second-moment bound for the projection averages. Therefore \(n^{-1}\sum_i\bar\varphi_{n,m,k,i}^{(1)}(W_i)=\bigOp(\sqrt{(m+1)/n})\), and part (a) gives \(U_{n,k}=\bigOp(\sqrt{(m+1)/n})=o_p(1)\).

\paragraph{Part (c)}
From part (a),
\[
    \sqrt n\,U_{n,k}
    =
    \frac k{\sqrt n}\sum_{i=1}^n\bar\varphi_{n,m,k,i}^{(1)}(W_i)
    +
    \bigOp\Big(\frac{m+1}{\sqrt n}\Big).
\]
Since \(m\) is fixed, \((m+1)/\sqrt n\to0\), so it suffices to derive the limiting distribution of the linear projection. Let \(X_{n,i}:=\bar\varphi_{n,m,k,i}^{(1)}(W_i)\) and \(V_n:=\sum_{i=1}^nX_{n,i}\). The sequence \(\{X_{n,i}:1\leq i\leq n\}\) is centred and \(m\)-dependent. The blocking construction follows the proof of \citet[Theorem~4.6]{henze2024asymptotic}: retain long blocks and discard intervening \(m\)-gaps, so retained block sums are independent. Choose \(b_n=\floor{n^\beta}\) with \(0<\beta<(r-2)/\{2(r-1)\}\), which is possible because \(r>2\) by assumption. Let \(J_n:=\floor{n/(b_n+m)}\), and define retained blocks \(B_{n,j}:=\{(j-1)(b_n+m)+1,\ldots,(j-1)(b_n+m)+b_n\}\), \(j=1,\ldots,J_n\), separated by gaps of length \(m\). Write \(V_{n,B}:=\sum_{j=1}^{J_n}\sum_{i\in B_{n,j}}X_{n,i}\).

The block sums \(\Xi_{n,j}:=\sum_{i\in B_{n,j}}X_{n,i}\) are mutually independent by \(m\)-dependence. For any index set \(A\subseteq[n]\), the same covariance count as in \eqref{eqn:cov_count} gives \(\|\sum_{i\in A}X_{n,i}\|_2^2\leq C|A|(m+1)\). Let \(D_n\) be the set of discarded indices. The set \(D_n\) consists of the \(J_n\) gaps, with total gap size \(mJ_n\leq mn/(b_n+m)\leq mn/b_n\), and a terminal remainder of size less than \(b_n+m\). Hence \(\#D_n=\bigO(nm/b_n+b_n+m)\), and
\[
    \left\|\frac{V_n-V_{n,B}}{\sqrt n}\right\|_2^2
    =
    \frac1n\left\|\sum_{i\in D_n}X_{n,i}\right\|_2^2
    \leq
    \frac{C(m+1)}{n}\#D_n
    \leq
    C(m+1)\left(\frac{m}{b_n}+\frac{b_n}{n}+\frac{m}{n}\right)
    =
    o(1).
\]
Thus \((V_n-V_{n,B})/\sqrt n=o_p(1)\).
Since \(\sigma_{n,m,k}\geq\lambda>0\), this also yields
\[
    \frac{\V(V_{n,B})}{\V(V_n)}\to1,
    \qquad
    \V(V_n)=\frac{n\sigma_{n,m,k}^2}{k^2}.
\]

Apply \citet[Theorem~2]{hansen2019asymptotic} to the retained observations \(X_{n,i}\), grouped by the retained blocks \(B_{n,j}\). The blocks are independent, within-block dependence is unrestricted, and their constructed size \(b_n\) satisfies the same Hansen--Lee numerical size restrictions recorded in \Cref{ass:clus_het}. Indeed, \(b_n^2/n\to0\) because \(\beta<1/2\). Also, since \(J_n=\bigO(n/b_n)\),
\[
    \frac{(J_nb_n^r)^{2/r}}{n}
    =
    \bigO\!\left(n^{2/r-1}b_n^{2(r-1)/r}\right)
    =
    \bigO\!\left(n^{\{2-r+2\beta(r-1)\}/r}\right)
    =
    o(1),
\]
by \(\beta<(r-2)/\{2(r-1)\}\). The required uniform integrability for \(X_{n,i}\) follows from \Cref{lem:UI_transfer_projection} and Jensen's inequality over lag-generic co-indices. Therefore
\[
    \frac{kV_{n,B}}{\sqrt n\,\sigma_{n,m,k}}\xrightarrow{d}\mathcal N(0,1).
\]
Combining the retained-block CLT, the gap negligibility, and the negligible Hoeffding remainder via Slutsky's theorem gives
\[
    \sigma_{n,m,k}^{-1}\sqrt n\,U_{n,k}\xrightarrow{d}\mathcal N(0,1),
\]
which is part (c). \qed

\subsection[Proof of theorem 2: exact m-dependent order-2 U-statistic]{Proof of \cref{thm:CLT_Ustats_mdep}: exact $m$-dependent order-$2$ U-statistic}\label{app:mdep}

\Cref{thm:CLT_Ustats_mdep} is the $k=2$ case of \Cref{thm:CLT_Ustats_mdep_k}. The order-$k$ assumptions reduce directly to the order-$2$ assumptions stated in \Cref{Sect:FixedM}: the first-order projection average is \(\bar\varphi_{n,m,2,i}^{(1)}\), the fixed-$m$ variance is \(\sigma_{n,m,2}^2\), and \eqref{eq:kernel_ui_k} becomes \eqref{eq:kernel_ui_2}. Therefore \Cref{thm:CLT_Ustats_mdep_k}(a)--(c), evaluated at $k=2$, give
\[
    U_{n,2}=\frac2n\sum_{i=1}^n\bar\varphi_{n,m,2,i}^{(1)}(W_i)+\bigOp\Big(\frac{m+1}{n}\Big),\qquad
    U_{n,2}=\bigOp\Big(\sqrt{\frac{m+1}{n}}\Big)=o_p(1),
\]
and
\[
    \sigma_{n,m,2}^{-1}\sqrt n\,U_{n,2}\xrightarrow{d}\mathcal N(0,1).
\]
This proves \Cref{thm:CLT_Ustats_mdep}. \qed

\subsection[Near-epoch-dependent order-k U-statistic]{Proof of \cref{thm:CLT_double_limit}}\label{app:ned_approximation}

Throughout, write $m:=m_n$, $\nu_m:=\nu_{m_n}$, and $\rho_m:=\rho_{m_n}$ where no confusion arises. Recall $\{W_i^{(m)}\}$ is $m$-dependent (\Cref{lem:approx_mdep}).
In this NED proof, \(\varphi_i^{(1)}(W_i)\), \(\mu_{\mathbf i,m}\), and \(\varphi_i^{(1,m)}(W_i^{(m)})\) denote the corresponding date-specific independent-copy centring and first-projection objects.
The proof has the same projection-plus-remainder structure as the clustered and exact $m$-dependent arguments, but inserts two NED transfers. Namely, the statistic is written as the leading term \(k n^{-1}\sum_{i=1}^n\varphi_i^{(1)}(W_i)\), plus the statistic-level truncation error, the exact-$m$ higher-order Hoeffding remainder of the centred truncated statistic, and the projection-transfer error.

\paragraph{Part (a): asymptotic linear representation.} For \(\mathbf i=(i_1,\ldots,i_k)\), let
\[
    \mu_{\mathbf i,m}:=\E[\varphi(W_{i_1}^{(m)},\ldots,W_{i_k}^{(m)})],
    \qquad
    \bar\mu_{n,k}^{[m]}:=\frac1{\Perm nk}\sum_{\mathbf i_{\neq}}\mu_{\mathbf i,m},
\]
write \(U_{n,k}^{[m]}:=\Perm nk^{-1}\sum_{\mathbf i_{\neq}}\varphi(W_{i_1}^{(m)},\ldots,W_{i_k}^{(m)})\), and define
\[
    \varphi_i^{(1,m)}\big(W_i^{(m)}\big)
    :=
    \frac1{\Perm{n-1}{k-1}}
    \sum_{\substack{(i_2,\ldots,i_k)_{\neq}\\ i_a\in[n]\setminus\{i\},\,a=2,\ldots,k}}
    \E\!\left[
    \varphi(W_i^{(m)},W_{i_2}^{(m)},\ldots,W_{i_k}^{(m)})
    -
    \mu_{(i,i_2,\ldots,i_k),m}
    \mid W_i^{(m)}
    \right].
\]
Starting from the definition of \(U_{n,k}\), add and subtract the truncated statistic, its centring, and the truncated first projection to obtain
\begin{align*}
    U_{n,k}
    &=
    U_{n,k}^{[m]}
    +
    \underbrace{\big(U_{n,k}-U_{n,k}^{[m]}\big)}_{\mathcal T_{n,k}^{[m]}}\\
    &=
    \big(U_{n,k}^{[m]}-\bar\mu_{n,k}^{[m]}\big)
    +
    \underbrace{\bar\mu_{n,k}^{[m]}}_{\mathcal H_{n,k}^{[m]}}
    +
    \mathcal T_{n,k}^{[m]}\\
    &=
    \frac kn\sum_{i=1}^n\varphi_i^{(1,m)}\big(W_i^{(m)}\big)
    +
    \underbrace{\left\{
    U_{n,k}^{[m]}-\bar\mu_{n,k}^{[m]}-\frac kn\sum_{i=1}^n\varphi_i^{(1,m)}\big(W_i^{(m)}\big)
    \right\}}_{\mathcal R_{n,k}^{[m]}}
    +
    \mathcal H_{n,k}^{[m]}
    +
    \mathcal T_{n,k}^{[m]}\\
    &=
    \frac kn\sum_{i=1}^n\varphi_i^{(1)}(W_i)
    +
    \underbrace{\frac kn\sum_{i=1}^n\Big\{\varphi_i^{(1,m)}\big(W_i^{(m)}\big)-\varphi_i^{(1)}(W_i)\Big\}}_{\mathcal P_{n,k}^{[m]}}
    +\mathcal R_{n,k}^{[m]}
    +\mathcal H_{n,k}^{[m]}
    +\mathcal T_{n,k}^{[m]}.
\end{align*}
The third equality inserts the exact-\(m\) linear-plus-remainder decomposition for the centred truncated statistic \(U_{n,k}^{[m]}-\bar\mu_{n,k}^{[m]}\). This is \eqref{eq:mdep_orderk_master_decomp} applied to the centred truncated kernel along the exact-\(m\) approximating sequence supplied by \Cref{lem:approx_mdep}; the corresponding rates are obtained from \Cref{lem:abstract_projection_count,lem:abstract_linked_remainder_controls} under the temporal dictionary \(i\sim_{n,m}j\) iff \(|i-j|\leq m\).
It remains to show that these four labelled terms are all \(o_p(n^{-1/2})\).

\emph{The truncation term \(\mathcal T_{n,k}^{[m]}\).} Take a distinct tuple \(\mathbf i=(i_1,\ldots,i_k)\), write \(W_{\mathbf i}:=(W_{i_1},\ldots,W_{i_k})\) and \(W_{\mathbf i}^{(m)}:=(W_{i_1}^{(m)},\ldots,W_{i_k}^{(m)})\). Minkowski's inequality over the \(\Perm nk\) summands gives
\[
    \big\|U_{n,k}-U_{n,k}^{[m]}\big\|_2
    =
    \left\|
    \frac1{\Perm nk}\sum_{\mathbf i_{\neq}}
    \left\{\varphi(W_{\mathbf i})-\varphi(W_{\mathbf i}^{(m)})\right\}
    \right\|_2
    \leq
    \frac1{\Perm nk}\sum_{\mathbf i_{\neq}}
    \big\|\varphi(W_{\mathbf i})-\varphi(W_{\mathbf i}^{(m)})\big\|_2 .
\]
For each summand, \Cref{lem:effective_perturbation}(a) gives
\(
    \big\|\varphi(W_{\mathbf i})-\varphi(W_{\mathbf i}^{(m)})\big\|_2
    \leq
    Ck\rho_m .
\)
Hence \(\|U_{n,k}-U_{n,k}^{[m]}\|_2\leq Ck\rho_m\), so
\(
    \sqrt n\,\big\|U_{n,k}-U_{n,k}^{[m]}\big\|_2\leq Ck\sqrt n\,\rho_m=o(1)
\)
by $\sqrt n\,\rho_{m_n}=o(1)$. $L_2$ convergence to $0$ implies convergence in probability, so
\begin{equation*}
    \mathcal T_{n,k}^{[m]}=o_p(n^{-1/2}). \tag{A}
\end{equation*}

\emph{The centring term \(\mathcal H_{n,k}^{[m]}\).} The truncated kernel need not be exactly centred under its own independent-copy law. Since the original kernel is centred on the corresponding independent-copy law, \(\E[\varphi(W_{i_1},\ldots,W_{i_k})]=0\) for each distinct tuple. Hence the triangle inequality, Lyapunov's inequality, and \Cref{lem:effective_perturbation}(a) give
\[
\begin{aligned}
    |\mu_{\mathbf i,m}|
    &=
    \big|\E[\varphi(W_{i_1}^{(m)},\ldots,W_{i_k}^{(m)})
    -\varphi(W_{i_1},\ldots,W_{i_k})]\big|\\
    &\leq
    \|\varphi(W_{i_1}^{(m)},\ldots,W_{i_k}^{(m)})
    -\varphi(W_{i_1},\ldots,W_{i_k})\|_2
    \leq Ck\rho_m,
\end{aligned}
\]
uniformly over \(\mathbf i_{\neq}\). Therefore \(\sqrt n|\bar\mu_{n,k}^{[m]}|=o(1)\), and hence
\begin{equation*}
    \mathcal H_{n,k}^{[m]}=o(n^{-1/2}). \tag{B}
\end{equation*}

\emph{The exact-\(m\) remainder \(\mathcal R_{n,k}^{[m]}\).} The sequence $\{W_i^{(m)}\}$ satisfies \Cref{def:m_dep} by \Cref{lem:approx_mdep}, and the uniform integrability condition \eqref{eq:kernel_ui_k} holds uniformly in \(m\) for $\{\varphi(W_{i_1}^{(m)},\ldots,W_{i_k}^{(m)})\}$. The tuple-centred truncated summands inherit the required uniform integrability because \(\sup_{\mathbf i_{\neq}}|\mu_{\mathbf i,m}|=\bigO(\rho_m)\). Applying \Cref{lem:abstract_linked_remainder_controls} under the temporal dictionary, uniformly along \(m=m_n\), gives
\begin{equation*}
    \mathcal R_{n,k}^{[m]}=\bigOp\Big(\frac{m+1}{n}\Big)=o_p(n^{-1/2}), \tag{C}
\end{equation*}
using \(m+1=m_n+1=o(\sqrt n)\).

\emph{The projection-transfer term \(\mathcal P_{n,k}^{[m]}\).} Write, for each $i$,
\[
    \varphi_i^{(1,m)}\big(W_i^{(m)}\big)-\varphi_i^{(1)}(W_i)=\underbrace{\Big[\varphi_i^{(1,m)}\big(W_i^{(m)}\big)-\varphi_i^{(1)}\big(W_i^{(m)}\big)\Big]}_{(\mathrm I)}+\underbrace{\Big[\varphi_i^{(1)}\big(W_i^{(m)}\big)-\varphi_i^{(1)}(W_i)\Big]}_{(\mathrm{II})}.
\]
For (I), use the same coupled independent co-indices to define both projections. For a fixed ordered co-index tuple \(\mathbf j=(j_2,\ldots,j_k)\), let \((Y_a,Y_a^{(m)})\), \(a=2,\ldots,k\), be coupled generic original and truncated co-indices, independent of \(W_i^{(m)}\). The corresponding summand in (I) is
\[
    \E\!\left[
    \varphi\big(W_i^{(m)},Y_2^{(m)},\ldots,Y_k^{(m)}\big)
    -
    \varphi\big(W_i^{(m)},Y_2,\ldots,Y_k\big)
    \mid W_i^{(m)}
    \right]-\mu_{(i,\mathbf j),m}.
\]
The co-index replacement bound in \Cref{lem:effective_perturbation}(c), together with the centring bound used for \(\mathcal H_{n,k}^{[m]}\), gives
\[
    \|(\mathrm I)\|_2
    \leq
    C(k-1)\rho_m.
\]
For (II), with the same generic original co-indices \(Y_2,\ldots,Y_k\), the corresponding summand is
\[
    \E\!\left[
    \varphi\big(W_i^{(m)},Y_2,\ldots,Y_k\big)
    -
    \varphi\big(W_i,Y_2,\ldots,Y_k\big)
    \mid W_i^{(m)},W_i
    \right],
\]
so \Cref{lem:effective_perturbation}(b) gives $\|(\mathrm{II})\|_2\leq C\rho_m$.
Combining via Minkowski's inequality,
\[
    \big\|\varphi_i^{(1,m)}\big(W_i^{(m)}\big)-\varphi_i^{(1)}(W_i)\big\|_2\leq Ck\rho_m
\]
for every $i$, hence
\[
    \Big\|\frac k{\sqrt n}\sum_{i=1}^n\Big[\varphi_i^{(1,m)}\big(W_i^{(m)}\big)-\varphi_i^{(1)}(W_i)\Big]\Big\|_2\leq k^2C\sqrt n\,\rho_m=o(1)
\]
by $\sqrt n\,\rho_{m_n}=o(1)$, so
\begin{equation*}
    \mathcal P_{n,k}^{[m]}=o_p(n^{-1/2}). \tag{D}
\end{equation*}

\emph{Assembly for part (a).} Combining (A)--(D),
\[
    \mathcal T_{n,k}^{[m]}+\mathcal R_{n,k}^{[m]}+\mathcal H_{n,k}^{[m]}+\mathcal P_{n,k}^{[m]}
    =
    o_p(n^{-1/2}).
\]
Substituting this rate into the displayed decomposition gives
\begin{equation*}
    \sqrt n\,U_{n,k}=\frac k{\sqrt n}\sum_{i=1}^n\varphi_i^{(1)}(W_i)+o_p(1), \tag{E}
\end{equation*}
purely in terms of the original, untruncated sequence $\{W_i\}$; the proof sequence $m_n$ no longer appears. This proves part (a).

\paragraph{Part (b): stochastic order.} Since $\E[\varphi_i^{(1)}(W_i)]=0$, it suffices to show \(\sigma_n^2=\bigO(1)\). For large \(h\), set \(q=\floor{(h-1)/2}\), and approximate \(\varphi_i^{(1)}(W_i)\) and \(\varphi_{i+h}^{(1)}(W_{i+h})\) by their conditional expectations \(A_i\) and \(A_{i+h}\) on the \(q\)-window fields \(\mathcal F_{i-\floor{q/2}}^{i+\ceil{q/2}}\) and \(\mathcal F_{i+h-\floor{q/2}}^{i+h+\ceil{q/2}}\), respectively. These approximants, not the original projections, are functions of disjoint blocks of the i.i.d.\ innovation base and are therefore independent. By \Cref{lem:effective_perturbation}(b), \(\|\varphi_j^{(1)}(W_j)-A_j\|_2=\bigO(\rho_q)\), \(j\in\{i,i+h\}\), while \(\sup_i\|\varphi_i^{(1)}(W_i)\|_2<\infty\) by conditional Jensen. Since \(\operatorname{Cov}(A_i,A_{i+h})=0\), the covariance is the sum of \(\operatorname{Cov}(\varphi_i^{(1)}(W_i)-A_i,\varphi_{i+h}^{(1)}(W_{i+h}))\) and \(\operatorname{Cov}(A_i,\varphi_{i+h}^{(1)}(W_{i+h})-A_{i+h})\). Cauchy--Schwarz gives \(|\operatorname{Cov}(\varphi_i^{(1)}(W_i),\varphi_{i+h}^{(1)}(W_{i+h}))|\leq C\rho_q\leq Ch^{-\eta}\), uniformly in \(i\); the finitely many small lags are absorbed by the same \(L_2\) bound. Since \(\eta>1\), the covariance series is summable uniformly, and \(\sigma_n^2=\bigO(1)\). Chebyshev's inequality yields \(n^{-1/2}\sum_{i=1}^n\varphi_i^{(1)}(W_i)=\bigOp(1)\). Therefore \(n^{-1}\sum_{i=1}^n\varphi_i^{(1)}(W_i)=\bigOp(n^{-1/2})\), and part (a) gives \(U_{n,k}=\bigOp(n^{-1/2})=o_p(1)\).

\paragraph{Part (c): the central limit theorem.} Apply \citet[Cor.~24.7]{davidson1994stochastic} to the triangular array $X_{nt}:=\varphi_t^{(1)}(W_t)/(\sigma_n\sqrt n)$, $t=1,\ldots,n$, $n\geq1$, viewed as a functional of the full i.i.d.\ innovation base $\{\mathcal E_s\}$. For an arbitrary approximation width $q$, the finite-window fields $\mathcal F_{t-\floor{q/2}}^{t+\ceil{q/2}}$ enter only to verify the NED approximation condition:
\begin{enumerate}[(i)]
    \item $\E[X_{nt}]=0$ since the date-specific projection is centred. Also
    \[
        \V\Big[\sum_tX_{nt}\Big]
        =
        n^{-1}\sigma_n^{-2}\V\Big[\sum_{i=1}^n\varphi_i^{(1)}(W_i)\Big]
        =
        1
    \]
    by definition of \(\sigma_n^2\).
    \item (Domination.) Conditional Jensen gives $\E|\varphi_i^{(1)}(W_i)|^r\leq\E|\varphi(W_i,W_{i_2},\ldots,W_{i_k})|^r$, so the original-kernel part of \eqref{eq:kernel_ui_k} implies the uniform bound \(C_\varphi:=\sup_i\|\varphi_i^{(1)}(W_i)\|_r<\infty\). With $c_{nt}:=C_\varphi/(\sigma_n\sqrt n)$,
    \[
        \|X_{nt}/c_{nt}\|_r=\|\varphi_t^{(1)}(W_t)\|_r/C_\varphi\leq1
    \]
    for every $n,t$.
    \item (NED, and mixing of the base.) By \Cref{lem:effective_perturbation}(b), $\|\varphi_t^{(1)}(W_t)-\E[\varphi_t^{(1)}(W_t)\mid\mathcal F_{t-\floor{q/2}}^{t+\ceil{q/2}}]\|_2\leq C\rho_q$ for every $q$. Hence $X_{nt}$ is $L_2$-NED on $\{\mathcal E_s\}$ with $\|X_{nt}-\E[X_{nt}\mid\mathcal F_{t-\floor{q/2}}^{t+\ceil{q/2}}]\|_2\leq (C/(\sigma_n\sqrt n))\rho_q$. Since \(\rho_q=\bigO(q^{-\eta})\) with \(\eta>1\) by the rate condition in \Cref{thm:CLT_double_limit}, $X_{nt}$ is $L_2$-NED of size $-1$. The base $\{\mathcal E_t\}$ is i.i.d., hence strongly mixing of every size, in particular size $-r/(r-2)$.
    \item (Domination rate.) Let \(c_n^\ast:=\max_{1\leq t\leq n}c_{nt}\). Since \(c_n^\ast\leq C_\varphi/(\lambda\sqrt n)\), one has \(\sup_n n(c_n^\ast)^2\leq C_\varphi^2/\lambda^2<\infty\).
\end{enumerate}
Items (iii)--(iv) verify conditions (c$'$)--(d$'$) of \citet[Cor.~24.7]{davidson1994stochastic}, and (i)--(ii) match its conditions (a)--(b); the corollary gives $\sum_{t=1}^nX_{nt}\xrightarrow{d}\mathcal N(0,1)$, i.e.
\begin{equation*}
    \frac1{\sigma_n\sqrt n}\sum_{i=1}^n\varphi_i^{(1)}(W_i)\xrightarrow{d}\mathcal N(0,1). \tag{F}
\end{equation*}
Combining (E) and (F), and using \(\sigma_n\geq\lambda\), Slutsky's theorem gives
\[
    \frac{\sqrt n\,U_{n,k}}{k\sigma_n}
    =
    \frac1{\sigma_n\sqrt n}\sum_{i=1}^n\varphi_i^{(1)}(W_i)+o_p(1),
\]
and hence $(k\sigma_n)^{-1}\sqrt n\,U_{n,k}\xrightarrow{d}\mathcal N(0,1)$. \qed

\subsection[Near-epoch-dependent order-2 U-statistic]{Proof of \cref{thm:CLT_NED_2}: near-epoch-dependent order-$2$ U-statistic}\label{app:order2_ned}

\Cref{thm:CLT_NED_2} is the $k=2$ case of \Cref{thm:CLT_double_limit}. The assumptions reduce directly: \Cref{ass:kernel_perturb_k} becomes \Cref{ass:kernel_perturb_2} branch by branch, the order-$k$ projection \(\varphi_i^{(1)}\) is the ordinary order-$2$ first-order projection, and the multiplier $k$ becomes $2$. The uniform integrability condition is exactly the order-$2$ version stated in \Cref{thm:CLT_Ustats_mdep}, imposed both on the original kernel and on the finite-range approximants, with the additional \(r>2\) requirement inherited for the CLT conclusion. Therefore \Cref{thm:CLT_double_limit}(a)--(c), evaluated at $k=2$, give
\[
    U_{n,2}=\frac2n\sum_{i=1}^n\varphi_i^{(1)}(W_i)+o_p(n^{-1/2}),\qquad
    U_{n,2}=\bigOp(n^{-1/2})=o_p(1),
\]
and
\[
    (2\sigma_n)^{-1}\sqrt n\,U_{n,2}\xrightarrow{d}\mathcal N(0,1),
\]
where \(\sigma_n^2\) is the order-$2$ variance normalisation in \Cref{thm:CLT_NED_2}. \qed

\subsection{Auxiliary lemmata for weak dependence}\label{app:ned}

\begin{lemma}[Exact $m$-dependence of the approximating sequence]\label{lem:approx_mdep}
    Under \Cref{ass:ned}, $\{W_i^{(m)}\}_{i\geq1}$ is $m$-dependent in the sense of \Cref{def:m_dep}.
\end{lemma}
\begin{proof}
By definition,
\[
    W_i^{(m)}
    =
    \E[W_i\mid\mathcal F_{i-\floor{m/2}}^{i+\ceil{m/2}}],
\]
so \(W_i^{(m)}\) is measurable with respect to the sigma-field generated by the retained innovation block \((\mathcal E_{i-\floor{m/2}},\ldots,\mathcal E_{i+\ceil{m/2}})\). By the Doob--Dynkin measurability lemma, there is a measurable map \(h_{i,m}\) such that
\[
    W_i^{(m)}
    =
    h_{i,m}(\mathcal E_{i-\floor{m/2}},\ldots,\mathcal E_{i+\ceil{m/2}}).
\]
Two retained innovation blocks are disjoint whenever the corresponding observation indices differ by more than \(m\). Since the base sequence is i.i.d., the sigma-fields generated by such separated groups of approximants are independent. Hence \(\{W_i^{(m)}\}_{i\geq1}\) is \(m\)-dependent.
\end{proof}

\begin{lemma}[Effective perturbation bounds]\label{lem:effective_perturbation} 
    Suppose \Cref{ass:ned,ass:kernel_perturb_k} hold. Then, uniformly over relevant indices,
    \begin{enumerate}[(a)]
        \item the coordinatewise kernel perturbation and the induced $U$-statistic perturbation satisfy
        \[
            \big\|\varphi(W_{i_1},\ldots,W_{i_k})-\varphi(W_{i_1}^{(m)},\ldots,W_{i_k}^{(m)})\big\|_2=\bigO(\rho_m),
            \qquad
            \big\|U_{n,k}-U_{n,k}^{[m]}\big\|_2=\bigO(\rho_m);
        \]
        \item the first-order projection perturbation and the conditional NED projection approximation satisfy
        \[
            \big\|\varphi_i^{(1)}(W_i)-\varphi_i^{(1)}(W_i^{(m)})\big\|_2
            +
            \big\|\varphi_i^{(1)}(W_i)-\E[\varphi_i^{(1)}(W_i)\mid\mathcal F_{i-\floor{m/2}}^{i+\ceil{m/2}}]\big\|_2
            =\bigO(\rho_m);
        \]
        \item for coupled generic co-indices \(Y_a,Y_a^{(m)}\), \(a=2,\ldots,k\), independent of \(W_i^{(m)}\),
        \[
            \big\|\varphi(W_i^{(m)},Y_2^{(m)},\ldots,Y_k^{(m)})
            -\varphi(W_i^{(m)},Y_2,\ldots,Y_k)\big\|_2
            =\bigO(\rho_m).
        \]
    \end{enumerate}
\end{lemma}
\begin{proof}
\textbf{Expected-H\"older branch.}

\textbf{Part (a).}
Define the hybrid tuples \(W_{\mathbf i}^{(m,0)}:=(W_{i_1},\ldots,W_{i_k})\), \(W_{\mathbf i}^{(m,k)}:=(W_{i_1}^{(m)},\ldots,W_{i_k}^{(m)})\), and, for \(1\leq a\leq k\),
	\[
	    W_{\mathbf i}^{(m,a)}
	    :=
	    (W_{i_1}^{(m)},\ldots,W_{i_a}^{(m)},W_{i_{a+1}},\ldots,W_{i_k}).
	\]
	The full kernel difference telescopes as
	\[
	    \varphi(W_{\mathbf i})-\varphi(W_{\mathbf i}^{(m)})
	    =
	    \sum_{a=1}^k
	    \{\varphi(W_{\mathbf i}^{(m,a-1)})-\varphi(W_{\mathbf i}^{(m,a)})\}.
	\]
	Each summand replaces only coordinate \(a\). The expected-H\"older branch and \Cref{ass:ned} therefore give
	\[
	    \|\varphi(W_{\mathbf i}^{(m,a-1)})-\varphi(W_{\mathbf i}^{(m,a)})\|_2
	    \leq
	    C\|W_{i_a}-W_{i_a}^{(m)}\|_2^\alpha
	    \leq
	    C\nu_m^\alpha
	    =
	    \bigO(\rho_m).
	\]
	Minkowski's inequality and fixed \(k\) give the tuple bound in part (a). Averaging over the \(\Perm nk\) summands gives the \(U\)-statistic bound.

\textbf{Part (b).}
For the projection perturbation, let \(Y_2,\ldots,Y_k\) be independent generic co-indices, independent of \((W_i,W_i^{(m)})\). Then
\[
    \varphi_i^{(1)}(W_i)-\varphi_i^{(1)}(W_i^{(m)})
    =
    \E\!\left[
        \varphi(W_i,Y_2,\ldots,Y_k)-\varphi(W_i^{(m)},Y_2,\ldots,Y_k)
        \mid W_i,W_i^{(m)}
    \right].
\]
Conditional Jensen's inequality and the expected-H\"older branch give
\[
\begin{aligned}
	    \big\|\varphi_i^{(1)}(W_i)-\varphi_i^{(1)}(W_i^{(m)})\big\|_2
	    &\leq
	    \big\|\varphi(W_i,Y_2,\ldots,Y_k)-\varphi(W_i^{(m)},Y_2,\ldots,Y_k)\big\|_2\\
	    &\leq C\|W_i-W_i^{(m)}\|_2^\alpha
	    \leq C\nu_m^\alpha
	    =\bigO(\rho_m).
\end{aligned}
	\]
For the conditional NED projection approximation, set \(\mathcal F_i^{(m)}:=\mathcal F_{i-\floor{m/2}}^{i+\ceil{m/2}}\). Since \(W_i^{(m)}\) is \(\mathcal F_i^{(m)}\)-measurable, \(\varphi_i^{(1)}(W_i^{(m)})\) is also \(\mathcal F_i^{(m)}\)-measurable. Hence the \(L_2\)-projection property of conditional expectation gives
\[
    \big\|\varphi_i^{(1)}(W_i)-\E[\varphi_i^{(1)}(W_i)\mid\mathcal F_i^{(m)}]\big\|_2
    \leq
    \big\|\varphi_i^{(1)}(W_i)-\varphi_i^{(1)}(W_i^{(m)})\big\|_2
    =
    \bigO(\rho_m).
\]

\textbf{Part (c).}
Set \(Y^{(m,1)}:=(Y_2,\ldots,Y_k)\), \(Y^{(m,k)}:=(Y_2^{(m)},\ldots,Y_k^{(m)})\), and, for \(2\leq a\leq k\),
\[
    Y^{(m,a)}
    :=
    (Y_2^{(m)},\ldots,Y_a^{(m)},Y_{a+1},\ldots,Y_k).
\]
Then
\[
    \varphi(W_i^{(m)},Y_2,\ldots,Y_k)
    -
    \varphi(W_i^{(m)},Y_2^{(m)},\ldots,Y_k^{(m)})
    =
    \sum_{a=2}^k
    \{\varphi(W_i^{(m)},Y^{(m,a-1)})-\varphi(W_i^{(m)},Y^{(m,a)})\}.
\]
Each term is a one-coordinate replacement of a generic independent co-index, so the expected-H\"older branch gives \(\bigO(\rho_m)\) for each summand. Minkowski's inequality and fixed \(k\) give part (c).

\textbf{Data-dependent modulus branch.}
The same proof applies after replacing each one-coordinate perturbation by the following truncation bound. For a generic one-coordinate replacement, write
\[
    \Delta_m
    :=
    \varphi(W,W_2^\star,\ldots,W_k^\star)
    -
    \varphi(\widetilde W,W_2^\star,\ldots,W_k^\star),
\]
where \((W,\widetilde W)\) is either an original/truncated coordinate pair or a coupled generic co-index pair, and the remaining coordinates \(Z_{-a}:=(W_2^\star,\ldots,W_k^\star)\) are the relevant independent generic arguments. On \(\{L(W,\widetilde W,Z_{-a})\le M\}\), the data-dependent modulus branch of \Cref{ass:kernel_perturb_k} and \Cref{ass:ned} give \(\|\Delta_m\indicator{L(W,\widetilde W,Z_{-a})\le M}\|_2\le M\nu_m\). On \(\{L(W,\widetilde W,Z_{-a})>M\}\), the triangle inequality, \eqref{eq:random_modulus_kernel_moment_k}, Markov's inequality, and H\"older's inequality give \(\|\Delta_m\indicator{L(W,\widetilde W,Z_{-a})>M}\|_2=\bigO(M^{-\tau})\), uniformly over the same indices. Hence
\[
    \|\Delta_m\|_2\leq C(M\nu_m+M^{-\tau}).
\]
Optimising with \(M\asymp\nu_m^{-1/(1+\tau)}\) yields the uniform one-coordinate bound
\[
    \|\Delta_m\|_2=\bigO(\nu_m^{\tau/(1+\tau)})=\bigO(\rho_m).
\]
For part (a), insert this bound into the actual-coordinate telescoping display for \(\varphi(W_{\mathbf i})-\varphi(W_{\mathbf i}^{(m)})\), then apply Minkowski's inequality and use fixed \(k\). For part (b), insert the same one-coordinate bound inside the conditional-Jensen display for \(\varphi_i^{(1)}(W_i)-\varphi_i^{(1)}(W_i^{(m)})\), and then use the \(L_2\)-projection property of conditional expectation for the conditional NED projection approximation. For part (c), insert the same bound into the generic co-index telescoping display and apply Minkowski's inequality over the \(k-1\) co-index replacements. \qedhere
\end{proof}

\begin{lemma}[Crossing bound for sign kernels]\label{lem:sign_crossing_bound}
Let \(W_i=(X_i,Y_i)\) and let \(W_i^{(m)}=(X_i^{(m)},Y_i^{(m)})\). Suppose that, for \(V\in\{X,Y\}\), with \(A_{ij,V}:=V_i-V_j\) and \(D_{i,m,V}:=V_i^{(m)}-V_i\), the conditional distribution function
\(
    F_{ij,V}(a\mid W_j,D_{i,m,V})
    :=
    \Prob(A_{ij,V}\leq a\mid W_j,D_{i,m,V})
\)
is continuously differentiable in a neighbourhood of zero and has conditional density \(f_{ij,V}\) satisfying
\[
    \sup_{|a|\leq \bar a}f_{ij,V}(a\mid W_j,D_{i,m,V})\leq C
\]
almost surely, uniformly over the relevant indices, for some \(\bar a>0\). Then
\[
    \big\|\operatorname{sign}(V_i-V_j)-\operatorname{sign}(V_i^{(m)}-V_j)\big\|_2
    \leq
    C\|W_i-W_i^{(m)}\|_2^{1/2}.
\]
Consequently, Kendall-type products and finite symmetrisations of such sign comparisons, including the order-\(3\) Spearman kernel, satisfy the expected-H\"older branch of \Cref{ass:kernel_perturb_2,ass:kernel_perturb_k} with \(\alpha=1/2\).
\end{lemma}
\begin{proof}
It is enough to prove the claim for one sign comparison. Write \(A:=A_{ij,V}\), \(D:=D_{i,m,V}\), and \(F(a\mid W_j,D):=F_{ij,V}(a\mid W_j,D)\). For \(0\leq u\leq\bar a\), the mean-value theorem gives some \(\xi_u\in(-u,u)\) such that
\[
\begin{aligned}
    \Prob(|A|\leq u\mid W_j,D)
    &=
    F(u\mid W_j,D)-F(-u\mid W_j,D)
    =
    2u f(\xi_u\mid W_j,D)\\
    &\leq
    2u\sup_{|a|\leq u}f(a\mid W_j,D)
    \leq
    2Cu.
\end{aligned}
\]
For \(u>\bar a\), the trivial bound by one gives the same linear bound after replacing \(C\) by \(C\vee\bar a^{-1}\). Thus \(\Prob(|A|\leq u\mid W_j,D)\leq Cu\) for all \(u\geq0\). Since \(A_{ij,V}^{(m)}=A+D\), a sign can change only when \(A\) and \(A+D\) lie on opposite sides of zero. Hence, with the convention \(\operatorname{sign}(0)=0\),
\[
    |\operatorname{sign}(A)-\operatorname{sign}(A+D)|^2
    \leq
    4\indicator{|A|\leq |D|}.
\]
Taking conditional expectations and using the derived crossing bound at the random radius \(|D|\) gives
\[
\begin{aligned}
    \E|\operatorname{sign}(A)-\operatorname{sign}(A+D)|^2
    &\leq
    4\E\!\left[\Prob\big(|A|\leq |D|\mid W_j,D\big)\right]\\
    &\leq
    C\E|D|.
\end{aligned}
\]
Since \(|D|\leq \|W_i-W_i^{(m)}\|\), Lyapunov's inequality gives
\[
    \E|\operatorname{sign}(A)-\operatorname{sign}(A+D)|^2
    \leq
    C\|W_i-W_i^{(m)}\|_2.
\]
Taking square roots yields the sign-comparison bound. For Kendall products, use \(|ab-a'b'|\leq |a-a'|+|b-b'|\) for \(a,b,a',b'\in[-1,1]\). For finite symmetrisations, apply the same argument term by term and use Minkowski's inequality. \qedhere
\end{proof}

\section{Proofs for Feasible Inference}\label{app:inference_proofs}

\subsection{Cluster-robust projection covariance estimator}\label{app:cluster_var}
Throughout, write
\[
    \tilde\varphi_{n,k,i}
    :=
    \frac1{\Perm{n-1}{k-1}}
    \sum_{\substack{(i_2,\ldots,i_k)_{\neq}\\ i_a\in[n]\setminus\{i\},\ a=2,\ldots,k}}
    \varphi(W_i,W_{i_2},\ldots,W_{i_k})
\]
for the all-coindex projection average, and
\[
    \tilde\sigma_{n,k}^2:=\dfrac{k^2}{n}\sum_{g=1}^G\Big(\sum_{i\in\G_g}\bar\varphi_{n,k}^{(1)}(W_i)\Big)^2
\]
for the infeasible \emph{oracle} statistic. The proof separates the infeasible oracle statistic from the feasible all-coindex projection statistic: the decomposition below isolates the feasible-to-oracle gap, while the final assembly combines its negligibility with oracle consistency.
The estimator-to-target comparison starts from
\[
    \frac{\hat\sigma_{n,k}^2}{\sigma_{n,k}^2}-1
    =
    \left(\frac{\tilde\sigma_{n,k}^2}{\sigma_{n,k}^2}-1\right)
    +
    \frac{\hat\sigma_{n,k}^2-\tilde\sigma_{n,k}^2}{\sigma_{n,k}^2}.
\]
The first term is the oracle term; the second is the feasible-to-oracle gap.

\begin{lemma}[Cluster-robust variance decomposition]\label{lem:clus_var_decomp}
Let
\[
    A_g:=\sum_{i\in\G_g}\bar\varphi_{n,k}^{(1)}(W_i),\qquad
    D_g:=\sum_{i\in\G_g}\bigl(\tilde\varphi_{n,k,i}-\bar\varphi_{n,k}^{(1)}(W_i)\bigr),
    \qquad
    E_g:=D_g-n_gU_{n,k}.
\]
Then, for every \(n\),
\( \displaystyle 
    \hat\sigma_{n,k}^2-\tilde\sigma_{n,k}^2
    =
    \frac{k^2}{n}\sum_{g=1}^GE_g(2A_g+E_g).
\)
\end{lemma}
\begin{proof}
By the definition of \(\hat\varphi^{(1)}_{n,k,i}\) in \Cref{Sect:ClusterEstimation},
\[
    \hat\varphi^{(1)}_{n,k,i}=
    \frac1{\Perm{n-1}{k-1}}
    \sum_{\substack{(i_2,\ldots,i_k)_{\neq}\\ i_a\in[n]\setminus\{i\},\ a=2,\ldots,k}}
    \varphi(W_i,W_{i_2},\ldots,W_{i_k})
    -U_{n,k} = \tilde\varphi_{n,k,i} - U_{n,k}.
\]
Let
\[
    d_{n,k,i}
    :=
    \tilde\varphi_{n,k,i}-\bar\varphi_{n,k}^{(1)}(W_i),
\]
then adding and subtracting \(\bar\varphi_{n,k}^{(1)}(W_i)\) gives
\[
    \hat\varphi^{(1)}_{n,k,i}=
    \bar\varphi_{n,k}^{(1)}(W_i)+d_{n,k,i}-U_{n,k}.
\]
Therefore, for each cluster \(g\),
\[
    \sum_{i\in\G_g}\hat\varphi^{(1)}_{n,k,i}
    =
    \sum_{i\in\G_g}\bar\varphi_{n,k}^{(1)}(W_i)
    +
    \sum_{i\in\G_g}d_{n,k,i}
    -
    n_gU_{n,k}.
\]
With \(A_g,D_g,E_g\) as defined in the statement, this says \(\sum_{i\in\G_g}\hat\varphi^{(1)}_{n,k,i}=A_g+E_g\). Therefore
\[
    \hat\sigma_{n,k}^2-\tilde\sigma_{n,k}^2
    =
    \frac{k^2}{n}\sum_{g=1}^G\{(A_g+E_g)^2-A_g^2\}
    =
    \frac{k^2}{n}\sum_{g=1}^GE_g(2A_g+E_g),
\]
which proves the claim. \qedhere
\end{proof}

\subsubsection{Oracle consistency}

The oracle term uses the cluster-generic first-order projection and is handled directly by the Hansen--Lee cluster law of large numbers.

\begin{lemma}[Oracle cluster-robust consistency]\label{lem:clus_oracle_consistency}
Under the assumptions of \Cref{thm:CLT_Ustats_Clus_k}, including \(\sigma_{n,k}\geq\lambda>0\),
\[
    \tilde\sigma_{n,k}^2/\sigma_{n,k}^2\xrightarrow{p}1.
\]
\end{lemma}
\begin{proof}
The cluster sums \(\sum_{i\in\G_g}\bar\varphi_{n,k}^{(1)}(W_i)\) form a mean-zero triangular array under \Cref{ass:sampling}, because \(\bar\varphi_{n,k}^{(1)}\) averages only generic co-index tuples. Applying \citet[Theorem 3]{hansen2019asymptotic} to these sums, using \Cref{ass:clus_het} for the cluster-size condition, \Cref{lem:UI_transfer_projection} and Jensen's inequality for the transferred uniform integrability of \(\bar\varphi_{n,k}^{(1)}(W_i)\), and \(\sigma_{n,k}\geq\lambda>0\) for non-degeneracy, gives the claim. \qedhere
\end{proof}

\subsubsection{Feasible-to-oracle gap}

The preceding plug-in bound controls the cluster-level projection error. The next lemma converts that bound into negligibility of the feasible variance estimator relative to the oracle variance estimator.

\begin{lemma}[Feasible-to-oracle cluster-robust gap]\label{lem:clus_feasible_oracle_gap}
Under the assumptions of \Cref{thm:CLT_Ustats_Clus_k}, including \(\sigma_{n,k}\geq\lambda>0\),
\[
    \frac{\hat\sigma_{n,k}^2-\tilde\sigma_{n,k}^2}{\sigma_{n,k}^2}=o_p(1).
\]
\end{lemma}
\begin{proof}
Use the notation \(A_g,D_g,E_g\) from \Cref{lem:clus_var_decomp}. Since
\( \displaystyle 
    D_g
    =
    \sum_{i\in\G_g}\bigl(\tilde\varphi_{n,k,i}-\bar\varphi_{n,k}^{(1)}(W_i)\bigr)
    =
    \sum_{i\in\G_g}d_{n,k,i},
\)
\Cref{lem:clus_plugin_error} gives
\[
    \frac1n\sum_{g=1}^GD_g^2=o_p(1).
\]
Also, \(\sum_gn_g^2/n\leq m_n\) by \Cref{ass:clus_het}, and \Cref{thm:CLT_Ustats_Clus_k}(b) gives \(U_{n,k}^2=\bigOp(m_n/n)\). Therefore
\[
    \frac1n\sum_{g=1}^GE_g^2
    \leq
    \frac2n\sum_{g=1}^GD_g^2
    +
    2U_{n,k}^2\frac1n\sum_{g=1}^Gn_g^2
    =
    o_p(1)+\bigOp\Big(\frac{m_n^2}{n}\Big)
    =
    o_p(1).
\]
By \Cref{lem:clus_var_decomp} and Cauchy--Schwarz,
\[
    |\hat\sigma_{n,k}^2-\tilde\sigma_{n,k}^2|
    \leq
    k^2
    \left(\frac1n\sum_{g=1}^GE_g^2\right)^{1/2}
    \left(\frac1n\sum_{g=1}^G(2A_g+E_g)^2\right)^{1/2}.
\]
The first factor is \(o_p(1)\). For the second factor, \((2A_g+E_g)^2\leq8A_g^2+2E_g^2\). Since \(\tilde\sigma_{n,k}^2=k^2n^{-1}\sum_gA_g^2\), \Cref{lem:clus_oracle_consistency} implies \(n^{-1}\sum_gA_g^2=\bigOp(\sigma_{n,k}^2)\). Also, \(n^{-1}\sum_gE_g^2=o_p(1)\). Hence
\[
    \left(\frac1n\sum_{g=1}^G(2A_g+E_g)^2\right)^{1/2}
    =
    \bigOp(\sigma_{n,k}\vee1),
\]
and so \(\hat\sigma_{n,k}^2-\tilde\sigma_{n,k}^2=o_p(\sigma_{n,k}\vee1)\). Since \(\sigma_{n,k}\geq\lambda>0\), the displayed ratio is \(o_p(1)\). \qedhere
\end{proof}

\subsubsection[Proof of cluster-robust variance consistency]{Proof of \cref{thm:cluster_var_consistency}}
Using the estimator-to-target decomposition displayed at the start of this subsection,
\[
    \frac{\hat\sigma_{n,k}^2}{\sigma_{n,k}^2}-1
    =
    \left(\frac{\tilde\sigma_{n,k}^2}{\sigma_{n,k}^2}-1\right)
    +
    \frac{\hat\sigma_{n,k}^2-\tilde\sigma_{n,k}^2}{\sigma_{n,k}^2}.
\]
The first term is \(o_p(1)\) by \Cref{lem:clus_oracle_consistency}, and the second term is \(o_p(1)\) by \Cref{lem:clus_feasible_oracle_gap}. Hence \(\hat\sigma_{n,k}^2/\sigma_{n,k}^2\xrightarrow{p}1\). \qed

\subsubsection{Auxiliary lemmata}

The following technical lemma supplies the aggregate law of large numbers for the all-coindex projection error used in \Cref{lem:clus_feasible_oracle_gap}.

\begin{lemma}[All-coindex projection error]\label{lem:clus_plugin_error}
Under the assumptions of \Cref{thm:CLT_Ustats_Clus_k},
\[
    \frac1n\sum_{g=1}^G\left(\sum_{i\in\G_g}d_{n,k,i}\right)^2=o_p(1).
\]
\end{lemma}
\begin{proof}
Let
\[
    \mathcal A_i:=\{(i_2,\ldots,i_k)_{\neq}:g(i,i_2,\ldots,i_k)_{\neq}\},
    \qquad
    \mathcal B_i:=\{(i_2,\ldots,i_k)_{\neq}:i_a\in[n]\setminus\{i\},\ a=2,\ldots,k\},
\]
and write \(Q_{n,i}:=\#\mathcal A_i\).
By \eqref{eq:orderk_Ri_count}, \(R_i:=\Perm{n-1}{k-1}-Q_{n,i}=\bigO(m_nn^{k-2})\) uniformly in \(i\). Since \(Q_{n,i}=\Perm{n-1}{k-1}-R_i\), \(\Perm{n-1}{k-1}\asymp n^{k-1}\), and \(m_n/n\to0\), one has \(Q_{n,i}/\Perm{n-1}{k-1}=1+\bigO(m_n/n)\) and \(Q_{n,i}\asymp \Perm{n-1}{k-1}\asymp n^{k-1}\), uniformly in \(i\). For \(\mathbf j=(j_2,\ldots,j_k)\in\mathcal A_i\), write
\[
    \varphi_{\mathbf j}^{(1)}(W_i):=\E[\varphi(W_i,W_{j_2},\ldots,W_{j_k})\mid W_i].
\]
\eqref{eq:clus_ui_k} and conditional Jensen's inequality imply that the centred summands below have uniformly bounded second moments.

Define the cluster-generic empirical fluctuation
\[
    e_{n,k,i}
    :=
    \frac1{Q_{n,i}}\sum_{\mathbf j\in\mathcal A_i}
    \left\{
    \varphi(W_i,W_{j_2},\ldots,W_{j_k})-\varphi_{\mathbf j}^{(1)}(W_i)
    \right\}.
\]
Since
\begin{align*}
	\tilde\varphi_{n,k,i} &= \frac1{\Perm{n-1}{k-1}}
	\sum_{\substack{(i_2,\ldots,i_k)_{\neq}\\ i_a\in[n]\setminus\{i\},\ a=2,\ldots,k}}
	\varphi(W_i,W_{i_2},\ldots,W_{i_k}) \\
	&=
	\frac1{\Perm{n-1}{k-1}}\sum_{\mathbf j\in\mathcal A_i}
	\varphi(W_i,W_{j_2},\ldots,W_{j_k})
	+
	\frac1{\Perm{n-1}{k-1}}\sum_{\mathbf j\in\mathcal B_i\setminus\mathcal A_i}
	\varphi(W_i,W_{j_2},\ldots,W_{j_k})\\
	&=
	\frac{Q_{n,i}}{\Perm{n-1}{k-1}}
	\left[
	\frac1{Q_{n,i}}\sum_{\mathbf j\in\mathcal A_i}\varphi_{\mathbf j}^{(1)}(W_i)
	+
	\frac1{Q_{n,i}}\sum_{\mathbf j\in\mathcal A_i}
	\left\{\varphi(W_i,W_{j_2},\ldots,W_{j_k})-\varphi_{\mathbf j}^{(1)}(W_i)\right\}
	\right]\\
	&\quad+
	\frac1{\Perm{n-1}{k-1}}\sum_{\mathbf j\in\mathcal B_i\setminus\mathcal A_i}
	\varphi(W_i,W_{j_2},\ldots,W_{j_k})\\
	&=
	\frac{Q_{n,i}}{\Perm{n-1}{k-1}}
	\{\bar\varphi_{n,k}^{(1)}(W_i)+e_{n,k,i}\}
	+
	\frac1{\Perm{n-1}{k-1}}\sum_{\mathbf j\in\mathcal B_i\setminus\mathcal A_i}
	\varphi(W_i,W_{j_2},\ldots,W_{j_k}),
\end{align*}
subtracting \(\bar\varphi_{n,k}^{(1)}(W_i)\) and collecting the coefficient on this oracle projection gives
{\small
\begin{align}
    d_{n,k,i}
    &=
    \tilde\varphi_{n,k,i}-\bar\varphi_{n,k}^{(1)}(W_i)
    \notag\\
    &=
    \frac{Q_{n,i}}{\Perm{n-1}{k-1}}e_{n,k,i}
    +
    \left(
    \frac{Q_{n,i}}{\Perm{n-1}{k-1}}-1
    \right)\bar\varphi_{n,k}^{(1)}(W_i)
    \notag\\
    &\quad+
    \frac1{\Perm{n-1}{k-1}}\sum_{\mathbf j\in\mathcal B_i\setminus\mathcal A_i}
    \varphi(W_i,W_{j_2},\ldots,W_{j_k})\\
    &=
    \mathcal E_{n,k,i}-\mathcal D_{n,k,i}+\mathcal N_{n,k,i},
    \label{eq:clus_plugin_gap_decomp}
    \\
    \mathcal E_{n,k,i}
    &:=
    \frac{Q_{n,i}}{\Perm{n-1}{k-1}}e_{n,k,i}, \quad
    \mathcal D_{n,k,i}
    :=
    \frac{\Perm{n-1}{k-1}-Q_{n,i}}{\Perm{n-1}{k-1}}\bar\varphi_{n,k}^{(1)}(W_i),
    \notag\\
    \mathcal N_{n,k,i}
    &:=
    \frac1{\Perm{n-1}{k-1}}\sum_{\mathbf j\in\mathcal B_i\setminus\mathcal A_i}
    \varphi(W_i,W_{j_2},\ldots,W_{j_k}) .
    \notag
\end{align}
}

\noindent The rest of the proof studies the above summands in turn.

\paragraph{Empirical fluctuation.}
Fix \(i,i'\) in the same cluster \(h\). Conditional on the full reference cluster \(\{W_a:a\in\G_h\}\), the admissible co-index tuples use only clusters outside \(h\). If the co-index cluster sets in \(\mathbf j\in\mathcal A_i\) and \(\mathbf j'\in\mathcal A_{i'}\) are disjoint, cross-cluster independence gives
\begin{align*}
    &\E\Big[
    \left\{\varphi(W_i,W_{j_2},\ldots,W_{j_k})-\varphi_{\mathbf j}^{(1)}(W_i)\right\}
    \left\{\varphi(W_{i'},W_{j'_2},\ldots,W_{j'_k})-\varphi_{\mathbf j'}^{(1)}(W_{i'})\right\}
    \,\Big|\, \{W_a:a\in\G_h\}\Big]
    =
    0,
\end{align*}
because each conditional mean is zero after integrating over its own co-index clusters by \Cref{ass:sampling}. Thus only pairs of co-index tuples sharing at least one cluster can contribute to \(\E[e_{n,k,i}e_{n,k,i'}]\).

For each fixed \(\mathbf j\), the number of \(\mathbf j'\in\mathcal A_{i'}\) sharing at least one co-index cluster with \(\mathbf j\) is \(\bigO(m_nn^{k-2})\): one of the \(k-1\) positions of \(\mathbf j'\) is assigned to one of the \(k-1\) clusters used by \(\mathbf j\), giving at most \(m_n\) choices for that coordinate, and the remaining \(k-2\) coordinates have at most \(n^{k-2}\) choices. Expanding \(\E[e_{n,k,i}e_{n,k,i'}]\) gives the denominator \(Q_{n,i}Q_{n,i'}\), while at most \(Q_{n,i}\bigO(m_nn^{k-2})\) tuple pairs can contribute. Therefore, using Cauchy--Schwarz and \(Q_{n,i}\asymp n^{k-1}\),
\[
    \big|\E[e_{n,k,i}e_{n,k,i'}]\big|
    \leq
    \frac{C\,Q_{n,i}\,m_nn^{k-2}}{Q_{n,i}Q_{n,i'}}
    \leq C\frac{m_n}{n}
\]
uniformly over \(i,i'\) in the same cluster. Hence
\[
    \frac1n\sum_{g=1}^G\E\left[\left(\sum_{i\in\G_g}e_{n,k,i}\right)^2\right]
    \leq
    \frac{C m_n}{n^2}\sum_{g=1}^Gn_g^2
    \leq C\frac{m_n^2}{n}\to0,
\]
using \(\sum_gn_g^2\leq m_nn\) and \Cref{ass:clus_het}. Markov's inequality gives the corresponding \(o_p(1)\) cluster-aggregate bound for \(e_{n,k,i}\).

The same conclusion holds for \(\mathcal E_{n,k,i}\) in \eqref{eq:clus_plugin_gap_decomp}, because \(Q_{n,i}/\Perm{n-1}{k-1}\leq1\).

\paragraph{Denominator correction.}
For \(\mathcal D_{n,k,i}\) in \eqref{eq:clus_plugin_gap_decomp}, \eqref{eq:orderk_Ri_count} gives \((\Perm{n-1}{k-1}-Q_{n,i})/\Perm{n-1}{k-1}=\bigO(m_n/n)\) uniformly in \(i\), so Cauchy--Schwarz and the oracle covariance bound give
\[
    \frac1n\sum_{g=1}^G
    \left[
    \sum_{i\in\G_g}
    \frac{\Perm{n-1}{k-1}-Q_{n,i}}{\Perm{n-1}{k-1}}\bar\varphi_{n,k}^{(1)}(W_i)
    \right]^2
    \leq
    \bigO\left(\frac{m_n^2}{n^2}\right)
    \frac1n\sum_{g=1}^G
    \left(\sum_{i\in\G_g}\bar\varphi_{n,k}^{(1)}(W_i)\right)^2
    =
    o_p(1).
\]

\paragraph{Non-generic complement.}
For \(\mathcal N_{n,k,i}\) in \eqref{eq:clus_plugin_gap_decomp}, \(\#(\mathcal B_i\setminus\mathcal A_i)=R_i\), so \eqref{eq:orderk_Ri_count} gives \(\#(\mathcal B_i\setminus\mathcal A_i)/\Perm{n-1}{k-1}=\bigO(m_n/n)\) uniformly in \(i\). Jensen's inequality and the uniform second-moment bound give
\[
    \E\left[
    \left(
    \frac1{\Perm{n-1}{k-1}}\sum_{\mathbf j\in\mathcal B_i\setminus\mathcal A_i}
    \varphi(W_i,W_{j_2},\ldots,W_{j_k})
    \right)^2
    \right]
    \leq
    C\left(\frac{\#(\mathcal B_i\setminus\mathcal A_i)}{\Perm{n-1}{k-1}}\right)^2
    \leq C\frac{m_n^2}{n^2}.
\]
Therefore,
\begin{align*}
    &\frac1n\sum_{g=1}^G
    \E\left[
    \left(
        \sum_{i\in\G_g}
        \frac1{\Perm{n-1}{k-1}}\sum_{\mathbf j\in\mathcal B_i\setminus\mathcal A_i}
        \varphi(W_i,W_{j_2},\ldots,W_{j_k})
        \right)^2
    \right] \\
    & \qquad \leq
    \frac1n\sum_{g=1}^G
    n_g\sum_{i\in\G_g}
    \E\left[
    \left(
        \frac1{\Perm{n-1}{k-1}}\sum_{\mathbf j\in\mathcal B_i\setminus\mathcal A_i}
        \varphi(W_i,W_{j_2},\ldots,W_{j_k})
    \right)^2
    \right]\\
    & \qquad \leq
    \frac{m_n}{n}\sum_{i=1}^n C\frac{m_n^2}{n^2}
    =
    C\frac{m_n^3}{n^2}
    \to0.
\end{align*}

\paragraph{Assembly.}
Combining the three component bounds in \eqref{eq:clus_plugin_gap_decomp} gives the claim.
\end{proof}

\subsection{HAC projection covariance estimator}\label{app:hac}
This proof broadly follows the same strategy as \citet{dehling2017testing} in the order-$2$ Kendall setting: establish consistency of the oracle HAC estimator built from the true first-order projection, then show that replacing that projection by its empirical plug-in version is asymptotically negligible. The argument extends that logic to general order-$k$ $U$-statistics under the paper's $L_2$-NED framework.

\subsubsection{Notation and master decomposition}

For a distinct tuple \(\mathbf i=(i_1,\ldots,i_k)\), define the aggregate higher-order residual
\[
    H_{\mathbf i}(W_{\mathbf i})
    :=
    \varphi(W_{i_1},\ldots,W_{i_k})
    -
    \sum_{a=1}^k\varphi_{\mathbf i_{-a}}^{(1)}(W_{i_a}),
\]
where \(\varphi_{\mathbf i_{-a}}^{(1)}(W_{i_a})\) is the tuple-specific first projection. For a generic tuple, \eqref{eq:hoeffding_reconstruction} and the \(S=\{a\}\) collapse noted above give
\[
    H_{\mathbf i}(W_{\mathbf i})
    =
    \sum_{\substack{S\subseteq[k]\\ |S|\geq2}}
    \varphi_{\mathbf i,S}^{(|S|)}(W_{\mathbf i_S}).
\]
At \(k=2\), this recovers the single degenerate kernel of the classical order-\(2\) Hoeffding decomposition.

Set
\[
    \varphi_i
    :=
    \varphi_i^{(1)}(W_i)
    :=
    \frac1{\Perm{n-1}{k-1}}
    \sum_{\substack{(i_2,\ldots,i_k)_{\neq}\\ i_a\in[n]\setminus\{i\},\,a=2,\ldots,k}}
    \E\!\left[
    \varphi(W_i,W_{i_2},\ldots,W_{i_k})
    \mid W_i
    \right],
\]
and set \(\bar\varphi_n^{(1)}:=n^{-1}\sum_{j=1}^n\varphi_j\),
\[
\begin{aligned}
    H_{n,k,i}&:=\frac1{\Perm{n-1}{k-1}}\sum_{\substack{(i_2,\ldots,i_k)_{\neq}\\ i_2,\ldots,i_k\in[n]\setminus\{i\}}}H_{(i,i_2,\ldots,i_k)}(W_i,W_{i_2},\ldots,W_{i_k}),\\
    \bar H_n&:=\frac1{\Perm nk}\sum_{(i_1,\ldots,i_k)_{\neq}}H_{\mathbf i}(W_{i_1},\ldots,W_{i_k}),
\end{aligned}
\]
and recall also that \(\hat\varphi^{(1)}_{n,k,i}=\tilde\varphi_{n,k,i}-U_{n,k}\), where \(\tilde\varphi_{n,k,i}\) denotes the anchored average over ordered co-index tuples.

\[
    \tilde\rho(\ell):=\frac1n\sum_{i=1}^{n-\ell}\varphi_i\varphi_{i+\ell},
    \qquad
    \tilde\sigma_n^2:=\tilde\rho(0)+2\sum_{\ell=1}^{n-1}\kappa(\ell/h_n)\tilde\rho(\ell)
\]
denote the HAC estimator built from the true, unobserved projection. The estimator-to-target comparison starts from
\[
    \hat\sigma_n^2-\sigma_n^2
    =
    \big(\hat\sigma_n^2-\tilde\sigma_n^2\big)
    +
    \big(\tilde\sigma_n^2-\sigma_n^2\big).
\]
The first term is the feasible-to-oracle gap; the second is the oracle term.

\subsubsection{Feasible-to-oracle gap}

\begin{lemma}[Plug-in error vanishes]\label{lem:plugin_vanish}
    Under the assumptions of \Cref{thm:CLT_double_limit}, with \eqref{eq:kernel_ui_k} imposed for some moment exponent \(r>2\), and under \Cref{ass:hac}, $\hat\sigma_n^2-\tilde\sigma_n^2=o_p(1)$.
\end{lemma}
\begin{proof}
    The proof has two moving parts. First, use the algebraic identity in \Cref{lem:phi1_hat_identity} to express the projection error as a small deterministic-centring component minus an anchored higher-order remainder. Second, control this anchored remainder and treat the only linear-in-\(H_{n,k,i}\) HAC term at the full-sample level.
    
    Write \(e_i:=\varphi_i-\hat\varphi^{(1)}_{n,k,i}\). By symmetry of \(\kappa\),
    \[
        \hat\sigma_n^2-\tilde\sigma_n^2 = \sum_{\ell=-(n-1)}^{n-1}\kappa\Big(\frac{|\ell|}{h_n}\Big)\frac1n\sum_{i=1}^{n-|\ell|}\Big[-\varphi_i e_{i+|\ell|}-\varphi_{i+|\ell|}e_i+e_ie_{i+|\ell|}\Big].
    \]
    By \Cref{lem:phi1_hat_identity},
    \[
        e_i=A_i-H_{n,k,i},
        \qquad
        A_i:=
        \bar\varphi_n^{(1)}+\bar H_n+\frac{k-1}{n-1}\big(\varphi_i-\bar\varphi_n^{(1)}\big).
    \]
    Substituting \(e_i=A_i-H_{n,k,i}\) and collecting one layer at a time gives
    {\small
    \begin{align}
        \hat\sigma_n^2-\tilde\sigma_n^2
        &=
        \underbrace{
        \sum_{\ell=-(n-1)}^{n-1}\kappa\Big(\frac{|\ell|}{h_n}\Big)\frac1n\sum_{i=1}^{n-|\ell|}
        \Big[
            -\varphi_i A_{i+|\ell|}
            -
            \varphi_{i+|\ell|}A_i
            +
            A_iA_{i+|\ell|}
        \Big]}_{(\mathrm A)}
        \notag\\
        &\quad+
        \sum_{\ell=-(n-1)}^{n-1}\kappa\Big(\frac{|\ell|}{h_n}\Big)\frac1n\sum_{i=1}^{n-|\ell|}
        \Big[
            \varphi_i H_{n,k,i+|\ell|}
            +
            \varphi_{i+|\ell|}H_{n,k,i}
            -
            A_iH_{n,k,i+|\ell|}
            -
            A_{i+|\ell|}H_{n,k,i}
            +
            H_{n,k,i}H_{n,k,i+|\ell|}
        \Big]
        \notag\\
        &=
        (\mathrm A)
        +
        \underbrace{
        \sum_{\ell=-(n-1)}^{n-1}\kappa\Big(\frac{|\ell|}{h_n}\Big)\frac1n\sum_{i=1}^{n-|\ell|}
        \Big[
            \varphi_i H_{n,k,i+|\ell|}
            +
            \varphi_{i+|\ell|}H_{n,k,i}
            -
            A_iH_{n,k,i+|\ell|}
            -
            A_{i+|\ell|}H_{n,k,i}
        \Big]}_{(\mathrm B)}
        \notag\\
        &\quad+
        \underbrace{
        \begin{aligned}
        &\sum_{\ell=-(n-1)}^{n-1}\kappa\Big(\frac{|\ell|}{h_n}\Big)
        \frac1n\sum_{i=1}^{n-|\ell|}\\
        &\hspace{2.5cm}\times
        H_{n,k,i}H_{n,k,i+|\ell|}
        \end{aligned}}_{(\mathrm C)}.
        \label{eq:hac_plugin_decomp}
    \end{align}
    }

    By \Cref{lem:hac_anchor_bound}, \eqref{eq:hac_H_bound} gives \(\|H_{n,k,i}\|_2=\bigO(n^{-\eta/(2\eta+1)})\) uniformly in \(i\), and \(\|\bar H_n\|_2=o(n^{-1/2})\). Hence \(\|\bar\varphi_n^{(1)}\|_2=\bigO(n^{-1/2})\) by the covariance summability used in \Cref{thm:CLT_double_limit}(b), \(\|\varphi_i\|_2\leq C\), and, for \(A_i=\bar\varphi_n^{(1)}+\bar H_n+\frac{k-1}{n-1}\{\varphi_i-\bar\varphi_n^{(1)}\}\), the triangle inequality gives \(\|A_i\|_2=\bigO(1)\) uniformly in \(i\). Moreover, the HAC weight conditions give \(\sum_{\ell=-(n-1)}^{n-1}|\kappa(|\ell|/h_n)|=\bigO(h_n)\): by evenness, this is bounded by a constant multiple of the scaled Riemann sum \(h_n\int_0^\infty|\kappa(x)|\,dx\). The following elementary bound is used repeatedly:
    \begin{equation}\label{eq:hac_kernel_product_bound}
        \E\Big|\frac1n\sum_{\ell=-(n-1)}^{n-1}\kappa\Big(\frac{|\ell|}{h_n}\Big)\sum_{i=1}^{n-|\ell|}X_{i,\ell}Y_{i,\ell}\Big|
        \leq
        \frac1n\sum_{\ell=-(n-1)}^{n-1}\Big|\kappa\Big(\frac{|\ell|}{h_n}\Big)\Big|\sum_{i=1}^{n-|\ell|}\|X_{i,\ell}\|_2\|Y_{i,\ell}\|_2.
    \end{equation}

\paragraph{Deterministic-centring terms.}
Write
    \[
        A_i=
        a_n\varphi_i+c_n,
        \qquad
        a_n:=\frac{k-1}{n-1},
        \qquad
        c_n:=
        \Big(1-\frac{k-1}{n-1}\Big)\bar\varphi_n^{(1)}+\bar H_n.
    \]
    Then \(\|c_n\|_2=\bigO(n^{-1/2})\), and the inner bracket in \((\mathrm A)\) becomes
    \(
        -\varphi_i A_{i+|\ell|}-\varphi_{i+|\ell|}A_i+A_iA_{i+|\ell|}
        =
        -(2a_n-a_n^2)\varphi_i\varphi_{i+|\ell|}
        -(1-a_n)c_n\big(\varphi_i+\varphi_{i+|\ell|}\big)
        +c_n^2.
    \)
    Thus every contribution to \((\mathrm A)\) contains either the small coefficient \(2a_n-a_n^2=\bigO(n^{-1})\) or at least one factor \(c_n\). Since \(\|c_n\|_2=\bigO(n^{-1/2})\), the terms containing \(c_n\) are \(\bigOp(h_n/\sqrt n)\) by \eqref{eq:hac_kernel_product_bound} and Markov's inequality. The local-product term \(\varphi_i\varphi_{i+|\ell|}\) is multiplied by \(2a_n-a_n^2=\bigO(n^{-1})\), so it is \(\bigOp(h_n/n)\). Therefore
    \[
        (\mathrm A)=\bigOp\Big(\frac{h_n}{\sqrt n}\Big)=o_p(1).
    \]

\paragraph{Terms with one anchored remainder.}
The bracket in \((\mathrm B)\) equals
\[
    (1-a_n)\{\varphi_iH_{n,k,i+|\ell|}+\varphi_{i+|\ell|}H_{n,k,i}\}
    -c_n\{H_{n,k,i}+H_{n,k,i+|\ell|}\}.
\]
The \(c_n\)-part is negligible by \eqref{eq:hac_kernel_product_bound}, \(\|c_n\|_2=\bigO(n^{-1/2})\), and \eqref{eq:hac_H_bound}; it is \(\bigOp(h_nn^{-1/2-\eta/(2\eta+1)})=o_p(1)\). The remaining two terms are \(o_p(1)\) by \Cref{lem:hac_linear_H}. Therefore
    \[
        (\mathrm B)=o_p(1).
    \]

\paragraph{Terms with two anchored remainders.}
For \((\mathrm C)\), each summand carries two factors of the form \(H_{n,k,\cdot}\), so \eqref{eq:hac_kernel_product_bound} and \eqref{eq:hac_H_bound}, followed by Markov's inequality, give
    \[
        (\mathrm C)=\bigOp\big(h_nn^{-2\eta/(2\eta+1)}\big)=o_p(1),
    \]
because \(h_n=o(\sqrt n)\) and \(2\eta/(2\eta+1)>1/2\).

\paragraph{Assembly.}
Therefore \((\mathrm A)+(\mathrm B)+(\mathrm C)=o_p(1)\). The decomposition in \eqref{eq:hac_plugin_decomp} proves the claim. \qedhere
\end{proof}

\subsubsection{The oracle term}

\begin{lemma}[Oracle consistency]\label{lem:oracle_consistency}
    Under the assumptions of \Cref{thm:CLT_double_limit}, with \eqref{eq:kernel_ui_k} imposed for some moment exponent \(r>2\), and under \Cref{ass:hac}, $\tilde\sigma_n^2-\sigma_n^2\xrightarrow{p}0$.
\end{lemma}
\begin{proof}
    Apply \citet[Thm.~2.1]{dejong2000consistency} to the scalar triangular array \(X_{n,t}:=n^{-1/2}\varphi_t^{(1)}(W_t)\). With this choice, the HAC covariance estimator in that theorem is exactly \(\tilde\sigma_n^2\): the diagonal terms give \(\tilde\rho(0)\), and the off-diagonal terms combine, using \(\kappa(x)=\kappa(-x)\), into \(2\sum_{\ell=1}^{n-1}\kappa(\ell/h_n)\tilde\rho(\ell)\).
    
    The HAC-weight regularity and bandwidth requirements are covered by \Cref{ass:hac}. The dependence requirement is also inherited from the present assumptions. The base sequence \(\{\mathcal E_t\}\) is i.i.d., hence strongly mixing with zero mixing coefficients at every positive lag, and \Cref{lem:effective_perturbation}(b) gives, for every approximation width \(m\),
    \[
        \|X_{n,t}-\E[X_{n,t}\mid\mathcal F_{t-\floor{m/2}}^{t+\ceil{m/2}}]\|_2
        \leq
        \frac{C}{\sqrt n}\rho_m .
    \]
    Since the theorem assumes \(\rho_m=\bigO(m^{-\eta})\) for some \(\eta>1\), this verifies the required \(L_2\)-NED rate for the scaled projection array.
    
    It remains to check domination. This is the only point at which the HAC argument needs a moment strictly above two. Conditional Jensen and the original-kernel part of \eqref{eq:kernel_ui_k} imply \(C_\varphi:=\sup_t\|\varphi_t^{(1)}(W_t)\|_r<\infty\) for some \(r>2\). Hence \(\|X_{n,t}\|_r\leq C_\varphi/\sqrt n\) uniformly in \(t\), and the NED approximation bound above has the same \(n^{-1/2}\) scale. Thus the triangular-array domination required by \citet[Thm.~2.1]{dejong2000consistency} holds with a scale proportional to \(n^{-1/2}\): the maximal scale vanishes and the sum of squared scales is bounded. The bandwidth restriction reduces to \(h_n^{-1}+h_n/n\to0\), which is implied by \Cref{ass:hac}, since \(h_n\to\infty\) and \(h_n=o(\sqrt n)\).
    
    Therefore de Jong and Davidson's theorem gives
    \[
        \tilde\sigma_n^2
        -
        \V\left(\frac1{\sqrt n}\sum_{t=1}^n\varphi_t^{(1)}(W_t)\right)
        \xrightarrow{p}0.
    \]
    The second term is \(\sigma_n^2\) by definition. Hence \(\tilde\sigma_n^2-\sigma_n^2\xrightarrow{p}0\). \qedhere
\end{proof}

\subsubsection[Proof of HAC consistency]{Proof of \cref{thm:hac_consistency}}
Consider the decomposition
\[
    \hat\sigma_n^2-\sigma_n^2
    =
    \big(\hat\sigma_n^2-\tilde\sigma_n^2\big)
    +
    \big(\tilde\sigma_n^2-\sigma_n^2\big).
\]
The first term is \(o_p(1)\) by \Cref{lem:plugin_vanish}, and the second term is \(o_p(1)\) by \Cref{lem:oracle_consistency}. Hence \(\hat\sigma_n^2-\sigma_n^2\xrightarrow{p}0\). \qed

\subsubsection{Auxiliary lemmata}

The next three lemmata provide the algebraic identity and anchored-remainder controls used in \Cref{lem:plugin_vanish}.

\begin{lemma}[Plug-in identity]\label{lem:phi1_hat_identity}
    For every $n>k$ and $i=1,\ldots,n$,
    \[
        \varphi_i-\hat\varphi^{(1)}_{n,k,i} = \bar\varphi_n^{(1)}-H_{n,k,i}+\bar H_n+\frac{k-1}{n-1}\big(\varphi_i-\bar\varphi_n^{(1)}\big).
    \]
\end{lemma}
\begin{proof}
    Start from
    \[
        \hat\varphi^{(1)}_{n,k,i}-\varphi_i
        =
        \underbrace{\{\tilde\varphi_{n,k,i}-\varphi_i\}}_{A_{n,k,i}}
        -
        U_{n,k}.
    \]
    The anchored term satisfies
    \begin{align*}
        A_{n,k,i}
        &=
        \frac1{\Perm{n-1}{k-1}}
        \sum_{\substack{(i_2,\ldots,i_k)_{\neq}\\ i_2,\ldots,i_k\in[n]\setminus\{i\}}}
        \varphi(W_i,W_{i_2},\ldots,W_{i_k})-\varphi_i\\
        &\overset{(A)}{=}
        \frac1{\Perm{n-1}{k-1}}
        \sum_{\substack{(i_2,\ldots,i_k)_{\neq}\\ i_2,\ldots,i_k\in[n]\setminus\{i\}}}
        \left\{
        H_{(i,i_2,\ldots,i_k)}(W_i,W_{i_2},\ldots,W_{i_k})
        +
        \sum_{a=2}^k\varphi_{(i,i_2,\ldots,i_{a-1},i_{a+1},\ldots,i_k)}^{(1)}(W_{i_a})
        \right\}\\
        &\overset{(B)}{=}
        H_{n,k,i}
        +
        \frac1{\Perm{n-1}{k-1}}
        \sum_{a=2}^k
        \sum_{\substack{(i_2,\ldots,i_k)_{\neq}\\ i_2,\ldots,i_k\in[n]\setminus\{i\}}}
        \varphi_{(i,i_2,\ldots,i_{a-1},i_{a+1},\ldots,i_k)}^{(1)}(W_{i_a})\\
        &\overset{(C)}{=}
        H_{n,k,i}
        +
        \frac{\Perm{n-2}{k-2}}{\Perm{n-1}{k-1}}
        (k-1)
        \sum_{j\neq i}
        \varphi_j\\
        &=
        H_{n,k,i}
        +
        \frac{k-1}{n-1}
        \sum_{j\neq i}
        \varphi_j .
    \end{align*}

    Equality (A) follows from the definition of the tuple-specific residual \(H_{\mathbf i}\). After averaging over co-index tuples, the first-coordinate projections contribute \(\varphi_i\) and cancel with the outside \(-\varphi_i\). Equality (B) uses the definition of \(H_{n,k,i}\). Equality (C) follows because each point \(j\neq i\) occupies each co-index position in exactly \(\Perm{n-2}{k-2}\) ordered tuples. Hence the corresponding tuple-specific projections average to \(\varphi_j\), and \(\Perm{n-2}{k-2}/\Perm{n-1}{k-1}=1/(n-1)\).

    The statistic itself decomposes as
    \[
        U_{n,k}=k\bar\varphi_n^{(1)}+\bar H_n.
    \]
    The latter identity follows by the same position-counting argument applied to all \(k\) positions of an unrestricted distinct tuple: each of the \(n\) points occupies each position in \(\Perm{n-1}{k-1}\) of the \(\Perm nk\) ordered tuples, and \(\Perm{n-1}{k-1}/\Perm nk=1/n\).
    Consequently,
    \begin{align*}
        \hat\varphi^{(1)}_{n,k,i}-\varphi_i
        &=
        \frac{k-1}{n-1}\sum_{j\neq i}\varphi_j
        -k\bar\varphi_n^{(1)}
        +H_{n,k,i}-\bar H_n\\
        &=
        -\bar\varphi_n^{(1)}
        -
        \frac{k-1}{n-1}\big(\varphi_i-\bar\varphi_n^{(1)}\big)
        +
        H_{n,k,i}-\bar H_n.
    \end{align*}
    Multiplying by \(-1\) gives the claim. \qedhere
\end{proof}

\begin{lemma}[Anchored HAC remainder bound]\label{lem:hac_anchor_bound}
    Under the assumptions of \Cref{thm:CLT_double_limit}, with \eqref{eq:kernel_ui_k} imposed for some moment exponent \(r>2\),
    \begin{equation}\label{eq:hac_H_bound}
        \|H_{n,k,i}\|_2=\bigO\big(n^{-\eta/(2\eta+1)}\big)\quad\text{uniformly in }i,
    \end{equation}
    and \(\|\bar H_n\|_2=o(n^{-1/2})\).
\end{lemma}
\begin{proof}
    Fix an auxiliary truncation level \(l=l_n\to\infty\). This \(l\) is only a proof device; it does not appear in \(\hat\sigma_n^2\). Let \(H_{n,k,i}^{(l)}\) denote the same anchored average with \(W_i,W_{i_2},\ldots,W_{i_k}\) replaced by \(W_i^{(l)},W_{i_2}^{(l)},\ldots,W_{i_k}^{(l)}\). First peel off the truncation error:
    \[
        H_{n,k,i}
        =
        \underbrace{(H_{n,k,i}-H_{n,k,i}^{(l)})}_{\widetilde{\mathcal T}_{n,k,i}^{(l)}}
        +
        H_{n,k,i}^{(l)}.
    \]
    Next split the truncated anchored average by the same lag-collision/lag-generic distinction used in \eqref{eq:mdep_orderk_master_decomp}:
    \[
        H_{n,k,i}
        =
        \widetilde{\mathcal T}_{n,k,i}^{(l)}
        +
        \underbrace{\widetilde{\mathcal H}_{n,k,i}^{(l)}}_{\text{lag-collision}}
        +
        \underbrace{\widetilde{\mathcal R}_{n,k,i}^{(l)}}_{\text{lag-generic}},
    \]
    where \(H_{n,k,i}^{(l)}=\widetilde{\mathcal H}_{n,k,i}^{(l)}+\widetilde{\mathcal R}_{n,k,i}^{(l)}\).
    
    \emph{Truncation error.} Since \(H_{\mathbf i}=\varphi-\sum_{a=1}^k\varphi_{\mathbf i_{-a}}^{(1)}\), the coordinatewise and projection perturbation bounds in \Cref{lem:effective_perturbation}(a)--(b) give
    \[
        \big\|H_{\mathbf i}(W_{i_1},\ldots,W_{i_k})-H_{\mathbf i}^{(l)}(W_{i_1}^{(l)},\ldots,W_{i_k}^{(l)})\big\|_2=\bigO(\rho_l)
    \]
    uniformly over tuples. Thus, by Minkowski's inequality,
    {\small
    \[
        \begin{aligned}
        \big\|\widetilde{\mathcal T}_{n,k,i}^{(l)}\big\|_2
        &=
        \left\|
        \frac1{\Perm{n-1}{k-1}}
        \sum_{\substack{(i_2,\ldots,i_k)_{\neq}\\ i_2,\ldots,i_k\in[n]\setminus\{i\}}}
        \begin{aligned}[t]
        &\Big\{H_{(i,i_2,\ldots,i_k)}(W_i,W_{i_2},\ldots,W_{i_k})\\
        &\qquad
        -H_{(i,i_2,\ldots,i_k)}^{(l)}(W_i^{(l)},W_{i_2}^{(l)},\ldots,W_{i_k}^{(l)})\Big\}
        \end{aligned}
        \right\|_2\\
        &\leq
        \frac1{\Perm{n-1}{k-1}}
        \sum_{\substack{(i_2,\ldots,i_k)_{\neq}\\ i_2,\ldots,i_k\in[n]\setminus\{i\}}}
        \Big\|
        H_{(i,i_2,\ldots,i_k)}(W_i,W_{i_2},\ldots,W_{i_k})
        -
        H_{(i,i_2,\ldots,i_k)}^{(l)}(W_i^{(l)},W_{i_2}^{(l)},\ldots,W_{i_k}^{(l)})
        \Big\|_2\\
        &=
        \bigO(\rho_l),
        \end{aligned}
    \]
    }
    uniformly in \(i\).

    \emph{Lag-collision tuples.} In \(\widetilde{\mathcal H}_{n,k,i}^{(l)}\), some pair among \(\{i,i_2,\ldots,i_k\}\) lies within lag \(l\). The anchored version of the linked-collision count in \Cref{lem:abstract_linked_remainder_controls}, applied under the temporal dictionary, gives \(\bigO(ln^{k-2})\) such co-index tuples. Since \eqref{eq:kernel_ui_k}, applied to the truncated kernels, and \eqref{eq:abstract_hoeffding_reconstruction} give a uniform second-moment bound for \(H_{\mathbf i}\) on truncated arguments, Minkowski's inequality gives \(\|\widetilde{\mathcal H}_{n,k,i}^{(l)}\|_2=\bigO(l/n)\).

    \emph{Lag-generic tuples.} In \(\widetilde{\mathcal R}_{n,k,i}^{(l)}\), the \(l\)-truncated variables in each tuple are mutually independent by \Cref{lem:approx_mdep}. Let \(\varphi_{\mathbf i_{-1}}^{(1,l)}(W_i^{(l)})\) denote the tuple-specific first projection intrinsic to the truncated sequence. Mutual independence gives
    \[
        \E[H_{(i,i_2,\ldots,i_k)}^{(l)}(W_i^{(l)},\ldots,W_{i_k}^{(l)})\mid W_i^{(l)}]
        =
        \varphi_{\mathbf i_{-1}}^{(1,l)}(W_i^{(l)})
        -
        \varphi_{\mathbf i_{-1}}^{(1)}(W_i^{(l)})
        -
        \sum_{a=2}^k\E[\varphi_{\mathbf i_{-a}}^{(1)}(W_{i_a}^{(l)})].
    \]
    The co-index replacement bound in \Cref{lem:effective_perturbation}(c) gives
    \[
        \|\varphi_{\mathbf i_{-1}}^{(1,l)}(W_i^{(l)})
        -\varphi_{\mathbf i_{-1}}^{(1)}(W_i^{(l)})\|_2
        =
        \bigO(\rho_l).
    \]
    Also, since the original tuple-specific projections are centred, \Cref{lem:effective_perturbation}(b) gives \(|\E[\varphi_{\mathbf i_{-a}}^{(1)}(W_{i_a}^{(l)})]|=\bigO(\rho_l)\), uniformly over \(a=2,\ldots,k\). Hence
    \[
        \widetilde{\mathcal R}_{n,k,i}^{(l)}
        =
        \underbrace{\Big\{\widetilde{\mathcal R}_{n,k,i}^{(l)}
        -
        \E[\widetilde{\mathcal R}_{n,k,i}^{(l)}\mid W_i^{(l)}]\Big\}}_{\widetilde{\mathcal R}_{n,k,i}^{(l),0}}
        +
        \bigO_{L_2}(\rho_l).
    \]
    The centred term \(\widetilde{\mathcal R}_{n,k,i}^{(l),0}\) is the anchored counterpart of the exact-\(m\) higher-order remainder \(\mathcal R_{n,m,k}\) in \eqref{eq:mdep_orderk_master_decomp}. After the anchor \(i\) is fixed, the same degeneracy-and-linking argument, namely \Cref{lem:abstract_linked_degenerate_variance} aggregated through \Cref{lem:abstract_linked_remainder_controls} under the temporal dictionary, applies to the remaining co-index tuples. The anchored count has \(\bigO(l n^{2k-3})\) potentially non-zero covariance pairs before normalisation; division by \(\Perm{n-1}{k-1}^2\asymp n^{2k-2}\) gives
    \[
        \E\big[(\widetilde{\mathcal R}_{n,k,i}^{(l),0})^2\big]
        =
        \bigO(l/n).
    \]
    Together with the \(\bigO(\rho_l)\) conditional-mean bias just obtained, this gives
    \[
        \|\widetilde{\mathcal R}_{n,k,i}^{(l)}\|_2
        =
        \bigO(\rho_l)+\bigO(\sqrt{l/n}).
    \]

    Collecting the truncation, collision, and generic pieces,
    \[
        \|H_{n,k,i}\|_2
        =
        \bigO(\rho_l)+\bigO(l/n)+\bigO(\sqrt{l/n})
    \]
    for every \(i\) and every \(l=l_n\). The rate condition in \Cref{thm:CLT_double_limit} gives \(\rho_l=\bigO(l^{-\eta})\) for some \(\eta>1\). Choosing \(l_n\sim n^{1/(2\eta+1)}\) balances \(\rho_l\) and \(\sqrt{l/n}\), and gives \eqref{eq:hac_H_bound}.

    At the statistic level, \(\bar H_n=U_{n,k}-k\bar\varphi_n^{(1)}\). The bounds in the proof of \Cref{thm:CLT_double_limit}(a), namely the truncation, exact-\(m\) remainder, centring, and projection-transfer bounds above, give \(\|\bar H_n\|_2=o(n^{-1/2})\). \qedhere
\end{proof}

\begin{lemma}[HAC-linear anchored remainder]\label{lem:hac_linear_H}
    Under the assumptions of \Cref{thm:CLT_double_limit} and \Cref{ass:hac},
    \[
        \sum_{\ell=-(n-1)}^{n-1}\kappa\Big(\frac{|\ell|}{h_n}\Big)\frac1n\sum_{i=1}^{n-|\ell|}\varphi_iH_{n,k,i+|\ell|}=o_p(1),
    \]
    and the same conclusion holds with \(\varphi_{i+|\ell|}H_{n,k,i}\) in place of \(\varphi_iH_{n,k,i+|\ell|}\).
\end{lemma}
\begin{proof}
    Consider the displayed term; the other one is identical by symmetry. After expanding the anchored average, write it as
    {\small
    \[
        S_n:=
        \frac1n
        \frac1{\Perm{n-1}{k-1}}
        \sum_{r=-(n-1)}^{n-1}
        \sum_{i=1}^{n-|r|}
        \kappa\Big(\frac{|r|}{h_n}\Big)
        \sum_{\substack{(i_2,\ldots,i_k)_{\neq}\\ i_2,\ldots,i_k\in[n]\setminus\{i+|r|\}}}
        \varphi_i\,
        H_{(i+|r|,i_2,\ldots,i_k)}(W_{i+|r|},W_{i_2},\ldots,W_{i_k}).
    \]
    }
    Thus \(S_n\) is a weighted anchored full-sample average over the lag-paired anchor block \(A_{i,r}:=\{i,i+|r|\}\) and the \(k-1\) co-indices \((i_2,\ldots,i_k)\).

    Fix an auxiliary truncation level \(q=q_n\). Let \(\varphi_i^{(q)}\) be the corresponding first projection evaluated at \(W_i^{(q)}\), and let \(H_{\mathbf j}^{(q)}\) be the residual \(H_{\mathbf j}\) with all arguments replaced by their \(q\)-window approximants and with the corresponding truncated projections. Define
    {\small
    \[
        S_n^{(q)}:=
        \frac1n
        \frac1{\Perm{n-1}{k-1}}
        \sum_{r=-(n-1)}^{n-1}
        \sum_{i=1}^{n-|r|}
        \kappa\Big(\frac{|r|}{h_n}\Big)
        \sum_{\substack{(i_2,\ldots,i_k)_{\neq}\\ i_2,\ldots,i_k\in[n]\setminus\{i+|r|\}}}
        \varphi_i^{(q)}
        H_{(i+|r|,i_2,\ldots,i_k)}^{(q)}(W_{i+|r|}^{(q)},W_{i_2}^{(q)},\ldots,W_{i_k}^{(q)}).
    \]
    }
    By \Cref{lem:effective_perturbation}(a)--(b), conditional Jensen's inequality, and \(\sum_\ell|\kappa(|\ell|/h_n)|=\bigO(h_n)\),
    \[
        S_n-S_n^{(q)}=\bigOp(h_n\rho_q).
    \]

    Split \(S_n^{(q)}\) according to lag relations among the co-indices relative to the fixed anchor block \(A_{i,r}\):
    \[
        S_n^{(q)}
        =
        \mathcal H_n^{(q)}
        +
        \mathcal R_n^{(q)},
    \]
    where \(\mathcal H_n^{(q)}\) collects terms for which some co-index is within lag \(q\) of an element of \(A_{i,r}\), or some pair of co-indices is within lag \(q\), and \(\mathcal R_n^{(q)}\) collects the centred terms for which the co-indices are lag-generic relative to the anchor block and to one another. The possible dependence between the two elements of \(A_{i,r}\), including the case \(|r|\leq q\), is therefore kept inside the fixed anchor block rather than counted as a collision. This is the anchored HAC analogue of the lag-collision/higher-order remainder split in \eqref{eq:mdep_orderk_master_decomp}. The collision tuples contribute
    \[
        \mathcal H_n^{(q)}=\bigOp(h_nq/n),
    \]
    by the pair-window count in \Cref{lem:abstract_linked_remainder_controls}, applied under the temporal dictionary after conditioning on the fixed anchor block \(A_{i,r}\). For each \(r\), the anchor block rules out \(\bigO(qn^{k-2})\) co-index tuples, and the HAC weights contribute the factor \(\sum_r|\kappa(|r|/h_n)|=\bigO(h_n)\).

    It remains to control \(\mathcal R_n^{(q)}\). On this lag-generic part, the \(q\)-window co-index variables are mutually independent conditional on the anchor block, and the residual is degenerate after centring, exactly as in the exact-\(m\) decomposition. For fixed \(r\), the anchor block only adds a finite-range conditioning object to the exact-\(m\) counting problem. The degeneracy-and-linking argument in \Cref{lem:abstract_linked_degenerate_variance}, aggregated through \Cref{lem:abstract_linked_remainder_controls} under the temporal dictionary, therefore gives variance \(\bigO(q/n)\) for the lag-\(r\) average. Jointly over lags, the same count is applied to the combined index object \((r,i,i_2,\ldots,i_k)\): non-zero covariance requires either effectively linked anchor blocks or linked co-index configurations, and the HAC weights restrict the lag summation to \(\bigO(h_n)\) effective values. Hence the joint variance is \(\bigO(h_nq/n)\), and
    \[
        \mathcal R_n^{(q)}=\bigOp(\sqrt{h_nq/n}).
    \]
    Thus
    \[
        S_n
        =
        \bigOp(h_n\rho_q)+\bigOp(h_nq/n)+\bigOp(\sqrt{h_nq/n}).
    \]
    Choose \(q_n\sim(h_nn)^{1/(2\eta+1)}\), which balances \(h_n\rho_q=\bigO(h_nq^{-\eta})\) and \(\sqrt{h_nq/n}\). The balanced rates are \(h_n^{(\eta+1)/(2\eta+1)}n^{-\eta/(2\eta+1)}=o(1)\), because \(h_n=o(\sqrt n)\) and \(\eta>1\). Also \(h_nq_n/n=o(1)\), proving the claim. \qedhere
\end{proof}

\printbibliography
\end{refsection}
\end{document}